\documentclass{lmcs}
\pdfoutput=1
\usepackage[utf8]{inputenc}

\usepackage{lastpage}
\lmcsdoi{19}{1}{17}
\lmcsheading{}{\pageref{LastPage}}{}{}%
{Dec.~17,~2021}{Mar.~06,~2023}{}

\keywords{Failure detectors, \tlap{}, counter automata, \fast{}, and IVy.}

\usepackage[colorlinks]{hyperref}
\usepackage{color}
\usepackage{tlatex}
\usepackage{algorithm}
\usepackage[noend]{algpseudocode}
\usepackage{algorithmicx}
\usepackage[english]{babel}
\usepackage{amsmath}
\usepackage{amssymb}
\usepackage{relsize}
\usepackage{tikz}
\usepackage{subcaption}
\usepackage{mathtools}

\usetikzlibrary{matrix,backgrounds}

\usepackage{enumitem}   

\usepackage{textcomp}

\allowdisplaybreaks

\usepackage{listings}

\usepackage{fancybox}


\makeatletter
\newcommand*{\coloneq}{\mathrel{\rlap{%
			\raisebox{0.3ex}{$\m@th\cdot$}}%
		\raisebox{-0.3ex}{$\m@th\cdot$}}%
	=}
\makeatother

\usepackage{multicol}
\usepackage{array}
\usepackage{multirow}
\usepackage{xcolor}

\usepackage{marginnote}

\newcommand{\tlap}[0]{$\textsc{TLA}^+$}
\newcommand{\tlaps}[0]{\textsc{TLAPS}}

\newcommand{\apalache}[0]{\textsc{APALACHE}}%

\newcommand{\fast}[0]{\textsc{FAST}}

\DeclareMathOperator{\F}{\mathbf{F}}
\DeclareMathOperator{\G}{\mathbf{G}}
\DeclareMathOperator{\All}{\mathbf{A}}
\DeclareMathOperator{\X}{\mathbf{X}}
\DeclareMathOperator{\U}{\mathbf{U}}

\newcommand{\nid}{1..N}
\newcommand{\fdsuspected}[2]{\mathit{suspected}\left[ #1, #2 \right]}
\newcommand{\fdtimeout}[2]{\mathop{\mbox{$\mathit{timeout}$}}\left[ #1 , #2 \right]}

\newcommand{\gsys}[1]{\mathcal{G}_{#1}}

\newcommand{\existsMsgOfAge}{\textsf{existsMsgOfAge}}
\newcommand{\tnumber}{\textsf{num}}
\newcommand{\tmpExistsMsgOfAge}{\textsf{tmpExistsMsgOfAge}}
\newcommand{\shiftLeft}[1]{\textsf{shiftLeft}}
\newcommand{\msgBufAfterDelivery}[1]{\textsf{msgBufAfterDelivery}}
\newcommand{\incMsgAge}[1]{\textsf{incMsgAge}}
\newcommand{\sTimer}[1]{\textsf{sTimer}}
\newcommand{\rTimer}[1]{\textsf{rTimer}}

\newcommand{\timeout}{\textsf{timeout}}

\newcommand{\howLongFromSCrash}{\textsf{hLFSC}}

\newcommand{\ltlx}{LTL\textbackslash{}X}

\newcommand{\gconf}[1]{\kappa_{#1}}
\newcommand{\gnextconf}[1]{\kappa^{\prime}_{#1}}

\newcommand{\comp}{\mathit{comp}}
\newcommand{\crash}{\mathit{crash}}

\newcommand{\loc}{\mathit{Loc}}
\newcommand{\locsnd}{\mathit{Loc}_{\mathit{snd}}}
\newcommand{\locrcv}{\mathit{Loc}_{\mathit{rcv}}}
\newcommand{\loccomp}{\mathit{Loc}_{\mathit{comp}}}
\newcommand{\loccrash}{\ell_{\mathit{crash}}}

\newcommand{\tempconfs}[1]{Q_{#1}}
\newcommand{\temptr}[1]{\mathit{Tr}_{#1}}
\newcommand{\temprel}[1]{\mathit{Rel}_{#1}}
\newcommand{\tempinitconf}[1]{\mathit{q}_{#1}^0}

\newcommand{\temptrans}[1]{\xrightarrow{\mathit{tr}_{#1}}}
\newcommand{\tempsnd}[2]{\xrightarrow{\mathit{csnd}_{#1}(#2)}}
\newcommand{\temprcv}[3]{\xrightarrow{\mathit{crcv}_{#1}(#2, \ldots, #3)}}
\newcommand{\tempcomp}[1]{\xrightarrow{\mathit{comp}_{#1}}}
\newcommand{\tempcrash}[1]{\xrightarrow{\mathit{crash}_{#1}}}
\newcommand{\tempstutter}{\xrightarrow{\stutter}}
\newcommand{\tempdatatype}{\mathcal{D}}
\newcommand{\tempconf}[1]{\rho_{#1}}
\newcommand{\tempnextconf}[1]{\rho^{\prime}_{#1}}

\newcommand{\pcFncName}{\mathit{pc}}
\newcommand{\getPC}[1]{\pcFncName(#1)}
\newcommand{\boxFncName}{\mathit{rcvd}}
\newcommand{\getBox}[2]{\boxFncName(#1, #2)}

\newcommand{\datumFncName}{\mathit{lvar}}
\newcommand{\getDatum}[2]{\datumFncName(#1, #2)}

\newcommand{\procFncName}{\mathit{lstate}}
\newcommand{\getProc}[2]{\procFncName(#1, #2)}
\newcommand{\bufFncName}{\mathit{buf}}
\newcommand{\getBuf}[3]{\bufFncName(#1, #2, #3)}
\newcommand{\activeFncName}{\mathit{active}}
\newcommand{\getActive}[2]{\activeFncName(#1, #2)}

\newcommand{\stutter}{\mathit{stutter}}

\newcommand{\Boolean}{\texttt{Bool}}
\newcommand{\False}{\bot}
\newcommand{\True}{\top}

\newcommand{\classname}{symmetric point--to--point}
\newcommand{\tid}{1..2}

\newcommand{\recalltheorem}[2]{%
	{\medskip\noindent\bfseries Theorem~\ref{#1}.~}{\itshape #2}
}

\newcommand{\recalllemma}[2]{%
	{\medskip\noindent\bfseries Lemma~\ref{#1}.~}{\itshape #2}
}

\newcommand{\qdot}{\mathbin{.}}

\newtheorem{construction}[thm]{Construction}
\newcommand{\msg}{\texttt{Msg}}
\newcommand{\set}{\texttt{Set}}

\newcommand{\csnd}{\mathit{csnd}}
\newcommand{\crcv}{\mathit{crcv}}

\newcommand{\nextLoc}{\mathit{nextLoc}}
\newcommand{\genMsg}{\mathit{genMsg}}
\newcommand{\nextVar}{\mathit{nextVal}}

\newcommand{\gconfs}[1]{\mathcal{C}_{#1}}
\newcommand{\ginitconf}[1]{\mathit{g}_{#1}^0}
\newcommand{\gtr}[1]{Tr_{#1}}
\newcommand{\grel}[1]{R_{#1}}

\newcommand{\transsched}{\xrightarrow{\texttt{Sched}}}
\newcommand{\transsnd}{\xrightarrow{\texttt{Snd}}}
\newcommand{\transcomp}{\xrightarrow{\texttt{Comp}}}
\newcommand{\transrcv}{\xrightarrow{\texttt{Rcv}}}
\newcommand{\transstutter}{\xrightarrow{\stutter}}

\newcommand{\CanRun}{\mathit{Enabled}}
\newcommand{\SFrozen}{\mathit{Frozen}_S}
\newcommand{\RFrozen}{\mathit{Frozen}_R}

\newcommand{\transposition}{\alpha}

\newcommand{\predsnd}{\mathit{Sending}}

\newcommand{\predrcv}{\mathit{Receiving}}

\newcommand{\Proc}[1]{\Pi_{#1}}

\newcommand{\switchTempConf}[2]{#2_Q(#1)}
\newcommand{\switchGlobalConf}[2]{#2_{C}(#1)}
\newcommand{\switchIndex}[2]{#1(#2)}
\newcommand{\switchRcv}[3]{\xrightarrow{#3_R(\crcv(#1, \ldots, #2))}}
\newcommand{\switchGlobalSys}[2]{#1_{G}(#2)}
\newcommand{\switchTemp}[2]{#1_{U}(#2)}

\newcommand{\swapij}[2]{(#1 \leftrightarrow #2)}

\newcommand{\lab}[1][\ell]{#1\,::}

\newcolumntype{M}{>{\hspace{0.35cm}$}c<{$\hspace{0.35cm}}}

\mathchardef\mhyphen="2D % Define a "math hyphen"

\newcommand{\mytemp}[1]{\mathcal{U}_{#1}}

\newcommand{\tladot}[1]{\ensuremath{\mbox{}#1\mbox{}}}

\begin{document}
%
%\title{Contribution Title\thanks{Supported by organization x.}}
\title{A case study on parametric verification \linebreak of failure detectors}
\author[T.H.~Tran]{Thanh-Hai Tran\lmcsorcid{0000-0002-2260-6877}}[a, b]
\author[I.~Konnov]{Igor Konnov\lmcsorcid{0000-0001-6629-3377}}[c]
\author[J.~Widder]{Josef Widder\lmcsorcid{0000-0003-2795-611X}}[c]
	%required
  
\address{TU Wien, Austria}
\email{thanh.tran@tuwien.ac.at}

\address{ConsenSys, Australia}	%required
\email{thanh-hai.tran@consensys.net}  %optional

\address{Informal Systems, Austria}	%optional
\email{igor@informal.systems, josef@informal.systems}  %optional

\thanks{This work was supported by Interchain Foundation (Switzerland) and the Austrian Science Fund (FWF) via the Doctoral College LogiCS W1255.
This work was done when the first author was at TU Wien. 
This paper is an extended version of papers~\cite{tran2020netys} and~\cite{tran2021case} that adds the formalization of the model of computation under partial synchrony and detailed proofs.
}

\begin{abstract}
  
  Partial synchrony is a model of computation in many distributed algorithms and modern blockchains. 
These algorithms are typically parameterized in the number of participants, and
their correctness requires the existence of bounds on message delays and on the relative speed of processes after reaching Global Stabilization Time (GST). 
These characteristics make partially synchronous algorithms both parameterized and 
parametric, which render automated verification of partially synchronous algorithms challenging.
In this paper, we present a case study on formal verification of both safety and liveness of the Chandra and Toueg failure detector that is based on partial synchrony. 
To this end, we first introduce and formalize 
the class of symmetric point-to-point algorithms that contains the failure detector.
Second, we show that these symmetric point-to-point algorithms have a cutoff, and the cutoff results hold in three models of computation: synchrony, asynchrony, and partial synchrony.
As a result, one can verify them by model checking small instances, but the verification problem stays parametric in time.
%For example, we can verify the failure detector by checking only instances with two processes by applying these cutoff results.
Next, we specify the failure detector and the partial synchrony assumptions in three frameworks:~\tlap{}, IVy, and counter automata.
Importantly, we tune our modeling to use the strength of each method: (1) We are using counters to encode message buffers with counter automata, (2)~we are using first-order relations to encode message buffers in IVy, and (3) we are using both approaches in~\tlap{}. 
By running the tools for~\tlap{} (TLC and \apalache{}) and counter automata (\fast{}), we demonstrate safety for fixed time bounds.  
This helped us to find the inductive invariants for fixed parameters, which we used as a starting point for the proofs with IVy. 
By running IVy, we prove safety for arbitrary time bounds. 
Moreover, we show how to verify liveness of the failure detector by reducing the verification problem to safety verification.  
Thus, both properties are verified by developing inductive invariants with IVy. 
We conjecture that correctness of other partially synchronous algorithms may be proven by following the presented methodology.

\end{abstract}

\maketitle              % typeset the header of the contribution

\setcounter{footnote}{0}

\section{Introduction} \label{sec:intro}

Distributed algorithms play a crucial role in modern infrastructure, but they are notoriously difficult to understand and to get right.
Network topologies, message delays, faulty processes, the relative
speed of processes, and fairness conditions might 
lead to behaviors that were neither intended nor anticipated by  algorithm designers. 
To be able to make meaningful statements about correctness, many specification and verification techniques for distributed algorithms~\cite{lamport2002specifying,lynch1988introduction,mcmillan2020ivy,druagoi2020programming} have been developed.

Verification techniques for distributed algorithms usually focus on two models of computation: synchrony~\cite{stoilkovska2019verifying} 
and asynchrony~\cite{KLVW17:FMSD,KLVW17:POPL}.
Synchrony is hard to implement in real systems, while many basic 
problems in fault-tolerant distributed computing are unsolvable in asynchrony.

Partial synchrony lies between synchrony and asynchrony, and escapes their shortcomings.
To guarantee liveness properties, proof-of-stake 
blockchains~\cite{buchman2018latest,yin2019hotstuff} and distributed algorithms~\cite{CT96,bravo2020making} assume time constraints under partial synchrony. 
That is the existence of bounds $\Delta$ on message delay, and $\Phi$ on the relative speed of processes after some time point. 
This combination makes partially synchronous algorithms
parametric in time bounds.
While partial synchrony is important for system designers, it is challenging for verification.

We thus investigate verification of distributed algorithms
under partial synchrony, and start with the specific class of failure detectors: the Chandra and Toueg failure detector~\cite{CT96}.
This is a well-known algorithm under partial synchrony that provides a service to solve many problems in fault-tolerant distributed computing.

\paragraph{Contributions.}
In this paper, we do formal verification of both safety and liveness  of the Chandra and Toueg failure detector in case of unknown bounds $ \Delta$ and $\Phi$.
In this case, both $\Delta$ and $\Phi$ are arbitrary, and the constraints on message delay and the relative speeds hold in every execution from the start.
\begin{enumerate}
\item We introduce and formalize
the class of symmetric point-to-point algorithms that contains the failure detector.

\item We prove that the symmetric point-to-point algorithms have a cutoff, and the cutoff properties hold in three models of computation: synchrony, asyncrony, and partial synchrony.
In a nutshell, a cutoff for a parameterized algorithm $\mathcal{A}$ and a property $\phi$ is a number $k$ such that $\phi$ holds for every instance of $\mathcal{A}$ if and only if $\phi$ holds for instances with $k$ processes~\cite{emerson1995reasoning,2015Bloem}.
Our cutoff results with $k = 2$ were presented in~\cite{tran2020netys,tran2021case}.
Hence, we verify the Chandra and Toueg failure detector under partial synchrony by checking instances with two processes.

\item We introduce encoding techniques to efficiently specify the failure detector based on our cutoff results.
These techniques can tune our modeling to use the strength of the tools: \fast{}~\cite{BardinFLP08}, Ivy~\cite{mcmillan2020ivy}, and model checkers for~\tlap{}~\cite{yu1999model,konnov2019tla}.

\item We demonstrate how to reduce the liveness properties Eventually Strong Accuracy, and Strong Completeness to safety properties.

\item We check the safety property Strong Accuracy, and the mentioned liveness properties
on instances with fixed parameters by using~\fast{}, and model checkers for \tlap{}.

\item To verify cases of arbitrary bounds $\Delta$ and $\Phi$, we find and prove inductive invariants of the failure detector with the interactive theorem prover Ivy. 
We reduce the liveness properties to safety properties by applying the mentioned techniques.
While our specifications are not in the decidable theories that Ivy supports, Ivy requires no additional user assistance to prove most of our inductive invariants.

\end{enumerate}

\paragraph{Structure.} In Section~\ref{sec:preliminaries}, we discusses challenges in verification of the Chandra and Toueg failure detector.
In Section~\ref{sec:netys20-comp-model}, we introduce
the class of symmetric point-to-point algorithms, and present how to formalize this class. 
Our cutoff results in the asynchronous model are presented in Section~\ref{sec:netys20-cutoff-res}, and the detailed proofs are provided in Section~\ref{sec:netys20-proofs}.
We extend the cutoff results for partial synchrony in Section~\ref{sec:forte21-cutoff-res}.
Our encoding technique to efficiently specify the failure detector is presented in Section~\ref{sec:forte21-encoding}.
In Section~\ref{sec:forte21-reduction}, we present how to reduce the mentioned liveness properties to safety ones.
Experiments for small $\Delta$ and $\Phi$ are described in  Section~\ref{sec:forte21-experiments-fixed-paras}.
Ivy proofs for parametric $\Delta$ and $\Phi$ are discussed in
Section~\ref{sec:forte21-experiments-unknown-paras}.
Finally, we discuss related work in Section~\ref{sec:relate-work}.

\section{Challenges in verification of failure detectors} \label{sec:preliminaries}
%\section{Preliminaries} \label{sec:preliminaries}

\begin{algorithm}[tb] 
  \caption{The eventually perfect failure detector algorithm in \cite{CT96}} \label{algo:detector1}
  \begin{algorithmic}[1] 
    \State \emph{Every process $p \in \nid$ executes the following}:
    %%%%%%%%%%%%%%%%%%%%%%%%%%%%%%%%%%%%%%%%%%%%%%%%%%%%%%%%%%%%%%%%%%%%%%%%%%%%%%%%
    \For{$\textbf{all } q \in \nid $}    \Comment{Initalization step}
    \State $\fdtimeout{p}{q} \coloneq \text{default-value}$ \label{algo:initStep1}
    %; $\suspected{p}{q} = \False$ \}   %\label{algo:initStep2}
    \State $\fdsuspected{p}{q} \coloneq \False$ \             \label{algo:initStep2}
    \EndFor
    %    %%%%%%%%%%%%%%%%%%%%%%%%%%%%%%%%%%%%%%%%%%%%%%%%%%%%%%%%%%%%%%%%%%%%%%%%%%%%%%%%
    \State Send ``alive'' to all $q \in  \nid$ \label{algo:task1}
    \Comment{Task 1: repeat periodically}
    %    %%%%%%%%%%%%%%%%%%%%%%%%%%%%%%%%%%%%%%%%%%%%%%%%%%%%%%%%%%%%%%%%%%%%%%%%%%%%%%%%    
    \For{\textbf{all} $q  \in \nid$}  \label{algo:task2}    \Comment{Task 2: repeat periodically}
    \If{$\fdsuspected{p}{q} = \False \textbf{ and not hear } q \text{ during last } \fdtimeout{p}{q} \text{ ticks}$} 
    \State $\fdsuspected{p}{q} \coloneq \True$    
    \EndIf    
    \EndFor
    \If{$\fdsuspected{p}{q}$} \label{algo:task3} 
    \Comment{Task 3: when receive ``alive'' from $q$}
    \State $\fdtimeout{p}{q} \coloneq \fdtimeout{p}{q} + 1$ %; $\fdsuspected{p}{q} = \False$ \}
    \State $\fdsuspected{p}{q} \coloneq \False$ 
    \EndIf            
  \end{algorithmic}
\end{algorithm} 

The Chandra and Toueg failure detector~\cite{CT96} can be seen as an oracle to get information about crash failures in the distributed system. 
The failure detector usually guarantees some of the following properties~\cite{CT96} (numbers $\nid$ denote the process identifiers):
\begin{itemize} \label{properties}
  \item Strong Accuracy: No process is suspected before it crashes.	
  
  $
  \begin{aligned}
  \qquad \G (\A p, q \in \nid \colon (\mathit{Correct}(p) \land \mathit{Correct}(q)) \Rightarrow \lnot \mathit{Suspected}(p, q))
  \end{aligned}
  $
  
  \item Eventual Strong Accuracy: There is a time after which correct processes are not suspected by any correct process.
  
  $
  \begin{aligned}
  \qquad \F \G (\A p, q \in \nid \colon (\mathit{Correct}(p) \land \mathit{Correct}(q)) \Rightarrow \lnot \mathit{Suspected}(p, q))
  \end{aligned}
  $
  
  \item Strong Completeness: Eventually every crashed process is  permanently suspected by every correct process.
  
  $
  \begin{aligned}
  \qquad \F \G (\A p, q \in \nid \colon (\mathit{Correct}(p) \land \lnot\mathit{Correct}(q)) \Rightarrow \mathit{Suspected}(p, q))
  \end{aligned}
  $
\end{itemize}
where $\F$ and $\G$ are temporal operators in linear temporal logic (LTL)~\cite{pnueli1977temporal}~\footnote{A brief introduction of LTL is provided in Appendix~\ref{sec:ltl}.}, predicate $\mathit{Suspect}(p, q)$ refers to whether process $p$ suspects process $q$ to have crashed, and predicate $\mathit{Correct}(p)$ refers to whether process $p$ is correct.
Given an execution trace, process $p$ is correct if $\mathit{Correct}(p)$ is true for every time point~\footnote{
  We deviate from the original definition in~\cite{CT96}, as it allows us to describe global states at specific times, without reasoning about potential crashes that happen in the future. Actually, our modeling captures more closely the failure patterns from~\cite{CT96}. Regarding the failure detector properties, the strong accuracy property is equivalent to the one in~\cite{CT96}. The other two properties have the form $\F \G( \ldots )$, where we may consider satisfaction only at times after the last process has crashed and thus our crashed predicates coincide with the ones in~\cite{CT96}.
  
}.
However, a process might crash later (and not recover).
Given an execution trace, if process $q$ crashes at time $t$, predicate $ \mathit{Correct}(q)$ is evaluated to false from time $t$.
Predicate $\mathit{Suspected}(p, q)$ corresponds to the variable \textit{suspected} in Algorithm~\ref{algo:detector1}, and depends on the variable $\textit{timeout}[p, q]$ and the waiting time of process $p$ for process $q$.

Algorithm~\ref{algo:detector1} presents the pseudo code of the failure detector of~\cite{CT96}.
A system instance has $N$ processes that  communicate with each other by sending-to-all and receiving messages through unbounded $N^2$ point-to-point communication channels.  
A process performs local computation based on received messages (we assume that a
process also receives the messages that it sends to itself).
In one system step, all processes may take up to one step.
Locally in each step, a process can only make a step in at most one of the locally concurrent tasks.
Some processes may crash, i.e., stop operating.
Correct processes follow Algorithm~\ref{algo:detector1} to detect crashes in the system.
Initially, every correct process sets a default value for a timeout of each other (Line~\ref{algo:initStep1}), i.e., how long it should wait for others, 
and assumes that no processes have crashed (Line~\ref{algo:initStep2}). 
Symbols $\bot$ and $\top$ refer to truth values false and true, respectively.
Every correct process $p$ has three tasks: (i) repeatedly sends an ``alive'' message to all processes (Line~\ref{algo:task1}), 
and (ii) repeatedly produces predictions about crashes of other processes based on timeouts (Line~\ref{algo:task2}),
and (iii) increases a timeout for process $q$ if $p$ has learned that its suspicion on $q$ is wrong (Line~\ref{algo:task3}).
Notice that process $p$ raises suspicion on the operation of process $q$ (Line~\ref{algo:task2}) by considering only information related to $q$: $\fdtimeout{p}{q}, \fdsuspected{p}{q}$, and messages that $p$ has received from $q$ recently.

Algorithm~\ref{algo:detector1} does not satisfy  Eventually Strong Accuracy under asynchrony since there exists no bound on message delay, and messages sent by correct processes might always arrive after the timeout expired. 
Liveness of the failure detector is based on the existence of bounds $\Delta$  on the message delay, and $\Phi$ on the relative speed of processes after reaching the Global Stabilization Time (GST) at some time point $T_0$~\cite{CT96}.\footnote{A time $T_0$ is called the Global Stabilization Time (GST) if $\Delta$ or $\Phi$ holds in $[T_0, \infty]$. }
There are many models of partial synchrony~\cite{DLS88,CT96}.
In this paper, we focus only on the case of unknown bounds $\Delta$ and $\Phi$ because other models might call for abstractions which are out of scope of this paper.
In our case, $T_0 = 1$, and both parameters $\Delta$ and $\Phi$ are arbitrary.
Moreover, the following constraints hold in every execution:
\begin{enumerate}[label=(TC\arabic*),leftmargin=1.3cm]
  \item If message $m$ is placed in the message buffer from process $q$ to process $p$ by some operation $Send(m, p)$ at a time $s_1 \geq 1$, and if process $p$ executes an operation $Receive(p)$ at a time $s_2$ with $s_2 \geq s_1 + \Delta$, then message $m$ must be delivered to $p$ at time $s_2$ or earlier.
  \label{TC1}
  
  \item In every contiguous time interval $\left[ t, t + \Phi \right]$ with $t \geq 1$, every correct process must take at least one step. 
  \label{TC2}  
\end{enumerate}
These constraints make the failure detector parametric in $\Delta$ and $\Phi$.

Moreover, Algorithm~\ref{algo:detector1} is parameterized by the number of processes and by the initial value of the timeout. 
If a default value of the timeout is too small, there exists a case in which sent messages are delivered after the timeout expired.
This behavior violates Strong Accuracy.

As a result, verification of the failure detector faces the following challenges:
\begin{enumerate}[nosep]
  \item Its model of computation lies between synchrony and asynchrony since multiple processes can take a step 
  in a global step. 
  \label{challenge:formalization}

  \item The failure detector is parameterized by the number of processes. Hence, we need to verify infinitely many instances of algorithms.
  \label{challenge:parameterization1}  
  
  \item The initial value of the timeout is an additional parameter in Algorithm~\ref{algo:detector1}.
  \label{challenge:timeout}  

  \item The failure detector relies on a global clock and local clocks.
  A straightforward encoding of a clock with an integer would produce
  an infinite state space.
  \label{challenge:clock}  
  
  \item The algorithm is parametric by time bounds $\Delta$ and $\Phi$.
  \label{challenge:bounds}  
  
  \item Eventually Strong Accuracy and Strong Completeness are liveness properties.
  \label{challenge:reduction} 
\end{enumerate}

\section{Model of Computation} \label{sec:netys20-comp-model}

In this section, we introduce
the class of symmetric point-to-point algorithms, and present how to formalize such algorithms as transition systems.
Since every process follows the same algorithm, we first define a process template that captures the process behavior in Section~\ref{sec:temp-model}.
Every process is an instance of the process template.

In Section~\ref{sec:global-model}, we present the formalization of the global system. 
This formalization is adapted with the time constraints under partial synchrony in Section~\ref{sec:time-constraints}, and our analysis is for the model under partial synchrony.

Intuitively, a global system is a composition of $N$ processes, $N^2$ point-to-point
outgoing message buffers,
and $N$ control components that capture what processes can take a step.
Every process is identified with a unique index in $\nid$, and follows the same deterministic algorithm.
Moreover, a global system allows: (i) multiple processes to take (at most) one step 
in one global step, 
%in one global transition, 
and (ii) some processes to crash.
Every process may execute three kinds of transitions: \emph{internal}, \emph{round}, and stuttering.
%Notice that in one global transition, 
Notice that in one global step, 
some processes may send a message to all, and some may receive messages and do computation.
Hence, we need to decide which processes move, and what happens to the message buffers.
We introduce four sub-rounds: \textit{Schedule}, \textit{Send}, \textit{Receive}, and \textit{Computation}.
The transitions for these sub-rounds are called internal ones. 
A global round transition is a composition of four internal transitions.
We formalize sub-rounds and global steps later.
As a result of modeling, there exists an arbitrary sequence of global configurations which is not accepted in asynchrony.
So, we define so-called \emph{admissible} sequences of global configurations under asynchrony.

Recall that the network topology of algorithms in the symmetric point-to-point class contains $N^2$ \emph{point-to-point} message buffers.
Every transposition on a set of process indexes preserves the network topology.
Importantly, every transposition on process indexes also preserves the
structures of both the process template and the global transition system.
It implies that both the process template and the global transition system are \emph{symmetric}.

\subsection{The Process Template} \label{sec:temp-model} 

We fix a set of process indexes as $\nid$.
Moreover, we assume that the message content does not have indexes of its receiver and sender. 
We let $\msg$ denote a set of potential messages, and $\set(\msg)$ denote the set of sets of messages.

We model a process template as a transition system $\mytemp{N} = (\tempconfs{N}, \temptr{N}, \temprel{N}, \tempinitconf{N})$ 
where 
\[
\tempconfs{N} = \loc \times 
\underbrace{\set(\msg) \times \ldots \times  \set(\msg)}_{N \text{ times}} \times 
\underbrace{\tempdatatype \times \ldots \times \tempdatatype}_{N \text{ times}} 
\] 
is a set of template states~\footnote{We denote $S_1 \times \ldots \times S_m$ by a set $\{ (s_1, \ldots, s_m) \mid \bigwedge_{1 \leq i \leq m} s_i \in S_i \}$ of tuples. 
The elements of the set $\tempconfs{N}$ are tuples with $2N + 1$ elements.}, 
$\temptr{N}$ is a set of template transitions, 
$\temprel{N} \subseteq \tempconfs{N} \times \temptr{N} \times \tempconfs{N}$ is a template transition relation, 
and $\tempinitconf{N} \in \tempconfs{N}$ is an initial state. 
These components of $\mytemp{N}$ are defined as follows.

\paragraph{\textbf{States.}}
A \emph{template state} $\tempconf{}$ is a tuple $\left(\ell, S_1, \ldots, S_N, d_1, \ldots, d_N \right)$ 
where:
\begin{itemize}
\item $\ell \in \loc$ refers to the value of a program counter that ranges over a set  $\loc$ of locations.
We assume that $\loc  = \locsnd \cup \locrcv \cup \loccomp \cup \{\loccrash\}$, and three sets  $\locsnd$, $\locrcv$, $\loccomp$ are disjoint, and $\loccrash$ is a special location of crashes.
To access the program counter, we use a function $\pcFncName \colon \tempconfs{N} \rightarrow \loc$ 
that takes a template state at its input, and produces its program counter as the
output. Let $\tempconf{}(k)$ denote the $k^{th}$ component in a template state  $\tempconf{}$. 
For every $\tempconf{} \in \tempconfs{N}$, we have $\getPC{\tempconf{}} = \tempconf{}(1)$ .

\item $S_i \in \set(\msg)$ refers to a set of messages. 
It is to store the messages received from a process $p_i$ for every $i \in \nid$.
To access a set of received messages from a particular process whose index is in $\nid$, 
we use a function $\boxFncName \colon \tempconfs{N} \times \nid \rightarrow \set(\msg)$ 
that takes a template state $\tempconf{}$ and a process index $i$ at its input, 
and produces the $(i + 1)^{th}$ component of $\tempconf{}$ at the output, 
i.e. for every $\tempconf{} \in \tempconfs{N}$, we have $\getBox{\tempconf{}}{i} = \tempconf{}(1 + i)$ .

\item $d_i \in \tempdatatype$ refers to a local variable related to a process $p_i$ for every $i \in \nid$.
To access a local variable related to a particular process whose index in $\nid$, 
we use a function $\datumFncName \colon \tempconfs{N} \times \nid \rightarrow \tempdatatype $ 
that takes a template state $\tempconf{}$ and a process index $i$ at its input,
and produces the $(1 + N + i)^{th}$ component of $\tempconf{}$ as the output, i.e. $\getDatum{\tempconf{}}{i} = \tempconf{}(1 + N + i)$ for every $\tempconf{} \in \tempconfs{N}$.
For example, for every process $p$ in Algorithm~\ref{algo:detector1}, variable $d_i$ of process $p$ refers to a tuple of the two variables $\textit{timeout}[p, i]$ and $\textit{suspect}[p, i]$.
\end{itemize}

\paragraph{Initial state.}
The initial state $\tempinitconf{N}$ is a tuple 
$\tempinitconf{N} = (\ell_0, \emptyset, \ldots, \emptyset, d_0, \ldots, d_0)
$ where $\ell_0$ is a location, every box for received messages is empty, and every local variable is assigned a constant $d_0 \in \tempdatatype$.

\paragraph{Transitions.}
We define $\temptr{N} = \mathit{CSnd} \cup \mathit{CRcv} \cup \{ \comp, \crash, \stutter \}$ where
\begin{itemize}
\item $\mathit{CSnd}$ is a set of transitions.
Every transition in $\mathit{CSnd}$ refers to a task that does some internal computation, and sends a message to all.
For example, in task 1 in Algorithm~\ref{algo:detector1}, a process increases its local clock, and performs an instruction to send ``alive'' to all.
We let $\csnd(m)$ denote a transition referring to a task with an action to send a message $m \in \msg$ to all.

\item $\mathit{CRcv}$ is a set of transitions.
Every transition in $\mathit{CRcv}$ refers to a task that receives $N$ sets of messages, and does some internal computation.
For example, in task 2 in Algorithm~\ref{algo:detector1}, a process increases its local clock, receives messages, and removes false-negative predictions.
We let $\crcv(S_1, \ldots, S_N)$ denote a transition referring to a task with an action to receive sets $S_1, \ldots, S_N$ of messages. 
These sets $S_1, \ldots, S_N$ are delivered by the global system.

\item $\comp$ is a transition which refers to a task with purely local computation. 
In other words, this task has neither send actions nor receive actions.

\item $\crash$ is a transition for crashes.

\item $\stutter$ is a transition for stuttering steps.
\end{itemize}

\paragraph{Transition relation.}

For two states $\tempconf{}, \tempnextconf{} \in \tempconfs{N}$ and a transition $\mathit{tr} \in \temptr{N}$, instead of $(\tempconf{}, \textit{tr}, \tempnextconf{})$, we write $\tempconf{} \temptrans{} \tempnextconf{}$.
In the model of~\cite{DLS88,CT96}, each process follows the same deterministic algorithm. 
Hence, we assume that 
for every $\tempconf{0} \temptrans{0} \tempnextconf{0}$ and $\tempconf{1} \temptrans{1} \tempnextconf{1}$, if $\tempconf{0} = \tempconf{1}$ and $\mathit{tr}_0 = \mathit{tr}_1$, 
then it follows that $\tempnextconf{0} = \tempnextconf{1}$. 
Moreover, we assume that there exist the following functions which are used to define constraints on the template transition relation:
\begin{itemize}
\item A function $\nextLoc \colon \loc \rightarrow \loc$ takes a location at its input and produces the next location as the output.

\item A function $\genMsg \colon \loc \rightarrow \texttt{Set}(\msg)$ takes a location at its input, and produces a singleton set that contains the message that is sent to all processes in the current task.
The output can be an empty set.
For example, if a process is performing a Receive task, the output of $\genMsg$ is an empty set.

\item A function $\nextVar \colon \loc \times \set(\msg) \times \tempdatatype \rightarrow \tempdatatype$ takes a location, a set of messages, and a local variable's value, and produces a new value of a local variable as the output.
\end{itemize} 
Let us fix functions $\nextLoc, \genMsg$ and $\nextVar$. 
We define the template transitions as follows.

\begin{enumerate}
\item For every message $m \in \msg$, for every pair of states $\tempconf{} , \tempnextconf{} \in \tempconfs{N}$, we have $\tempconf{} \tempsnd{}{m} \tempnextconf{}$ if and only if
\begin{enumerate}
\item $\getPC{\tempconf{}} \in \locsnd \land \getPC{\tempnextconf{}} = \nextLoc(\getPC{\tempconf{}}) \land \{m\} = \genMsg(\getPC{\tempconf{}})$
\item $\forall i \in \nid \colon \getBox{\tempconf{}}{i} = \getBox{\tempnextconf{}}{i}$
\item $\forall i \in \nid \colon \getDatum{\tempnextconf{}}{i} =  
\nextVar \big( \getPC{\tempconf{}},  \emptyset, \getDatum{\tempconf{}}{i} \big) $
\end{enumerate}

Constraint (a) implies that the update of a program counter and the construction of a sent message $m$ depend on only the current value of a program counter, and a process sends only $m$ to all in this step.
For example, process $p$ in Algorithm~\ref{algo:detector1} sends only message ``alive'' in Task 1.
Constraint (b) refers to that no message was delivered.
Constraint (c) implies that the value of $\getDatum{\tempnextconf{}}{i}$  depends only on the current location
and the value of $\getDatum{\tempconf{}}{i}$.
The empty set in Constraint (c) means that no messages have been delivered.

\item For arbitrary sets of messages $S_1, \ldots, S_N \subseteq \msg$,  for every pair of states $\tempconf{} , \tempnextconf{} \in \tempconfs{N}$, we have $\tempconf{} \temprcv{}{S_1}{S_N} \tempnextconf{}$ if and only if the following constraints hold:
\begin{enumerate}
\item $\getPC{\tempconf{}} \in \locrcv \land \getPC{\tempnextconf{}} = \nextLoc(\getPC{\tempconf{}}) \land \emptyset = \genMsg(\getPC{\tempconf{}})$ 

\item $\forall i \in \nid \colon \getBox{\tempnextconf{}}{i} = \getBox{\tempconf{}}{i} \cup S_i$

\item $\forall i \in \nid \colon \getDatum{\tempnextconf{}}{i} =   \nextVar \big( \getPC{\tempconf{}}),  
S_i,
\getDatum{\tempconf{}}{i}\big)$
\end{enumerate}
Constraint (a) in $\crcv$ is similar to constraint (a) in $\csnd$, except that no message is sent in this sub-round.
Constraint (b) refers that messages in a set $S_i$ are from a process indexed $i$, and have been delivered in this step.
For example, in Algorithm~\ref{algo:detector1} Constraint (b) implies that $\getBox{\tempconf{}}{i} \subseteq \{ \text{``alive''} \}$ for every template state $\tempconf{}$ and every index $1 \leq i \leq N$.
After the first ``alive'' message was received, the value of $\getBox{\tempconf{}}{i}$ is unchanged.
This does not raise any issues in our analysis as Line~7 in Algorithm~\ref{algo:detector1} considers only how long process $p$ has waited for a new message from process $q$.
Constraint (c) in $\crcv$ implies that the value of $\getDatum{\tempnextconf{}}{i}$  depends on only the current location, 
the set $S_i$ of messages that have been delivered, 
and the value of $\getDatum{\tempconf{}}{i}$.

\item For every pair of states $\tempconf{} , \tempnextconf{} \in \tempconfs{N}$, we have $\tempconf{} \tempcomp{} \tempnextconf{}$ if and only if the following constraints hold:
\begin{enumerate}
\item $\getPC{\tempconf{}} \in \loccomp \land \getPC{\tempnextconf{}} = \nextLoc(\getPC{\tempconf{}}) \land \emptyset = \genMsg(\getPC{\tempconf{}})$
\item $\forall i \in \nid \colon \getBox{\tempnextconf{}}{i} = \getBox{\tempconf{}}{i}$
\item $\forall i \in \nid \colon  \getDatum{\tempnextconf{}}{i} = \nextVar \big( \getPC{\tempconf{}},  
\emptyset,
\getDatum{\tempconf{}}{i} \big)$
\end{enumerate}
Hence, this step has only local computation. No message is sent or delivered.

\item For every pair of states $\tempconf{} , \tempnextconf{} \in \tempconfs{N}$, we have
$\tempconf{} \tempcrash{}  \tempnextconf{} $ if and only if the following constraints hold:
\begin{enumerate}
\item $\getPC{\tempconf{}} \neq \loccrash \land \getPC{\tempnextconf{}} = \loccrash$
\item $\forall i \in \nid \colon \getBox{\tempconf{}}{i} = \getBox{\tempnextconf{}}{i} \land \getDatum{\tempconf{}}{i} = \getDatum{\tempnextconf{}}{i}$
\end{enumerate}
Only the program counter is updated by switching to $\loccrash$.

\item For every pair of states $\tempconf{} , \tempnextconf{} \in \tempconfs{N}$, we have $\tempconf{} \tempstutter{} \tempnextconf{}$ if and only if $\tempconf{} = \tempnextconf{}$.
\end{enumerate} 

\subsection{Modeling the Global Distributed Systems} \label{sec:global-model} 

We now present the formalization of the global system.
In this model, multiple processes might take a step in a global step.
This characteristic allows us to extend this model with partial synchrony constraints that are formalized in Section~\ref{sec:time-constraints}.
To capture the semantics of asynchrony, we simply need a constraint that only one process can take a step in a global step~\cite{hagit2004}.
This constraint is formalized in the end of this subsection.

Given $N$ processes which are instantiated from the same process template $\mytemp{N} = (\tempconfs{N}, \temptr{N}, \temprel{N}, \tempinitconf{N})$, the global system is a composition of (i) these processes, and (ii) $N^2$ point-to-point buffers for in-transit messages, and (iii) $N$ control components that capture what processes can take a step.
We formalize the global system as a transition system $\gsys{N} = \left( \gconfs{N}, \gtr{N}, \grel{N}, \ginitconf{N} \right)$
where

\begin{itemize}
  \item $\gconfs{N} = (\tempconfs{N})^N \times \set(\msg)^{N \cdot N} \times \Boolean^N$ is a set of global configurations, 
  
  \item $\gtr{N}$ is a set of global \emph{internal}, \emph{round}, and stuttering transitions,
  
  \item  $\grel{N} \subseteq \gconfs{N} \times \gtr{N} \times \gconfs{N}$ is a global transition relation, and

  \item $\ginitconf{N}$ is an initial configuration. 

\end{itemize}
These components are defined as follows.

\paragraph{Configurations.}

A \emph{global configuration} $\gconf{}$ is defined as a following tuple 
\[
\gconf{} = \left(q_1, \ldots, q_N, S^1_1, S^2_1 \ldots, S^r_s, \ldots S^N_N, act_1, \ldots, act_N \right)
\]
where:
\begin{itemize}
\item $q_i \in \tempconfs{N}$: This component is a state of a process $p_i$ for every $i \in \nid$.
To access a local state of a particular process, we use a function $\procFncName \colon \gconfs{N} \times \nid \rightarrow \tempconfs{N} $
that takes input as a global configuration $\gconf{}$ and a process index $i$, 
and produces output as the $i^{th}$ component of $\gconf{}$ which is a state of a process $p_i$.
Let $\gconf{}(i)$ denote the $i^{th}$ component of a global configuration $\gconf{}$. 
For every $i \in \nid$, we have $\getProc{\gconf{}}{i} = \gconf{}(i) = q_i$.

\item $S^r_s \in \set(\msg)$: This component  is a set of in-transit messages from a process $p_s$ to a process $p_r$ for every $s, r \in \nid$.
To access a set of in-transit messages between two processes, we use a function $\bufFncName \colon \gconfs{N} \times \nid \times \nid \rightarrow \set(\msg)$ that takes input as a global configuration $\gconf{}$, and two process indexes $s, r$, and produces output as the $(s \cdot N  + r)^{th}$ component of $\gconf{}$ which is a message buffer from a process $p_s$ (sender) to a process $p_r$ (receiver).
Formally, we have $\getBuf{\gconf{}}{s}{r} = \gconf{}(s \cdot N  + r) = S^r_s$ for every $s, r \in \nid$.

\item $act_i \in \Boolean$: This component says whether a process $p_i$ can take one step in a global step for every $i \in \nid$.
To access a control component, we use a function $\activeFncName \colon \gconfs{N} \times \nid \rightarrow \Boolean$
that takes input as a configuration $\gconf{}$ and a process index $i$, and produces output as the $((N + 1) \cdot N + i)^{th}$ component of $\gconf{}$ which refers to whether a process $p_i$ can take a step.
Formally, we have $\getActive{\gconf{}}{i} = \gconf{}((N + 1) \cdot N + i)$ for every $i \in \nid$.
The environment sets the values of $act_1, \ldots, act_N$ in the sub-round Schedule defined later.
\end{itemize}
We will write $\gconf{} \in (\tempconfs{N})^N \times \set(\msg)^{N \cdot N} \times \Boolean^N$ or $\gconf{} \in \gconfs{N}$.
\paragraph{Initial configuration.}
The global system $\gsys{N}$ has one initial configuration $\ginitconf{N}$, and it 
must satisfy the following constraints:
\begin{enumerate}
\item $\forall i \in \nid \colon \lnot \getActive{\ginitconf{N}}{i} \land \getProc{\ginitconf{N}}{i} = \tempinitconf{N}$
\item $\forall s, r \in \nid \colon \getBuf{\ginitconf{N}}{s}{r}  = \emptyset$
\end{enumerate}

\paragraph{Global stuttering transition.}
We extend the relation $\leadsto$ with stuttering: for every configuration $\gconf{}$, we allow $\gconf{} \leadsto \gconf{}$.
The stuttering transition is necessary in the proof of Lemma~\ref{lem:TraceEquivTwo} presented in Section~\ref{sec:netys20-cutoff-res}.

\paragraph{Global internal transitions.}
In the model of~\cite{DLS88,CT96}, many processes can take a step in a global step.
We assume that a computation of the distributed system is organized in rounds, i.e. global ticks, and every round is organized as four sub-rounds called \textit{Schedule, Send, Receive}, and \textit{Computation}.
To model that as a transition system, for every sub-round we define a corresponding transition:  $\transsched$ for the sub-round \textit{Schedule}, $\transsnd$ for the sub-round \textit{Send}, $\transrcv$ for the sub-round \textit{Receive}, $\transcomp$ for the sub-round \text{Comp}. 
These transitions are called global \emph{internal} transitions.
We define the semantics of these sub-rounds as follows.

\begin{enumerate}
\item Sub-round \textit{Schedule}. The environment starts with a global configuration where every process is inactive, and move to another  by non-deterministically deciding what processes become crashed, and 
what processes take a step in the current global step. 
Every correct process takes a stuttering step, 
and every faulty process is inactive.
If a process $p$ is crashed in this sub-round, every incoming message buffer to $p$ is set to the empty set.
Formally, for  $ \gconf{}, \gnextconf{} \in \gconfs{N}$, we have $\gconf{} \transsched \gnextconf{}$ if the following constraints hold:
\begin{enumerate}
\item $\forall i \in \nid \colon \lnot \getActive{\gconf{}}{i}$
\item $\forall i \in \nid \colon \getProc{\gconf{}}{i} \transstutter  \getProc{\gnextconf{}}{i}   \lor \getProc{\gconf{}}{i} \tempcrash{} \getProc{\gnextconf{}}{i}$
\item $\forall i \in \nid \colon \getPC{\getProc{\gnextconf{}}{i}} = \loccrash \Rightarrow \lnot \getActive{\gnextconf{}}{i}$
\item $\forall s, r \in \nid \colon \getPC{ \getProc{\gnextconf{}}{r} } \neq \loccrash \Rightarrow \getBuf{\gconf{}}{s}{r} = \getBuf{\gnextconf{}}{s}{r} $
\item $\forall r \in \nid \colon \getPC{ \getProc{\gnextconf{}}{r} } = \loccrash \Rightarrow (\forall s \in \nid \colon \getBuf{\gnextconf{}}{s}{r} = \emptyset)$
\end{enumerate}

We let predicate $\CanRun(\gconf{}, i, L)$ denote whether process $i$ whose location at the configuration $\gconf{}$ is in $L$ takes a step from $\gconf{}$. Formally, we have
%, i.e., 
\[
\CanRun(\gconf{}, i, L)   \triangleq \; \getActive{\gconf{}}{i} \land  \getPC{\getProc{\gconf{}}{i}} \in L
\]
Predicate $\CanRun$ is used in the definitions of other sub-rounds.
\item Sub-round \emph{Send}. Only processes that perform send actions can take a step in this sub-round.
Such processes become inactive at the end of this sub-round.
Fresh sent messages are added to corresponding message buffers. 
To define the semantics of the sub-round Send, we use the following predicates:

\begin{align*}
  \qquad \SFrozen(\gconf{}, \gnextconf{}, i)   \triangleq \; & \;\;\;\; \getProc{\gconf{}}{i} \transstutter \getProc{\gnextconf{}}{i}  \\
  & \land  \getActive{\gconf{}}{i} = \getActive{\gnextconf{}}{i}   \\
  & \land \forall r \in \nid \colon \getBuf{\gconf{}}{i}{r} = \getBuf{\gnextconf{}}{i}{r}   \\
  \qquad \predsnd(\gconf{}, \gnextconf{}, i, m)   \triangleq \; & \;\;\;\; \forall r \in \nid \colon m \notin \getBuf{\gconf{}}{i}{r}   \\
  & \land \forall r \in \nid \colon \getBuf{\gnextconf{}}{i}{r} = \{m\} \cup \getBuf{\gconf{}}{i}{r} \,  \\
  & \land \getProc{\gconf{}}{i} \tempsnd{}{m} \getProc{\gnextconf{}}{i}
\end{align*}

Formally, for  $ \gconf{}, \gnextconf{} \in \gconfs{N}$, we have $\gconf{} \transsnd \gnextconf{}$ if the following constraints hold:
\begin{enumerate}
\item $\forall i \in \nid \colon \lnot \CanRun(\gconf{}, i, \locsnd) \Leftrightarrow \SFrozen(\gconf{}, \gnextconf{}, i)$
\item $\forall i \in \nid \colon  \CanRun(\gconf{}, i, \locsnd) \Leftrightarrow \exists m \in \msg \colon \predsnd(\gconf{}, \gnextconf{}, i, m)$
\item $\forall i \in \nid \colon \CanRun(\gconf{}, i, \locsnd) \Rightarrow \lnot \getActive{\gnextconf{}}{i}$
\end{enumerate}

The semantics of the Send sub-round forces that the Send primitive is atomic.

\item Sub-round \textit{Receive}. Only processes that perform receive actions can take a step in this sub-round.
Such processes become inactive at the end of this sub-round.
Sets of delivered messages that may be empty are removed from corresponding message buffers. 
To define the semantics of this sub-round, we use the following predicates:
\begin{align*}
\qquad \RFrozen(\gconf{}, \gnextconf{}, i)   \triangleq \; & \;\;\;\; \getProc{\gconf{}}{i} \transstutter \getProc{\gnextconf{}}{i}  \\
& \land  \getActive{\gconf{}}{i} = \getActive{\gnextconf{}}{i}   \\
& \land \forall s \in \nid \colon \getBuf{\gconf{}}{s}{i} = \getBuf{\gnextconf{}}{s}{i}  \\
\predrcv(\gconf{}, \gnextconf{}, i, S_1, \ldots, S_N)   \triangleq \; & 
\;\;\;\; \forall s \in \nid \colon S_s \cap \getBuf{\gnextconf{}}{s}{i} = \emptyset \\
& \land \forall s \in \nid \colon \getBuf{\gnextconf{}}{s}{i} \cup S_s = \getBuf{\gconf{}}{s}{i} \\
& \land \getProc{\gconf{}}{i} \temprcv{}{S_1}{S_N} \getProc{\gnextconf{}}{i} 
\end{align*}
Formally, for  $ \gconf{}, \gnextconf{} \in \gconfs{N}$, we have $\gconf{} \transrcv \gnextconf{}$ if the following constraints hold:
\begin{enumerate}
\item $\forall i \in \nid \colon \lnot \CanRun(\gconf{}, i, \locrcv) \Leftrightarrow \RFrozen(\gconf{}, \gnextconf{}, i)$

\item $\begin{array}[t]{l l}
\forall i \in \nid \colon & \CanRun(\gconf{}, i, \locrcv) \\
& \qquad \Leftrightarrow \exists S_1, \ldots, S_N \subseteq \msg \colon \predrcv(\gconf{}, \gnextconf{}, i, S_1, \ldots, S_N)
\end{array}$

\item $\forall i \in \nid \colon \CanRun(\gconf{}, i, \locrcv) \Rightarrow \lnot \getActive{\gnextconf{}}{i}$
\end{enumerate}

\item Sub-round \textit{Computation}. Only processes that perform internal computation actions can take a step in this sub-round.
Such processes become inactive at the end of this sub-round.
Every message buffer is unchanged. 
Formally, for  $ \gconf{}, \gnextconf{} \in \gconfs{N}$, we have $\gconf{} \transcomp \gnextconf{}$ if the following constraints hold:
\begin{enumerate}
\item $\forall i \in \nid \colon  \CanRun(\gconf{}, i, \loccomp) \Leftrightarrow \getProc{\gconf{}}{i} \tempcomp{} \getProc{\gnextconf{}}{i}$
\item $\forall i \in \nid \colon \lnot  \CanRun(\gconf{}, i, \loccomp) \Leftrightarrow \getProc{\gconf{}}{i} \transstutter \getProc{\gnextconf{}}{i}$
\item $\forall s, r \in \nid \colon \getBuf{\gconf{}}{s}{r} = \getBuf{\gnextconf{}}{s}{r}$
\item $\forall i \in \nid \colon \CanRun(\gconf{}, i, \loccomp) \Rightarrow \lnot \getActive{\gnextconf{}}{i}$
\end{enumerate}

\end{enumerate}

\begin{rem}
Predicate $\predsnd$ refers that at most one ``alive'' message in Algorithm~\ref{algo:detector1} is in every message buffer.
In Section~\ref{sec:time-constraints}, we extend our formalization by introducing the notion of time and by modeling time constraints under partial synchrony.
In that formalization, every ``alive'' message is tagged with its age, and therefore, the message buffers can have multiple messages.
\end{rem}

\begin{rem}
  Observe that the definitions of $\gconf{} \transsched \gnextconf{}$, and $\gconf{} \transsnd \gnextconf{}$, and $\gconf{} \transrcv \gnextconf{}$, and $\gconf{} \transcomp \gnextconf{}$ allow $\gconf{} = \gnextconf{}$, that is stuttering.
  This captures, e.g.,
global steps in~\cite{DLS88,CT96} where no process sends a message.
\end{rem}

\paragraph{Global round transitions.}
Intuitively, every global \emph{round} transition is induced by a sequence of four transitions: a $\transsched$ transition, a $\transsnd$ transition, a $\transrcv$ transition, and a $\transcomp$ transition.
We let $\leadsto$ denote global round transitions. 
For every pair of global configurations $ \gconf{0}, \gconf{4} \in \gconfs{N}$, we say $\gconf{0} \leadsto \gconf{4}$ if there exist three global configurations $\gconf{1}, \gconf{2}, \gconf{3} \in \gconfs{N}$ such that
$
\gconf{0} \transsched \gconf{1} \transsnd \gconf{2} \transrcv \gconf{3} \transcomp \gconf{4}
$.
Moreover, global round transitions allow some processes to crash only in the sub-round Schedule.
We call these faults \emph{clean crashes}.
Notice that correct process $i$ can make at most one global internal transition in every global round transition since the component $act_i$ is false after process $i$ makes a transition.

\paragraph{Admissible sequences.}
An infinite sequence $\pi = \gconf{0} \gconf{1} \ldots$ of global configurations in $\gsys{N}$ is \emph{admissible} if the following constraints hold:
\begin{enumerate}
  \item $\gconf{0}$ is the initial state, i.e. $\gconf{0} = \ginitconf{N}$, and
  \item $\pi$ is stuttering equivalent with an infinite sequence $\pi^{\prime} = \gnextconf{0} \gnextconf{1} \ldots$ such that 
  $
  \gnextconf{4k} \transsched \gnextconf{4k+1} \transsnd \gnextconf{4k+2} \transrcv \gnextconf{4k+3} \transcomp \gnextconf{4k+4}
  $ for every $k \geq 0$. 
  
\end{enumerate} 
Notice that it immediately follows by this definition that if $\pi = \gconf{0} \gconf{1} \ldots$ is an admissible sequence of configurations in $\gsys{N}$, then $\gnextconf{4k} \leadsto \gnextconf{4k+4}$ for every $k \geq 0$.
From now on, we only consider admissible sequences of global configurations.

\paragraph{Admissible sequences under synchrony.} Let $\pi = \gconf{0} \gconf{1} \ldots $ be an admissible sequence of global configurations in $\mathcal{G}_N$. 
As every correct process makes a transition in every global step under synchrony\cite{hagit2004}, we say that $\pi$ is under synchrony if every correct process is active after a sub-round Schedule. 
Formally, for every transition $\gconf{} \transsched \gnextconf{}$ in $\pi$, the following constraint holds:
$
\forall i \in \nid \colon \getPC{\getProc{\gnextconf{}}{i}} \neq \loccrash \Rightarrow \getActive{\gnextconf{}}{i} 
$.

\paragraph{Admissible sequences under asynchrony.} Let $\pi = \gconf{0} \gconf{1} \ldots $ be an admissible sequence of global configurations in $\mathcal{G}_N$. 
As at most one process can make a transition in every global step under asynchrony\cite{hagit2004}, we say that $\pi$ is under asynchrony if at most one process is active after a sub-round Schedule. 
Formally, for every transition $\gconf{} \transsched \gnextconf{}$ in $\pi$, the following constraint holds:
$
\forall i, j \in \nid \colon \getActive{\gnextconf{}}{i} \land \getActive{\gnextconf{}}{j} \Rightarrow i = j
$.

\subsection{Modeling Time Constraints under Partial Synchrony} \label{sec:time-constraints}

Time parameters in partial synchrony only reduce the execution space compared to asynchrony.
Hence, we can formalize the system behaviors under partialy synchrony by extending the above formalization of the system behaviors with the notion of time, message ages, time constraints, and admissible sequences of configurations under partial synchrony.
They are defined as follows.

\paragraph{Time.} Time is progressing with global round transitions. 
Formally, let $\pi = \gconf{0} \gconf{1} \ldots $ be an admissible sequence of global configurations in $\mathcal{G}_N$.
We say that the configuration $\gconf{0}$ is at time 0, and that four configurations $\gconf{4k - 3}, \ldots, \gconf{4k}$ are at time $k$ for every $k > 0$. 

Recall that in Section~\ref{sec:global-model}, a global round transition 
is induced of
 a sequence of four sub-rounds: Schedule, Send, Receive, and Computation.
In an admissible sequence $\pi = \gconf{0} \gconf{1} \ldots $ of global configurations in $\mathcal{G}_N$, for every $k > 0$, every sub-sequence of four configurations $\gconf{4k-3}, \ldots, \gconf{4k}$ presents one global round transition.
Configuration $\gconf{4k-3}$ is in sub-round Schedule, and configuration $\gconf{4k}$ is in sub-round Computation for every $k > 0$.
So, the notion of time says that the global round transition $\gconf{4k-3} \leadsto \gconf{4k}$ happens at time $k$.

\paragraph{Message ages.} 
Now we discuss the formalization of message ages. 
For every sent message $m$, the global system tags it with its current age, i.e., $(m, \mathit{age}_m)$. 
Message ages require that the type of message buffers needs to be changed to
$\bufFncName \colon \gconfs{N} \times \nid \times \nid \rightarrow \set(\msg \times \mathbb{N})$.

In our formalization, when message $m$ was added to the message buffer in sub-round Send, its age is~0. Instead of predicate $\predsnd$, our formalization now uses the following predicate $\predsnd'$.
\begin{align*}
\qquad \predsnd'(\gconf{}, \gnextconf{}, i, m)   \triangleq \; & \;\;\;\; \forall r \in \nid \colon (m, 0) \notin \getBuf{\gconf{}}{i}{r}   \\
& \land \forall r \in \nid \colon \getBuf{\gnextconf{}}{i}{r} = \{(m, 0)\} \cup \getBuf{\gconf{}}{i}{r} \,  \\
& \land \getProc{\gconf{}}{i} \tempsnd{}{m} \getProc{\gnextconf{}}{i}
\end{align*}

Message ages are increased by 1 when the global system takes a $\transsched$ transition.
Formally, for every time $k \geq 0$, for every process $s, r \in \nid$, the following constraints hold:
\begin{enumerate}[label=(\roman*)]
	\item For every message $(m, \textit{age}_m)$ in $\getBuf{\gconf{4k}}{s}{r}$, there exists a message $(m^{\prime}, \textit{age}_{m^{\prime}})$ in $\getBuf{\gconf{4k+1}}{s}{r}$
	such that $m = m^{\prime}$ and $\textit{age}_{m^{\prime}} = \textit{age}_m + 1$. \label{MsgAge1}

	\item For every message $(m^{\prime}, \textit{age}_{m^{\prime}})$ in $\getBuf{\gconf{4k+1}}{s}{r}$, there exists a message $(m, \textit{age}_m)$ in $\getBuf{\gconf{4k}}{s}{r}$ 
	such that $m = m^{\prime}$ and $\textit{age}_{m^{\prime}} = \textit{age}_m + 1$. \label{MsgAge2}  
\end{enumerate}
Constraint \ref{MsgAge1} ensures that every in-transit message age will be added by one time-unit in the sub-round Schedule.
Constraint \ref{MsgAge2} sensures that no new messages will be added in $\getBuf{\gconf{4k+1}}{s}{r}$.
These two constraints are used to replace Constraint (1d) about unchanged message buffers in the definition of sub-round Schedule in Section~\ref{sec:global-model}.

Moreover, the age of an in-transit message is unchanged in other sub-rounds.
Formally, for every time $k > 0$, for every $0 \leq \ell \leq 3$, for every pair of processes $s, r \in \nid$, for every message $(m, \textit{age}_m)$ in $\getBuf{\gconf{4k-\ell}}{s}{r}$, there exists  $(m^{\prime}, \textit{age}_{m^{\prime}})$ in $\getBuf{\gconf{4k-3}}{s}{r}$
such that $m = m^{\prime}$ and $\textit{age}_m = \textit{age}_{m^{\prime}}$.

Finally, message ages are not delivered to processes in sub-round Receive.
Instead of predicate $\predrcv$, our formalization now uses the following predicate $\predrcv'$.
\begin{align*}
\predrcv(\gconf{}, \gnextconf{}, i, S_1, \ldots, S_N)   \triangleq \; & 
\;\;\;\; \forall s \in \nid \colon S_s \cap \getBuf{\gnextconf{}}{s}{i} = \emptyset \\
& \land \forall s \in \nid \colon \getBuf{\gnextconf{}}{s}{i} \cup S_s = \getBuf{\gconf{}}{s}{i} \\
& \land \getProc{\gconf{}}{i} \temprcv{}{g(S_1)}{g(S_N)} \getProc{\gnextconf{}}{i} 
\end{align*}
where function $g \colon \set(\msg \times \mathbb{N}) \rightarrow \set(\msg)$ is to detag message ages in a set $S$.
Formally, we have two following constraints:
\begin{enumerate}
  \item For every $(m, age_m) \in S$, it holds $m \in g(S)$.
  
  \item For every $m \in g(S)$, there exists $age_m \in \mathbb{N}$ such that $(m, age_m) \in S$.
\end{enumerate}

\paragraph{Partial synchrony constraints.} We here focus on the case of unknown bounds.
Recall that Constraints~\ref{TC1} and \ref{TC2} hold in this case.
Given an admissible sequence $\pi = \gconf{0} \gconf{1} \ldots $ of global configurations in $\mathcal{G}_N$, Constraints~\ref{TC1} and \ref{TC2} on $\pi$
can be formalized as follows, respectively: 
\begin{enumerate}[label=(PS\arabic*),leftmargin=1.3cm]
	\item For every process $r \in \nid$, for every time $k > 0$, if $\CanRun(\kappa_{4k - 2},r,\locrcv)$,  then for every process $s \in \nid$, there exists no message $(m, \textit{age}_m)$ in $\getBuf{\gconf{4k - 1}}{s}{r}$ such that $\textit{age}_m \geq \Delta$. \label{PS1}

	\item For every process $i \in \nid$, for every time interval $[k, k + \Phi ]$, if 
  we have that 
  $\getPC{\getProc{\gconf{\ell}}{i}} \neq \loccrash$ for every configuration index in the interval 
  $\left[ 4k-3, 4(k + \Phi) \right]$,
	then there exist a configuration index $t$ in $\left[ 4k-3, 4(k + \Phi) \right]$ and a set $L$ of locations such that  $\CanRun(\kappa_{t},i,L)$ where $L$ is one of $\locsnd, \locrcv,$ and $\loccomp$. \label{PS2}
\end{enumerate}

Constraint~\ref{PS1} requires that if process $r$ takes a step from $\gconf{4k - 2}$ in the sub-round Receive, then there exists no in-transit messages (sent to process $r$) whose ages are at least $\Delta$ time-units in $\gconf{4k - 1}$, that is, older messages must have been received before that.
In principle, partial synchrony allows messages to be older than $\Delta$ time-units as long as the receiver does not take a step after the message reaches age Delta. 
Consistent to~\ref{TC1}, whenever a receiver takes a step after a message is older than $\Delta$ time-units, the reception step removes it from the buffer. 
This limitation is to enabled processes. 
Constraint~\ref{PS2} ensures that for every time interval $[k, k + \Phi ]$ with configurations  $\gconf{4k-3}, \ldots, \gconf{4(k + \Phi)}$, for every process $i \in 1..N$, if process $i$ is correct in this time interval, there exist a configuration $\gconf{\ell_0} \in \{ \gconf{4k-3}, \ldots, \gconf{4(k + \Phi)} \}$ and a set $L$ of locations such that the location of process $i$ at $\gconf{\ell_0}$ is in $L$ and process $i$ takes a step from $\gconf{\ell_0}$.

\paragraph{Admissible sequences under partial synchrony.} Let $\pi = \gconf{0} \gconf{1} \ldots $ be an admissible sequence of global configurations in $\mathcal{G}_N$. 
We say that $\pi$ is under partial synchrony if Constraints~\ref{PS1} and~\ref{PS2} hold in $\pi$.
Notice that admissible sequences under partial synchrony allow multiple processes to make a transition in a time unit.

\section{Cutoff Results in the Model of the Global Distributed Systems} \label{sec:netys20-cutoff-res}

Let $\mathcal{A}$ be a~\classname{} algorithm.
In this section, we show cutoff results for the number of processes in the algorithm $\mathcal{A}$ in the unrestricted model.
These results are Theorems \ref{thm:NCutoffOneIndexNetysA} and \ref{thm:NCutoffTwoIndexesNetysA},
and the detailed proofs are provided in Section~\ref{sec:netys20-proofs}.
With these cutoff results,
one can verify two properties Strong Completeness and Eventually Strong Accuracy of the failure detector of \cite{CT96} by model checking two instances of sizes 1 and 2 in case of synchrony.

\newcommand{\thmNCutoffOneIndexNetysA}{
  Let $\mathcal{A}$ be a~\classname{} algorithm under the unrestricted model. Let $\gsys{1}$ and $\gsys{N}$ be instances of 1 and $N$ processes respectively for some $N \geq 1$.
  Let $\mathit{Path}_1$ and $\mathit{Path}_N$ be sets of all admissible sequences of configurations in $\gsys{1}$ and in $\gsys{N}$, respectively.
  Let $\omega_{\{i\}}$ be a \ltlx{} formula  in which every predicate takes one of the forms: $P_1(i)$ or $P_2(i, i)$ where $i$ is an index in $\nid$. 
  Then, it follows that:
  \[
  \Big(\forall \pi_N \in \mathit{Path}_N \colon \gsys{N}, \pi_N \models \bigwedge_{i \in \nid}  \omega_{\{i\}} \Big) \; \Leftrightarrow \; \Big(\forall \pi_1 \in \mathit{Path}_1 \colon \gsys{1}, \pi_1 \models  \omega_{\{1\}}\Big)
  \]
}

\begin{thm} \label{thm:NCutoffOneIndexNetysA}
  \thmNCutoffOneIndexNetysA
\end{thm}

\newcommand{\thmNCutoffTwoIndexesNetysA}{
  Let $\mathcal{A}$ be a~\classname{} algorithm under the unrestricted model. Let $\gsys{2}$ and $\gsys{N}$ be instances of 2 and $N$ processes respectively for some $N \geq 2$.
  Let $\mathit{Path}_2$ and $\mathit{Path}_N$ be sets of all admissible sequences of configurations in $\gsys{2}$ and in $\gsys{N}$, respectively.
  Let $\psi_{\{i, j\}}$ be an \ltlx{} formula in which every predicate takes one of the forms: $P_1(i)$, or $P_2(j)$, or $P_3(i, j)$, or $P_4(j, i)$ where $i$ and $j$ are different indexes in $\nid$.
  It follows that:
  \[
  \big(  \forall \pi_N \in \mathit{Path}_N \colon \gsys{N}, \pi_N \models  \bigwedge^{i \neq j}_{i, j \in \nid} \psi_{\{i, j\}} \big)  \Leftrightarrow  \big(  \forall \pi_2 \in \mathit{Path}_2 \colon \gsys{2}, \pi_2 \models    \psi_{\{1, 2 \}} \big) 
  \]
}

\begin{thm} \label{thm:NCutoffTwoIndexesNetysA}
  \thmNCutoffTwoIndexesNetysA
\end{thm}

Since the proof of Theorem~\ref{thm:NCutoffOneIndexNetysA} is similar to the one of Theorem~\ref{thm:NCutoffTwoIndexesNetysA}, 
we focus on Theorem~\ref{thm:NCutoffTwoIndexesNetysA} here.
Its proof is based on the symmetric characteristics in the system model (the network topology and the three functions $\nextLoc{}, \genMsg{}, \text{ and } \nextVar{}$) and correctness properties, and on the following lemmas. 

\begin{itemize}
  \item Lemma \ref{lem:SymProcTemp} says that every transposition on a set of process indexes $\nid$ preserves the structure of the process template $\mytemp{N}$.
  
  \item Lemma \ref{lem:SymGlobalSys} says that every transposition on a set of process indexes $\nid$ preserves the structure of the global transition system $\gsys{N}$ for every $N \geq 1$.
    
  \item Lemma \ref{lem:TraceEquivTwo} says that $\gsys{2}$ and $\gsys{N}$ are trace equivalent under a set $AP_{\{1, 2\}}$ of predicates that take one of the forms: $P_1(i)$, or $P_2(j)$, or $P_3(i, j)$, or $P_4(j, i)$.

\end{itemize}

In the following, we present definitions and constructions to prove these lemmas.

\subsection{Index Transpositions and Symmetric Point-to-point Systems}

We first recall the definition of transposition.
Given a set $\nid$ of indexes, we call a bijection $\transposition \colon \nid \rightarrow \nid$ a transposition between two indexes $i, j  \in \nid$ if the following properties hold: $\switchIndex{\transposition}{i} = j$, and
$\transposition(j) = i$, and
$\forall k \in \nid \colon (k \neq i \land k \neq j) \Rightarrow \transposition(k) = k$.
We let $\swapij{i}{j}$ denote a transposition between two indexes $i$ and $j$.

The application of a transposition to a template state is given in Definition~\ref{def:TransposAndTempConf}.
Informally, applying a transposition $\transposition =\swapij{i}{j}$ to a template state $\tempconf{}$ generates a new template state
%, denoted $\switchTempConf{\tempconf{}}{\transposition}$, 
by switching only the evaluation of $\boxFncName$ and $\datumFncName$ at indexes $i$ and $j$. 
The application of a transposition to a global configuration is provided in Definition~\ref{def:TransposAndGlobalConf}.
In addition to process configurations, we need to change message buffers and control components.
We override notation by writing $\switchTempConf{\tempconf{}}{\transposition}$ and $\switchGlobalConf{\gconf{}}{\transposition}$ to refer the application of a transposition $\transposition$ to a  state $\tempconf{}$ and to a configuration $\gconf{}$, respectively.
These functions $\transposition_Q$ and $\transposition_C$ are named a local transposition  and a global transposition, respectively.

\begin{defi}[Local Transposition] \label{def:TransposAndTempConf}
Let $\mytemp{N}$ be a process template with process indexes $\nid$, 
and $\tempconf{} = \left(\ell, S_1, \ldots, S_N, d_1, \ldots, d_N \right)$ be a state in $\mytemp{N}$.
Let $\transposition =\swapij{i}{j}$ be a transposition on $\nid$.
The application of $\transposition$ to $\tempconf{}$, denoted as $\switchTempConf{\tempconf{}}{\transposition}$, generates a tuple
$ \left(\ell^{\prime}, S^{\prime}_1, \ldots, S^{\prime}_N, d^{\prime}_1, \ldots, d^{\prime}_N \right)$
such that 
\begin{enumerate}
\item $\ell = \ell^{\prime}$, and $S_i = S^{\prime}_j$, and $S_j = S^{\prime}_i$, and
 $d_i = d^{\prime}_j$ and $d_j = d^{\prime}_i$, and
\item $\forall k \in \nid \colon (k \neq i \land k \neq j) \Rightarrow (S_k = S^{\prime}_k \land d_k = d^{\prime}_k)$
\end{enumerate}
\end{defi}

\begin{defi}[Global Transposition] \label{def:TransposAndGlobalConf}
  Let $\gsys{N}$ be a global system with process indexes $\nid$, 
  and $\gconf{}$ be a  configuration in $\gsys{N}$.
  Let $\transposition =\swapij{i}{j}$ be a transposition on $\nid$.
  The application of $\transposition$ to $\gconf{}$, denoted as $\switchGlobalConf{\gconf{}}{\transposition}$, generates a configuration in $\gsys{N}$ which satisfies following properties:
  \begin{enumerate}
    \item $\forall i \in \nid \colon \getProc{\switchGlobalConf{\gconf{}}{\transposition}}{\switchIndex{\transposition}{i}}   = \switchTempConf{\getProc{\gconf{}}{i}}{\transposition}$.
    \item $\forall s, r \in \nid \colon \getBuf{\switchGlobalConf{\gconf{}}{\transposition}}{\switchIndex{\transposition}{s}}{\switchIndex{\transposition}{r}} = \getBuf{\gconf{}}{s}{r} $
    \item $\forall i \in \nid \colon \getActive{\switchGlobalConf{\gconf{}}{\transposition}}{\switchIndex{\transposition}{i}} = \getActive{\gconf{}}{i}$
  \end{enumerate}
\end{defi}

Since the content of every message in $\msg$ does not have indexes of the receiver and sender, no transposition affects the messages.
We define the application of a transposition to one of send, compute, crash, and stutter template transitions return the same transition.
We extend the application of a transposition to a receive template transition as in Definition~\ref{def:TransposeAndRcv}.

\begin{defi}[Receive-transition Transposition] \label{def:TransposeAndRcv}
Let $\mytemp{N}$ be a process template with indexes $\nid$, and $\transposition =\swapij{i}{j}$ be a transposition on $\nid$. 
Let $\crcv(S_1, \ldots, S_N)$ be a transition in $\mytemp{N}$ which refers to a task with a receive action.
We let $\transposition_R(\crcv(S_1, \ldots, S_N))$  denote
the application of $\transposition$ to $\crcv(S_1, \ldots, S_N)$, 
and this application returns a new transition $ \crcv(S^{\prime}_1, \ldots, S^{\prime}_N)$ in $\mytemp{N}$ such that:
\begin{enumerate}
\item $S_i = S^{\prime}_j$, and $S_j = S^{\prime}_i$, and
\item $\forall k \in \nid \colon (k \neq i \land k \neq j) \Rightarrow (S_k = S^{\prime}_k \land d_k = d^{\prime}_k)$
\end{enumerate}
\end{defi}

We let $\switchTemp{\transposition}{\mytemp{N}}$ and $\switchGlobalSys{\transposition}{\gsys{N}}$ denote the application of a transposition $\transposition$ to a process template $\mytemp{N}$ and a global transition system $\gsys{N}$, respectively.
Since these definitions are straightforward, we skip them in this chapter.
We prove later that $\switchTempConf{\mytemp{N}}{\transposition} = \mytemp{N}$ and $\switchGlobalConf{\gsys{N}}{\transposition} = \gsys{N}$ (see Lemmas \ref{lem:SymProcTemp} and \ref{lem:SymGlobalSys}).

\newcommand{\lemSymProcTemp}{%
  Let $\mytemp{N} = (\tempconfs{N}, \temptr{N}, \temprel{N}, \tempinitconf{N})$ be a process template with indexes  $\nid$.
  Let $\transposition =\swapij{i}{j}$ be a transposition on $\nid$, and $\transposition_Q$ be a local transposition based on $\transposition$ (from Definition~\ref{def:TransposAndTempConf}).
  The following properties hold:
  \begin{enumerate}
  \item $\transposition_Q$ is a bijection from $\tempconfs{N}$ to itself.
  
  \item The initial state is preserved under $\transposition_Q$, i.e. $\switchTempConf{\tempinitconf{N}}{\transposition} = \tempinitconf{N}$.
  
  \item Let $\tempconf{}, \tempnextconf{} \in \mytemp{N}$ be states such that $\tempconf{} \temprcv{}{S_1}{S_N} \tempnextconf{}$ for some sets of messages $S_1, \ldots, S_N$ in \emph{$\set(\msg)$}. 
  It follows 
  $\switchTempConf{\tempconf{}}{\transposition} \switchRcv{S_1}{S_N}{\transposition} \switchTempConf{\tempnextconf{}}{\transposition} $.
  
  \item Let $\tempconf{}, \tempnextconf{}$ be states in $\mytemp{N}$, and $\mathit{tr} \in \temptr{N}$ be one of send, local computation, crash and stutter transitions such that $\tempconf{} \temptrans{} \tempnextconf{}$.
  Then,  $\switchTempConf{\tempconf{}}{\transposition} \temptrans{} \switchTempConf{\tempnextconf{}}{\transposition} $.

  \end{enumerate}
}

\begin{lem}[Symmetric Process Template] \label{lem:SymProcTemp}
  \lemSymProcTemp 
\end{lem}

\newcommand{\lemSymGlobalSys}{%
Let 
  $\gsys{N} = \left( \gconfs{N}, \gtr{N}, \grel{N}, \ginitconf{N} \right)$ be a global transition system.
  Let $\transposition$ be a transposition on $\nid$,  and $\transposition_C$ be a global transposition  based on $\transposition$ (from Definition~\ref{def:TransposAndGlobalConf}).
  The following properties hold:
  \begin{enumerate}
  \item $\transposition_C$  is a bijection from $\gconfs{N}$ to itself.
  
  \item The initial configuration is preserved under $\transposition_C$, i.e. $\switchGlobalConf{\ginitconf{N}}{\transposition} = \ginitconf{N}$.

  \item Let $\gconf{}$ and $\gnextconf{}$ be configurations  in $\gsys{N}$, and $\mathit{tr} \in \gtr{N}$ be either a internal transition such that $\gconf{} \xrightarrow{\mathit{tr}} \gnextconf{}$.
  It follows  $\switchGlobalConf{\gconf{}}{\transposition} \xrightarrow{\mathit{tr}} \switchGlobalConf{\gnextconf{}}{\transposition} $.
    
  \item Let $\gconf{}$ and $\gnextconf{}$ be configurations  in $\gsys{N}$. If  $\gconf{} \leadsto \gnextconf{}$, then  $\switchGlobalConf{\gconf{}}{\transposition} \leadsto \switchGlobalConf{\gnextconf{}}{\transposition} $.
  \end{enumerate}
}

\begin{lem}[Symmetric Global System] \label{lem:SymGlobalSys}
  \lemSymGlobalSys 
\end{lem}

\subsection{Trace Equivalence of $\gsys{2}$ and $\gsys{N}$ under $AP_{\{1, 2\}}$} \label{subsec:traceequiv}

Let $\gsys{2}$ and $\gsys{N}$ be two global transition systems whose processes follow the same~\classname{} algorithm.
In the following, our goal is to prove Lemma~\ref{lem:TraceEquivTwo} that says $\gsys{2}$ and $\gsys{N}$ are trace equivalent under a set $AP_{\{1, 2\}}$ of predicates which take one of the forms: $Q_1(1), Q_2(2), Q_3(1, 2)$, or $Q_4(2, 1)$.
To do that, we first present two construction techniques: Construction~\ref{const:TempConfProj} to construct a state in $\mytemp{2}$ from a state in $\mytemp{N}$,
and Construction~\ref{const:GlobalConfProj} to construct a global configuration in $\gsys{2}$ from a given global configuration in $\gsys{N}$.
Second, we define trace equivalence under a set $\textit{AP}_{\{1, 2\}}$ of predicates in which every predicate takes one of the forms: $P_1(i)$, or $P_2(j)$, or $P_3(i, j)$, or $P_4(j, i)$.
Our definition of trace equivalence under $\textit{AP}_{\{1, 2\}}$ is extended from the definition of trace equivalence in~\cite{hoare1980model}.
Next, we present two Lemmas~\ref{lem:IDProjPath} and~\ref{lem:IDRefilledPath}. These lemmas are required in the proof of Lemma~\ref{lem:TraceEquivTwo}.

To keep the presentation simple, when the context is clear, we simply write $\mytemp{N}$, instead of fully $\mytemp{N} = \left(\tempconfs{N}, \temptr{N}, \temprel{N}, \tempinitconf{N}\right)$.
We also write $\gsys{N}$, instead of fully $\gsys{N} = \left( \gconfs{N}, \gtr{N}, \grel{N}, \ginitconf{N} \right)$.

\begin{construction}[State Projection] \label{const:TempConfProj}
  Let $\mathcal{A}$ be an arbitrary~\classname{} algorithm.
  Let $\mytemp{N}$ be a process template of $\mathcal{A}$ for some $N \geq 2$, and
  $\tempconf{}^N$ be a process configuration of $\mytemp{N}$.
  We construct a tuple $\tempconf{}^2 = (pc_1, rcvd_1, rcvd_2, v_1, v_2)$ based on $\tempconf{}^N$ and a set $\{1, 2\}$ of process indexes in the following way:
  \begin{enumerate}
    \item $pc_1  = \getPC{\tempconf{}^N}$.
    
    \item For every $i \in \{1, 2\}$, it follows  $rcvd_i = \getBox{\tempconf{}^N}{i}$.
    
    \item For every $i \in \{1, 2\}$, it follows  $v_i = \getDatum{\tempconf{}^N}{i}$.   
  \end{enumerate}  
\end{construction}

\begin{construction}[Configuration Projection] \label{const:GlobalConfProj}
  Let $\mathcal{A}$ be a~\classname{} algorithm.
  Let $\gsys{2}$ and $\gsys{N}$ be two global transition systems of two instances of $\mathcal{A}$ for some $N \geq 2$, and $\gconf{}^N \in \gconfs{N}$ be a global configuration in $\gsys{N}$.
  A tuple 
  \[
  \gconf{}^2 = (s_1, s_2, buf^1_1, buf^2_1, buf^1_2, buf^2_2, act_1, act_2)
  \] 
  is constructed based on the configuration $\gconf{}^N$ and a set $\{1, 2\}$ of indexes  in the following way: 
  \begin{enumerate}
  \item For every  $i \in \{1, 2\}$, a component $s_i$ is constructed from  $\getProc{\gconf{}^N}{i}$ with Construction~\ref{const:TempConfProj} and indexes $\{1, 2\}$.      
  
  \item For every $s, r \in \{1, 2\}$, it follows $buf^r_s = \getBuf{\gconf{}^N}{s}{r}$.

  \item For every process $i \in \{1, 2\}$, it follows $act_i = \getActive{\gconf{}^N}{i} $.
  \end{enumerate}  
\end{construction}

Note that a tuple $\tempconf{}^2$ constructed  with Construction~\ref{const:TempConfProj} is a  state in $\mytemp{2}$, and a tuple $\gconf{}^2$ constructed  with Construction \ref{const:GlobalConfProj} is a configuration in $\gsys{2}$.
We call  $\tempconf{}^2$ (and $\gconf{}^2$) the \emph{index projection} of $\tempconf{}^N$ (and $\gconf{}^N$) on indexes $\{1, 2\}$.   
The following Lemma \ref{lem:IDProjPath} says that Construction~\ref{const:GlobalConfProj} allows us to construct an admissible sequence of global configurations in $\gsys{2}$ based on a given admissible sequence in $\gsys{N}$.
Intuitively, the index projection throws away processes $3..N$ as well as their corresponding messages and buffers.
Moreover, for every $i,j \in \{1, 2\}$, the index projection preserves (i) when process $i$ takes a step, and (ii) what action process $i$ takes at time $t \geq 0$, and (iii) messages between process $i$ and process $j$.
For example, Figure~\ref{fig:3to2} demonstrates an execution in $\mathcal{G}_2$ that is constructed based on a given execution in $\mathcal{G}_3$ with the index projection.

\begin{figure}[tb]
	\centering
	\resizebox{.85\textwidth}{!}{
		\begin{tikzpicture}[x=1cm,y=0.4cm]
		\draw[-stealth] (-3,0)--(5,0) node[right] {};
		
		\draw[-stealth] (-3,1.2)--(5,1.2) node[right] {};
		
		\draw[-stealth] (-3,2.4)--(5,2.4) node[right] {};
		
		\node[circle,fill=black,inner sep=0pt,minimum size=3pt] (a1) at (-2, 2.4) {};
		
		\node[circle,fill=black,inner sep=0pt,minimum size=3pt] (a2) at (-1, 1.2) {};
		
		\node[circle,fill=black,inner sep=0pt,minimum size=3pt] (a3) at (-2.5,1.2) {};
		
		\node[circle,fill=black,inner sep=0pt,minimum size=3pt] (a4) at (-1.5, 0) {};

		\draw[thick,->] (a3) to (a1);
		
		\draw[thick,->] (a3) to (a2);
		
		\draw[thick,->] (a3) to (a4);

		\node[circle,fill=black,inner sep=0pt,minimum size=3pt] (b1) at (-0.5, 2.4) {};
		
		\node[circle,fill=black,inner sep=0pt,minimum size=3pt] (b2) at (0.5, 2.4) {};
		
		\node[circle,fill=black,inner sep=0pt,minimum size=3pt] (b3) at (1.5, 1.2) {};
		
		\node[circle,fill=black,inner sep=0pt,minimum size=3pt] (b4) at (1, 0) {};

		\draw[thick,->] (b1) to (b2);
		
		\draw[thick,->] (b1) to (b3);
		
		\draw[thick,->] (b1) to (b4);
		
		\node[circle,fill=black,inner sep=0pt,minimum size=3pt] (c1) at (3.25, 2.4) {};
		
		\node[circle,fill=black,inner sep=0pt,minimum size=3pt] (c2) at (4, 1.2) {};
		
		\node[circle,fill=black,inner sep=0pt,minimum size=3pt] (c3) at (3.5, 0) {};
		
		\node[circle,fill=black,inner sep=0pt,minimum size=3pt] (c4) at (2.5, 0) {};

		\draw[thick,->] (c4) to (c1);
		
		\draw[thick,->] (c4) to (c2);
		
		\draw[thick,->] (c4) to (c3);
		
		\node[rectangle,inner sep=0pt, draw=gray!75,fill=gray!20,minimum size=4mm] at (-3.8,0) {$p_1$};
		
		\node[rectangle,inner sep=0pt, draw=gray!75,fill=gray!20,minimum size=4mm] at (-3.8,1.2) {$p_2$};
		
		\node[rectangle,inner sep=0pt, draw=gray!75,fill=gray!20,minimum size=4mm] at (-3.8,2.4) {$p_3$};
		
		\node[draw, circle] at (-5,1.2) {$\mathcal{G}_3$};
		
		\draw[dashed, fill=blue!30] (-6,-1)--(5.5,-1) node[right] {};

		\draw[-stealth] (-3,-3.2)--(5,-3.2) node[right] {};
		
		\draw[-stealth] (-3,-2)--(5,-2) node[right] {};

		\node[circle,fill=black,inner sep=0pt,minimum size=3pt] (d2) at (-1, -2) {};
		
		\node[circle,fill=black,inner sep=0pt,minimum size=3pt] (d3) at (-2.5,-2) {};
		
		\node[circle,fill=black,inner sep=0pt,minimum size=3pt] (d4) at (-1.5, -3.2) {};

		\draw[thick,->] (d3) to (d2);
		
		\draw[thick,->] (d3) to (d4);
		
		\node[circle,fill=black,inner sep=0pt,minimum size=3pt] (e2) at (4, -2) {};
		
		\node[circle,fill=black,inner sep=0pt,minimum size=3pt] (e3) at (3.5, -3.2) {};
		
		\node[circle,fill=black,inner sep=0pt,minimum size=3pt] (e4) at (2.5, -3.2) {};

		\draw[thick,->] (e4) to (e2);
		
		\draw[thick,->] (e4) to (e3);

		\node[rectangle,inner sep=0pt, draw=gray!75,fill=gray!20,minimum size=4mm] at (-3.8,-3.2) {$p_1$};
		
		\node[rectangle,inner sep=0pt, draw=gray!75,fill=gray!20,minimum size=4mm] at (-3.8,-2) {$p_2$};

		\node[circle,draw] at (-5,-2.6) {$\mathcal{G}_2$};

		\end{tikzpicture}
	}
	\caption{Given execution in $\mathcal{G}_3$, construct an execution in $\mathcal{G}_2$ by index projection.}  \label{fig:3to2}
\end{figure}

\newcommand{\lemIDProjPath}{%
  Let $\mathcal{A}$ be a~\classname{} algorithm.
  Let $\gsys{2}$ and $\gsys{N}$ be two transition systems such that all processes in $\gsys{2}$ and $\gsys{N}$ follow $\mathcal{A}$, and $N \geq 2$.
  Let $\pi^N = \gconf{0}^N\gconf{1}^N \ldots$ be an admissible sequence of configurations in $\gsys{N}$.
  Let $\pi^2 = \gconf{0}^2\gconf{1}^2 \ldots$ be a sequence of configurations in $\gsys{2}$  such that $\gconf{k}^2$ is the index projection of  $\gconf{k}^N$ on indexes $\{1, 2\}$ for every $k \geq 0$.
  Then, $\pi^2$ is admissible in $\gsys{2}$.
}

\begin{lem} \label{lem:IDProjPath}
  \lemIDProjPath
\end{lem}

\begin{proof}[Sketch of proof]
The proof of Lemma \ref{lem:IDProjPath} is based on the following observations:
\begin{enumerate}
	\item The application of Construction \ref{const:TempConfProj} to an initial template state  of $\mytemp{N}$ constructs an initial template state of $\mytemp{2}$. 
	
	\item Construction \ref{const:TempConfProj} preserves the template transition relation.
	
	\item The application of Construction \ref{const:GlobalConfProj} to an initial global configuration of $\gsys{N}$ constructs an initial global configuration of $\gsys{2}$. 
	
	\item Construction \ref{const:GlobalConfProj} preserves the global transition relation. \qedhere
\end{enumerate}
\end{proof}

Moreover, Lemma~\ref{lem:IDRefilledPath} says that given an admissible sequence $\pi^2 = \gconf{0}^2 \gconf{1}^2 \ldots $ in $\gsys{2}$, there exists an admissible sequence $\pi^N = \gconf{0}^N \gconf{1}^N \ldots $ in $\gsys{N}$ such that $\gconf{i}^2$ is the index projection of $\gconf{i}^N$ on indexes $\{1, 2\}$ for every $0 \leq i$.

\begin{defi}[Trace Equivalence under $\textit{AP}_{\{1, 2\}}$]  \label{def:traceEquivalence}
  Let $\mathcal{A}$ be an arbitrary~\classname{} algorithm.
  Let $\gsys{2} = (\tempconfs{2}, \gtr{2}, \temprel{2}, \tempinitconf{2})$ and $\gsys{N} = (\tempconfs{N}, \gtr{N}, \temprel{N}, \tempinitconf{N}) $ be global transition systems of $\mathcal{A}$ for some $N \geq 2$.  
  Let $\textit{AP}_{\{1, 2\}}$ be a set of predicates that take one of the forms: $P_1(i)$, or $P_2(j)$, or $P_3(i, j)$, or $P_4(j, i)$.
  Let $L \colon \tempconfs{2} \cup \tempconfs{N} \rightarrow 2^{\textit{AP}}$ be an evaluation function. 
  We say that $\gsys{2}$ and $\gsys{N}$ are trace equivalent under $\textit{AP}_{\{1, 2\}}$ if the following constraints hold:
  
  \begin{enumerate}
    \item For every admissible sequence $\pi^2 = \gconf{0}^2 \gconf{1}^2 \ldots $ of configurations in $\gsys{2}$, there exists an admissible sequence of configurations $\pi^N = \gconf{0}^N \gconf{1}^N \ldots $in $\gsys{N}$ such that $L(\gconf{i}^2) = L(\gconf{i}^N)$ for every $i \geq 0$.
    
    \item For every admissible sequence admissible sequence of configurations $\pi^N = \gconf{0}^N \gconf{1}^N \ldots $in $\gsys{N}$, there exists an admissible sequence $\pi^2 = \gconf{0}^2 \gconf{1}^2 \ldots $ of configurations in $\gsys{2}$ such that $L(\gconf{i}^2) = L(\gconf{i}^N)$ for every $i \geq 0$.
    
  \end{enumerate}
 
\end{defi}

\newcommand{\lemIDRefilledPath}{
  Let $\mathcal{A}$ be an arbitrary~\classname{} algorithm.
  Let $\gsys{2}$ and $\gsys{N}$ be global transition systems of $\mathcal{A}$ for some $N \geq 2$.  
  Let $\pi^2 = \gconf{0}^2 \gconf{1}^2 \ldots $ be an admissible sequence of configurations in $\gsys{2}$.
  There exists an admissible sequence $\pi^N = \gconf{0}^N \gconf{1}^N \ldots  $ of configurations in $\gsys{N}$ such that 
  $\gconf{i}^2$ is the index projection of $\gconf{i}^N$ on indexes $\{1, 2\}$ for every $i \geq 0$.  
}

\begin{lem} \label{lem:IDRefilledPath}
\lemIDRefilledPath
\end{lem}

\begin{proof}[Sketch of proof]
We construct an execution $\pi_N$ in $\mathcal{G}_N$ based on $\pi_2$ such that all processes $3..N$ crash from the beginning, and $\pi_2$ is an index projection of $\pi_N$. 
For instance, Figure~\ref{fig:2to3} demonstrates an execution in $\mathcal{G}_3$ that is constructed based on one in $\mathcal{G}_2$.
We have that $\pi_2$ is admissible in $\mathcal{G}_2$. %$\qedhere$
\end{proof}

\begin{figure}[tb]
	%\begin{center}
	\centering
	\resizebox{.85\textwidth}{!}{
		\begin{tikzpicture}[x=1cm,y=0.4cm]

		\draw[dashed, fill=blue!30] (-6,5)--(5.5,5) node[right] {};

		\draw[-stealth] (-3,6)--(5,6) node[right] {};
		
		\draw[-stealth] (-3,7.2)--(5,7.2) node[right] {};

		\node[circle,fill=black,inner sep=0pt,minimum size=3pt] (d2) at (-1, 7.2) {};
		
		\node[circle,fill=black,inner sep=0pt,minimum size=3pt] (d3) at (-2.5, 7.2) {};
		
		\node[circle,fill=black,inner sep=0pt,minimum size=3pt] (d4) at (-1.5, 6) {};

		\draw[thick,->] (d3) to (d2);
		
		\draw[thick,->] (d3) to (d4);
		
		\node[circle,fill=black,inner sep=0pt,minimum size=3pt] (e2) at (4, 7.2) {};
		
		\node[circle,fill=black,inner sep=0pt,minimum size=3pt] (e3) at (3.5, 6) {};
		
		\node[circle,fill=black,inner sep=0pt,minimum size=3pt] (e4) at (2.5, 6) {};

		\draw[thick,->] (e4) to (e2);
		
		\draw[thick,->] (e4) to (e3);

		\node[rectangle,inner sep=0pt, draw=gray!75,fill=gray!20,minimum size=4mm] at (-3.8, 6) {$p_1$};
		
		\node[rectangle,inner sep=0pt, draw=gray!75,fill=gray!20,minimum size=4mm] at (-3.8,7.2) {$p_2$};

		\node[circle,draw] at (-5,6.6) {$\mathcal{G}_2$};

		\draw[-stealth] (-3,1.4)--(5,1.4) node[right] {};
		
		\draw[-stealth] (-3,2.6)--(5,2.6) node[right] {};
		
		\draw[-stealth] (-3,3.8)--(5,3.8) node[right] {};
		
		\node[inner sep=0pt,minimum size=3pt] (a1) at (-1.8, 3.8) {};
		
		\node[circle,fill=black,inner sep=0pt,minimum size=3pt] (a2) at (-1, 2.6) {};
		
		\node[circle,fill=black,inner sep=0pt,minimum size=3pt] (a3) at (-2.5,2.6) {};
		
		\node[circle,fill=black,inner sep=0pt,minimum size=3pt] (a4) at (-1.5, 1.4) {};

		\draw[dashed,->] (a3) to (a1);
		
		\draw[thick,->] (a3) to (a2);
		
		\draw[thick,->] (a3) to (a4);

		\node[inner sep=0pt,minimum size=3pt] (c1) at (3.1, 3.8) {};
		
		\node[circle,fill=black,inner sep=0pt,minimum size=3pt] (c2) at (4, 2.6) {};
		
		\node[circle,fill=black,inner sep=0pt,minimum size=3pt] (c3) at (3.5, 1.4) {};
		
		\node[circle,fill=black,inner sep=0pt,minimum size=3pt] (c4) at (2.5, 1.4) {};

		\draw[dashed,->] (c4) to (c1);
		
		\draw[thick,->] (c4) to (c2);
		
		\draw[thick,->] (c4) to (c3);
		
		\node[rectangle,inner sep=0pt, draw=gray!75,fill=gray!20,minimum size=4mm] at (-3.8,1.4) {$p_1$};
		
		\node[rectangle,inner sep=0pt, draw=gray!75,fill=gray!20,minimum size=4mm] at (-3.8,2.6) {$p_2$};
		
		\node[rectangle,inner sep=0pt, draw=gray!75,fill=gray!20,minimum size=4mm] at (-3.8,3.8) {$p_3$};
		
		\node[circle,draw] at (-5,2.6) {$\mathcal{G}_3$};
		
		\node[anchor=south west,inner sep=0] at (-3.2,3.6) {\includegraphics[width=0.2cm]{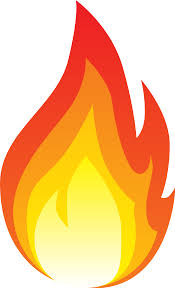}};

		\node[anchor=south west,inner sep=0] at (-1.8,3.6) {\includegraphics[width=0.2cm]{./img/fire.jpeg}};
		
		\node[anchor=south west,inner sep=0] at (3.1,3.6) {\includegraphics[width=0.2cm]{./img/fire.jpeg}};

		\end{tikzpicture}
	}
	\caption{Construct an execution in $\mathcal{G}_3$ based on a given execution in $\mathcal{G}_2$.}
	\label{fig:2to3}
\end{figure}

\newcommand{\lemTraceEquivTwo}{
  Let $\mathcal{A}$ be a~\classname{} algorithm. Let $\gsys{2}$ and $\gsys{N}$ be its instances for some $N \geq 2$.
  Let $\textit{AP}_{\{1, 2\}}$ be a set of predicates that take one of the forms: $P_1(1)$, $P_2(2)$,  $P_3(1, 2)$ or $P_4(2, 1)$.
  It follows that $\gsys{2}$ and $\gsys{N}$ are trace equivalent under $\textit{AP}_{\{1, 2\}}$.
}

\begin{lem} \label{lem:TraceEquivTwo}
  \lemTraceEquivTwo
\end{lem}

\begin{proof}[Sketch of proof]
  The proof of Lemma~\ref{lem:TraceEquivTwo} is based on Definition \ref{def:traceEquivalence}, Lemma~\ref{lem:IDProjPath}, and Lemma~\ref{lem:IDRefilledPath}. %$\qedhere$
\end{proof}

\section{Detailed Proofs for Cutoff Results in the Model of the Global Distributed Systems}
\label{sec:netys20-proofs}
In this section, we present the detailed proofs for Theorems~\ref{thm:NCutoffOneIndexNetysA} and \ref{thm:NCutoffTwoIndexesNetysA}.
In Sections \ref{app:SymProcTemp} and \ref{app:SymGlobalSys} we prove that every transposition on a set of process indexes $\nid$ preserves the
structure of the process template $\mytemp{N}$ and the
structure of the global transition system $\gsys{N}$ for every $N \geq 1$, respectively.
In Section~\ref{app:TraceEquivalenceTwo} we show that $\gsys{2}$ and $\gsys{N}$ are trace equivalent under $AP_{\{1, 2\}}$.
Next, we prove that $\gsys{1}$ and $\gsys{N}$ are trace equivalent under $AP_{\{1\}}$ in Section~\ref{app:TraceEquivalenceOne}.
Then, the detailed proofs for Theorems~\ref{thm:NCutoffOneIndexNetysA} and \ref{thm:NCutoffTwoIndexesNetysA} is presented in Section~\ref{app:CutoffOne}.
Finally, we discuss why we can verify the strong completeness property of the failure detector of~\cite{CT96} under synchrony by model checking instances of size 2 by applying our cutoff results.

\subsection{Transpositions and Process Templates} \label{app:SymProcTemp}

The proof of Lemma~\ref{lem:SymProcTemp} requires the following propositions: 
\begin{itemize}
  \item Given a transposition $\transposition$, Proposition~\ref{prop:TranspositionAsBijectionInTemp} says that a function $\transposition_Q$, which refers to the application of $\transposition$ to a state in $\tempconfs{N}$,  is a bijection from $\tempconfs{N}$ to itself.
  
  \item Proposition~\ref{prop:SymInitTempConf} says that $\transposition_Q$ has no effect on the initial template state $\tempinitconf{N}$.
  
  \item Propositions \ref{prop:TranspositionTempReceive} and \ref{prop:TranspositionOtherTempActions} describe the relationship between transpositions and template transitions.
\end{itemize}

\recalllemma{lem:SymProcTemp}{\lemSymProcTemp{}}

\begin{proof} We have: point 1 holds by Proposition~\ref{prop:TranspositionAsBijectionInTemp},  
  and point 2 holds by Proposition~\ref{prop:SymInitTempConf}, 
  and point 3 holds by Proposition~\ref{prop:TranspositionTempReceive}, 
  and point 4 holds by Proposition~\ref{prop:TranspositionOtherTempActions}.
\end{proof}

\newcommand{\propTranspositionAsBijectionInTemp}{
  Let $\mytemp{N} = (\tempconfs{N}, \temptr{N}, \temprel{N}, \tempinitconf{N})$ be a process template with indexes  $\nid$.
  Let $\transposition =\swapij{i}{j}$ be a transposition on $\nid$,  and $\transposition_Q$ be a local transposition based on $\transposition$ (from Definition~\ref{def:TransposAndTempConf}). 
  Then, $\transposition_Q$  is a bijection from $\tempconfs{N}$ to itself.
}

\begin{prop} \label{prop:TranspositionAsBijectionInTemp}  
  \propTranspositionAsBijectionInTemp
\end{prop}

\begin{proof}
Since two transpositions $\swapij{i}{j}$ and $\swapij{j}{i}$ are equivalent, we assume  $i < j$.
To show that  $\transposition_Q$ is a bijection from $\tempconfs{N}$ to itself, we prove that the following properties hold:
\begin{enumerate}[nosep]
\item For every template state $\tempnextconf{} \in \tempconfs{N}$, there exists a template configuration $\tempconf{} \in \tempconfs{N}$ such that $\switchTempConf{\tempconf{}}{\transposition} = \tempnextconf{}$.
\item For every pair of states $\tempconf{1}, \tempconf{2} \in \tempconfs{N}$, if $\switchTempConf{\tempconf{1}}{\transposition} = \switchTempConf{\tempconf{2}}{\transposition}$, then $\tempconf{1} = \tempconf{2}$.
\end{enumerate}

We first show that Point 1 holds. Assume that $\tempnextconf{}$ is a following tuple 
\[ \tempnextconf{} = (\ell, S_1, \ldots, S_i, \ldots, S_j, \ldots, S_N, d_1, \ldots, d_i, \ldots, d_j, \ldots, d_N) \]
where $\ell \in \loc, S_i \in \set(\msg), d_i \in \tempdatatype$ for every $i \in \nid$.
Let  $\tempconf{}$ be the following tuple
\[ \tempconf{} = (\ell, S_1, \ldots, S_j, \ldots, S_i, \ldots, S_N, d_1, \ldots, d_j, \ldots, d_i, \ldots, d_N) \]
where $S_k = S^{\prime}_k \land d_k = d^{\prime}_k$ for every $k \in \nid \setminus \{i, j\}$.
By Definition~\ref{def:TransposAndTempConf}, we have $\switchTempConf{\tempconf{}}{\transposition} = \tempnextconf{}$.
Moreover, by the definition of a process template in Section~\ref{sec:temp-model}, it follows $\tempconf{} \in \tempconfs{N}$.

We now focus on Point 2. By definition of the application of a process--index transposition to a template state, it is easy to check that $\transposition_Q( (\switchTempConf{\tempconf{}}{\transposition})) = \tempconf{}$ for every $\tempconf{} \in \tempconfs{N}$. 
It follows that 
$
\tempconf{1} = \transposition_Q((\switchTempConf{\tempconf{1}}{\transposition})) = \transposition_Q((\switchTempConf{\tempconf{2}}{\transposition})) = \tempconf{2} 
$ since $\switchTempConf{\tempconf{1}}{\transposition} = \switchTempConf{\tempconf{2}}{\transposition}$.

Therefore, Proposition~\ref{prop:TranspositionAsBijectionInTemp} holds.
%$\qedhere$
\end{proof}

\newcommand{\propSymInitTempConf}{
  Let $\mytemp{N} = (\tempconfs{N}, \temptr{N}, \temprel{N}, \tempinitconf{N})$ be a process template.
  Let $\transposition =\swapij{i}{j}$ be a transposition on $\nid$, and
  $\transposition_Q$ be a local transposition based on $\transposition$ (from Definition~\ref{def:TransposAndTempConf}).
  It follows that $\switchTempConf{\tempinitconf{N}}{\transposition} = \tempinitconf{N}$.
}

\begin{prop} \label{prop:SymInitTempConf}
  \propSymInitTempConf
\end{prop}

\begin{proof}
By definition of $\tempinitconf{N}$ in Section~\ref{sec:temp-model}, we have $\getBox{\tempinitconf{N}}{i} = \getBox{\tempinitconf{N}}{j}$ and $\getDatum{\tempinitconf{N}}{i} = \getDatum{\tempinitconf{N}}{j}$.
It immediately follows  $\switchTempConf{\tempinitconf{N}}{\transposition} = \tempinitconf{N}$. %$\qedhere$
\end{proof}

\newcommand{\propTranspositionTempReceive}{
  Let $\mytemp{N} = (\tempconfs{N}, \temptr{N}, \temprel{N}, \tempinitconf{N})$ be a process template with indexes $\nid$.
  Let $\tempconf{0}$ and $\tempconf{1}$ be states in $\mytemp{N}$ such that $\tempconf{0} \temprcv{}{S_1}{S_N} \tempconf{1}$ for some sets of messages: $S_1, \ldots, S_N \subseteq \emph{\set(\msg)}$. 
  Let $\transposition = \swapij{i}{j}$ be a transposition on $\nid$, and $\transposition_R$ be a receive-transition transposition based on $\transposition$ (from Definition~\ref{def:TransposeAndRcv}).
  It follows $\switchTempConf{\tempconf{0}}{\transposition} 
  \switchRcv{S_1}{S_N}{\transposition}
  \switchTempConf{\tempconf{1}}{\transposition}
  $.  
}

\begin{prop} \label{prop:TranspositionTempReceive}
  \propTranspositionTempReceive
\end{prop}

\begin{proof}
  We prove that all Constraints (a)--(c) between two states $\switchTempConf{\tempconf{0}}{\transposition}$ and $\switchTempConf{\tempconf{1}}{\transposition}$  in the transition $\csnd$ defined in Section~\ref{sec:temp-model} hold.
  First, we focus on Constraint (a). We have $\getPC{\switchTempConf{\tempconf{1}}{\transposition}} 
   =   \getPC{\tempconf{1}}$ by Definition~\ref{def:TransposAndTempConf}.
  We have
  $\getPC{\tempconf{1}}  = \nextLoc(\getPC{\tempconf{0}})$ by the semantics of $\crcv(S_1, \ldots, S_N)$ in Section~\ref{sec:temp-model}.
  We have $\nextLoc(\getPC{\tempconf{0}}) = \nextLoc(\getPC{\switchTempConf{\tempconf{0}}{\transposition}})$ by Definition~\ref{def:TransposAndTempConf}.
  It follows $
  \getPC{\switchTempConf{\tempconf{1}}{\transposition}} = \nextLoc(\getPC{\switchTempConf{\tempconf{0}}{\transposition}})$.
  Moreover, we have
  \begin{alignat*}{3}
  \{m\} & = \genMsg(\getPC{\tempconf{0}}) && \quad \text{(by the semantics of $\crcv$ in Section~\ref{sec:temp-model})} \\
  & = \genMsg(\getPC{\switchTempConf{\tempconf{0}}{\transposition}}) && \quad \text{(by Definition~\ref{def:TransposAndTempConf}})               
  \end{alignat*}
  Hence, Constraint (a) holds.  
  
  Now we focus on Constraint~(b). 
  By Definition~\ref{def:TransposAndTempConf},  we have  
  $\getBox{\switchTempConf{\tempconf{0}}{\transposition}}{k} = \getBox{\tempconf{0}}{k}$ and
  $\getBox{\switchTempConf{\tempconf{1}}{\transposition}}{k} = \getBox{\tempconf{1}}{k}$ for every $k \in \nid \setminus \{i, j\}$.
  We have
  $\getBox{\tempconf{1}}{k} = S_k \cup \getBox{\tempconf{0}}{k}$ by the semantics of $\crcv(S_1, \ldots, S_N)$ in Section~\ref{sec:temp-model}.
  It follows $\getBox{\switchTempConf{\tempconf{1}}{\transposition}}{k} = S_k \cup \getBox{\switchTempConf{\tempconf{0}}{\transposition}}{k}$ for every $k \in \nid \setminus \{i, j\}$.
  Now we focus on $\getBox{\switchTempConf{\tempconf{1}}{\transposition}}{i}$. 
  We have $\getBox{\switchTempConf{\tempconf{1}}{\transposition}}{i} = \getBox{\tempconf{1}}{j}$ and $\getBox{\switchTempConf{\tempconf{0}}{\transposition}}{i} = \getBox{\tempconf{0}}{j}$ by Definition~\ref{def:TransposAndTempConf}.
  Since $\getBox{\tempconf{1}}{j} = \getBox{\tempconf{0}}{j} \cup S_j$, it follows $\getBox{\switchTempConf{\tempconf{1}}{\transposition}}{i} = \getBox{\switchTempConf{\tempconf{0}}{\transposition}}{i} \cup S_j$.
  By similar arguments, we have $\getBox{\switchTempConf{\tempconf{1}}{\transposition}}{j} = \getBox{\switchTempConf{\tempconf{0}}{\transposition}}{j} \cup S_i$.
  Hence, Constraint (b) holds.
  
  Now we focus on Constraint (c).  
  By similar arguments in the proof of Constraint (b), for every $k \in \nid \setminus \{i, j\}$,  we have    
  \[
  \getDatum{ \switchTempConf{\tempconf{1}}{\transposition} }{k} =  
  \nextVar \big( \getPC{\switchTempConf{\tempconf{0}}{\transposition}},  
  S_k,
  \getDatum{\switchTempConf{\tempconf{0}}{\transposition}}{k}\big) 
  \]

  Now we focus on $ \getDatum{ \switchTempConf{\tempconf{1}}{\transposition} }{i}$.
  We have $\getDatum{ \switchTempConf{\tempconf{1}}{\transposition} }{i} 
  = \getDatum{ \tempconf{1} }{j}$ by Definition~\ref{def:TransposAndTempConf}.
  By the semantics of $\crcv(S_1, \ldots, S_N)$ in Section~\ref{sec:temp-model}, it follows that
  $\getDatum{ \tempconf{1} }{j} =  \nextVar \big( \getPC{\tempconf{0}}, 
  S_j,
  \getDatum{\tempconf{0}}{j} \big) $.
  By Definition~\ref{def:TransposAndTempConf}, we have 
  
  $\begin{aligned}
  & \getDatum{ \switchTempConf{\tempconf{1}}{\transposition} }{i} \\
    \qquad\: \quad = \;   & \getDatum{ \tempconf{1} }{j} \\
  \qquad\: \quad = \;  & \nextVar \big( \getPC{\tempconf{0}},  
  S_j,
  \getDatum{\tempconf{0}}{j} \big) \\
  \qquad\: \quad = \; & \nextVar \big( \getPC{\switchTempConf{\tempconf{0}}{\transposition}}, 
  S_j,
  \getDatum{\switchTempConf{\tempconf{0}}{\transposition}}{i}\big) \\
  \end{aligned}$
  
  Moreover, by similar arguments, we have 
  \[
  \getDatum{ \switchTempConf{\tempconf{1}}{\transposition} }{j} = \nextVar \big( \getPC{\switchTempConf{\tempconf{0}}{\transposition}},  
  S_i,
  \getDatum{\switchTempConf{\tempconf{0}}{\transposition}}{i}\big)
  \]
  
  Constraint (c) holds. 
  Hence, we have
 $\switchTempConf{\tempconf{0}}{\transposition} \switchRcv{S_1}{S_N}{\transposition} \switchTempConf{\tempconf{1}}{\transposition}$. %$\qedhere$
\end{proof}

\newcommand{\propTranspositionOtherTempActions}{
  Let $\mytemp{N} = (\tempconfs{N}, \temptr{N}, \temprel{N}, \tempinitconf{N})$ be a process template with indexes $\nid$.
  Let $\tempconf{}$ and $\tempnextconf{}$ be states in $\mytemp{N}$, and $\mathit{tr} \in \temptr{N}$ be a transition such that $\tempconf{} \temptrans{} \tempnextconf{}$ and $\mathit{tr}$ refers to a task without a receive action.
  Let $\transposition = \swapij{i}{j}$ be a transposition on $\nid$, and $\transposition_Q$ be a local transposition based on $\transposition$ (from Definition~\ref{def:TransposAndTempConf}).
  It follows  $\switchTempConf{\tempconf{}}{\transposition} \temptrans{} \switchTempConf{\tempnextconf{}}{\transposition} $.
}

\begin{prop} \label{prop:TranspositionOtherTempActions}
  \propTranspositionOtherTempActions
\end{prop}

\begin{proof}
  We prove Proposition~\ref{prop:TranspositionOtherTempActions} by case distinction.  
  \begin{itemize}
  \item \emph{Case $\tempconf{0} \tempsnd{}{m} \tempconf{1}$.} 
  By similar arguments in the proof of Proposition~\ref{prop:TranspositionTempReceive}, it follows  $\getPC{\switchTempConf{\tempconf{1}}{\transposition}}  =  \nextLoc(\getPC{\switchTempConf{\tempconf{0}}{\transposition}})$ and $\{m\} = \genMsg(\getPC{\switchTempConf{\tempconf{0}}{\transposition}})$.
  Constraint (a) holds.
  By Definition~\ref{def:TransposAndTempConf},  for every $k \in \nid \setminus \{i, j\}$,  we have  
  $\getBox{\switchTempConf{\tempconf{0}}{\transposition}}{k} = \getBox{\tempconf{0}}{k}$ and
  $\getBox{\switchTempConf{\tempconf{1}}{\transposition}}{k} = \getBox{\tempconf{1}}{k}$.
  Hence, it follows $\getBox{\switchTempConf{\tempconf{1}}{\transposition}}{k} =  \getBox{\switchTempConf{\tempconf{0}}{\transposition}}{k}$ for every $k \in \nid \setminus \{i, j\}$.  
  Now we focus on $\getBox{\switchTempConf{\tempconf{1}}{\transposition}}{i}$. 
  We have $\getBox{\switchTempConf{\tempconf{1}}{\transposition}}{i} = \getBox{\tempconf{1}}{j}$ and $\getBox{\switchTempConf{\tempconf{0}}{\transposition}}{i} = \getBox{\tempconf{0}}{j}$ by Definition~\ref{def:TransposAndTempConf}.
  Since $\getBox{\tempconf{1}}{j} = \getBox{\tempconf{0}}{j}$, it follows that $\getBox{\switchTempConf{\tempconf{1}}{\transposition}}{i} = \getBox{\switchTempConf{\tempconf{0}}{\transposition}}{i}$.  
  By similar arguments, we have $\getBox{\switchTempConf{\tempconf{1}}{\transposition}}{j} = \getBox{\switchTempConf{\tempconf{0}}{\transposition}}{j}$.
  Constraint (b) holds.
  By similar arguments in the proof of Proposition~\ref{prop:TranspositionTempReceive}, for every $k \in \nid$, we have
  \[
   \getDatum{ \switchTempConf{\tempconf{1}}{\transposition} }{k}  = \nextVar \big( \getPC{\switchTempConf{\tempconf{0}}{\transposition}},  
  \emptyset,
   \getDatum{\switchTempConf{\tempconf{0}}{\transposition}}{k}\big)
  \]
  Constraint (c) holds.
  It follows $\switchTempConf{\tempconf{1}}{\transposition} \tempsnd{}{m} \switchTempConf{\tempnextconf{1}}{\transposition}$.
  
  \item \emph{Case $\tempconf{0} \tempcomp{} \tempconf{1}$.} Similar to the case of $\csnd$.
  
  \item \emph{Case $\tempconf{0} \tempcrash{} \tempconf{1}$.} We have $\getPC{\switchTempConf{\tempconf{1}}{\transposition}} = \getPC{\tempconf{1}}$ and $\getPC{\switchTempConf{\tempconf{0}}{\transposition}} = \getPC{\tempconf{0}}$ by Definition~\ref{def:TransposAndTempConf}. 
  By the transitions' assumptions, we have $\getPC{\switchTempConf{\tempconf{1}}{\transposition}} = \loccrash$ and $\getPC{\switchTempConf{\tempconf{0}}{\transposition}}  \neq \loccrash$.
  By similar arguments in the case of $\csnd$, it follows $\forall k \in \nid \colon \getBox{\switchTempConf{\tempconf{1}}{\transposition}}{k} = \getBox{\switchTempConf{\tempconf{0}}{\transposition}}{k} \land \getDatum{\switchTempConf{\tempconf{1}}{\transposition}}{k} = \getDatum{\switchTempConf{\tempconf{0}}{\transposition}}{k}$.
  Hence, it holds
  $\switchTempConf{\tempconf{0}}{\transposition} \tempcrash{} \switchTempConf{\tempconf{1}}{\transposition}$.

  \item \emph{Case $\tempconf{} \tempstutter{} \tempnextconf{}$.} Similar to the case of $\csnd$. 
  \end{itemize}
  Hence, Proposition~\ref{prop:TranspositionOtherTempActions} holds. %$\qedhere$
\end{proof}  

\subsection{Transpositions and Global Systems} \label{app:SymGlobalSys}
The proof strategy of Lemma~\ref{lem:SymGlobalSys} is similar to the one of Lemma~\ref{lem:SymProcTemp}, and the proof of Lemma~\ref{lem:SymGlobalSys} requires the following Propositions~\ref{prop:TranspositionAsBijectionInGlobal}, \ref{prop:SymInitGlobalConf}, \ref{prop:SymGlobalSmallActions}
and \ref{prop:SymGlobalTransitions}.

\recalllemma{lem:SymGlobalSys}{\lemSymGlobalSys}

\begin{proof} We have: point 1 holds by Proposition~\ref{prop:TranspositionAsBijectionInGlobal},  
  and point 2 holds by Proposition~\ref{prop:SymInitGlobalConf}, 
  and point 3 holds by Proposition~\ref{prop:SymGlobalSmallActions},
  and point 4 holds by Proposition~\ref{prop:SymGlobalTransitions}.
\end{proof}

\newcommand{\propTranspositionAsBijectionInGlobal}{
  Let $\gsys{N} = \left( \gconfs{N}, \gtr{N}, \grel{N}, \ginitconf{N} \right)$ be a global transition system with indexes $\nid$.
  Let $\transposition$ be a process--index transposition on $\nid$, and $\transposition_C$ be a global transposition  based on $\transposition$ (from Definition~\ref{def:TransposAndGlobalConf}).
  Then, $\transposition_C$ is a bijection from $\gconfs{N}$ to itself.
}

\begin{prop} \label{prop:TranspositionAsBijectionInGlobal}  
  \propTranspositionAsBijectionInGlobal
\end{prop}

\begin{proof}
  By applying similar arguments in the proof of Proposition~\ref{prop:TranspositionAsBijectionInTemp}.
\end{proof}

\newcommand{\propSymInitGlobalConf}{
  Let $\gsys{N} = \left( \gconfs{N}, \gtr{N}, \grel{N}, \ginitconf{N} \right)$ be a global transition system with indexes $\nid$.
  Let $\transposition$ be a process--index transposition on $\nid$, and
  $\transposition_C$ be a global transposition based on $\transposition$ (from Definition~\ref{def:TransposAndGlobalConf}).
  It follows that $\switchGlobalConf{\ginitconf{N}}{\transposition} = \ginitconf{N}$.
}

\begin{prop} \label{prop:SymInitGlobalConf}
  \propSymInitGlobalConf
\end{prop}

\begin{proof}
  By applying similar arguments in the proof of Proposition~\ref{prop:SymInitTempConf}.
\end{proof}

\newcommand{\propSymGlobalSmallActions}{
  Let $\gsys{N} = \left( \gconfs{N}, \gtr{N}, \grel{N}, \ginitconf{N} \right)$ be a global transition system with indexes $\nid$ and a process template $\mytemp{N} = (\tempconfs{N}, \temptr{N}, \temprel{N}, \tempinitconf{N})$.
  Let $\transposition$ be a process--index transposition on $\nid$, and
  $\transposition_C$ be a global transposition based on $\transposition$ (from Definition~\ref{def:TransposAndGlobalConf}).
  Let $\gconf{}$ and $\gnextconf{}$  be configurations in $\gsys{N}$, and $\mathit{tr} \in \gtr{N}$ be  an internal transition such that $\gconf{} \xrightarrow{\mathit{tr}} \gnextconf{}$.
  Then, $\switchGlobalConf{\gconf{}}{\transposition} \xrightarrow{\mathit{tr}} \switchGlobalConf{\gnextconf{}}{\transposition} $. 
}

\begin{prop} \label{prop:SymGlobalSmallActions}
  \propSymGlobalSmallActions
\end{prop}

\begin{proof}
  We prove Proposition~\ref{prop:SymGlobalSmallActions} by case distinction.
  
  \begin{enumerate}
  \item \emph{Sub-round Schedule.}   
  We prove that all Constraints (a)--(e) for the sub-round Schedule hold as follows.
  
  First, we focus on Constraint (a).
  By Proposition~\ref{prop:TranspositionAsBijectionInGlobal}, both $\switchGlobalConf{\gconf{}}{\transposition}$ and $\switchGlobalConf{\gnextconf{}}{\transposition}$ are configurations in $\gsys{N}$. 
  By Definition~\ref{def:TransposAndGlobalConf}, for every $i \in \nid$, we have 
  $ \getActive{\switchGlobalConf{\gconf{}}{\transposition}}{\switchIndex{\transposition}{i}} =\getActive{\gconf{}}{i}$.
  We have $\lnot \getActive{\gconf{}}{i}$ by the semantics of the sub-round Schedule in Section~\ref{sec:global-model}.
  It follows $\lnot \getActive{\switchGlobalConf{\gconf{}}{\transposition}}{\switchIndex{\transposition}{i}}$.
  Hence, the sub-round Schedule can start with a configuration $\switchGlobalConf{\gconf{0}}{\transposition}$. Constraint(a) holds.
    
  Now we focus on Constraint (b) by examining process transitions.     
  For every $i \in \nid$,  by Lemma~\ref{lem:SymProcTemp}, we have
    \begin{align*}
    \getProc{\gconf{}}{i} \tempstutter \getProc{\gnextconf{}}{i} 
    & \Rightarrow
    \switchTempConf{\getProc{\gconf{}}{i}}{\transposition} \tempstutter \switchTempConf{\getProc{\gnextconf{}}{i} }{\transposition} 
    \end{align*}
    By Definition~\ref{def:TransposAndGlobalConf}, it follows  $\switchTempConf{\getProc{\gconf{}}{i}}{\transposition} = \getProc{\switchGlobalConf{\gconf{}}{\transposition}}{\switchIndex{\transposition}{i}}$ and $\switchTempConf{\getProc{\gnextconf{}}{i} }{\transposition} = \getProc{\switchGlobalConf{\gnextconf{}}{\transposition}}{\switchIndex{\transposition}{i}}$.
    Hence, it follows
    \begin{align*}
    \getProc{\gconf{}}{i} & \tempstutter \getProc{\gnextconf{}}{i} \\
    & \Rightarrow \getProc{\switchGlobalConf{\gconf{}}{\transposition}}{\switchIndex{\transposition}{i}} 
    \tempstutter \getProc{\switchGlobalConf{\gnextconf{}}{\transposition}}{\switchIndex{\transposition}{i}} 
    \end{align*}
    By similar arguments, we have
    \begin{align*}
    \getProc{\gconf{}}{i} & \tempcrash{} \getProc{\gnextconf{}}{k} \\ 
    & \Rightarrow \getProc{\switchGlobalConf{\gconf{}}{\transposition}}{\switchIndex{\transposition}{i}} 
    \tempcrash{} \getProc{\switchGlobalConf{\gnextconf{}}{\transposition}}{\switchIndex{\transposition}{k}}
    \end{align*}
    Hence, every process makes either a crash transition or a stuttering step from a configuration $\switchGlobalConf{\gconf{}}{\transposition}$ to a configuration $\switchGlobalConf{\gnextconf{}}{\transposition}$. 
    Constraint (b) holds.
    
    We now focus on Constraint (c) by examining control components of crashed processes.
    Assume that $\getPC{\getProc{\gnextconf{}}{r}} = \loccrash$ for some $i \in \nid$.
    By the semantics of the sub-round Schedule in Section~\ref{sec:global-model}, we have $\lnot \getActive{\gnextconf{}}{i} $.
    By Definition~\ref{def:TransposAndGlobalConf}, it follows 
    $\lnot  \getActive{\switchGlobalConf{\gconf{}}{\transposition}}{\switchIndex{\transposition}{i}} \land \lnot  \getActive{\gnextconf{}}{i}$. 
    Constraint (c) hodls.
    
    Now we focus on Constraint (d) by examining incoming message buffers to correct processes.
    By Definition~\ref{def:TransposAndGlobalConf}, we have $\getBuf{\switchGlobalConf{\gnextconf{}}{\transposition}}{\switchIndex{\transposition}{s}}{\switchIndex{\transposition}{r}}  = \getBuf{\gnextconf{}}{s}{r}$ for every $s, r \in \nid$.     
    By the semantic of the sub-round Schedule in Section~\ref{sec:global-model}, if $\getPC{\getProc{\gnextconf{}}{r}} \neq \loccrash$, 
    , then $\getBuf{\gnextconf{}}{s}{r} = \getBuf{\gconf{}}{s}{r}$.
    By Definition~\ref{def:TransposAndGlobalConf}, we have
    $ \getBuf{\gconf{}}{s}{r} = \getBuf{\switchGlobalConf{\gconf{}}{\transposition}}{\switchIndex{\transposition}{s}}{\switchIndex{\transposition}{r}}$. 
    Hence, Constraint (d) holds since  $
    \getBuf{\switchGlobalConf{\gnextconf{}}{\transposition}}{\switchIndex{\transposition}{s}}{\switchIndex{\transposition}{r}} = \getBuf{\switchGlobalConf{\gconf{}}{\transposition}}{\switchIndex{\transposition}{s}}{\switchIndex{\transposition}{r}}  
    $
    
    Now we focus on Constraint (e) by 
    examining 
    incoming message buffers to a crashed process.
    Let $r$ be an index in $\nid$ such that $\getPC{\getProc{\gnextconf{}}{r}} = \loccrash$.
%    If $\getPC{\getProc{\gnextconf{}}{r}} = \loccrash$, then
%    $   \getBuf{\switchGlobalConf{\gnextconf{}}{\transposition}}{\switchIndex{\transposition}{s}}{\switchIndex{\transposition}{r}}  = \getBuf{\gnextconf{}}{s}{r}
%    $ by similar arguments in the above case of $\getPC{\getProc{\gnextconf{}}{r}} \neq \loccrash$. 
    By similar arguments in the above case of $\getPC{\getProc{\gnextconf{}}{r}} \neq \loccrash$,
    if $\getPC{\getProc{\gnextconf{}}{r}} = \loccrash$, then
    $   \getBuf{\switchGlobalConf{\gnextconf{}}{\transposition}}{\switchIndex{\transposition}{s}}{\switchIndex{\transposition}{r}}  = \getBuf{\gnextconf{}}{s}{r}
    $. 
    By the semantics of the sub-round Schedule in Section~\ref{sec:global-model}, we have
    $\getBuf{\gnextconf{}}{s}{r}  = \emptyset$.
    It follows 
    $\getBuf{\switchGlobalConf{\gnextconf{}}{\transposition}}{\switchIndex{\transposition}{s}}{\switchIndex{\transposition}{r}}  = \emptyset$. 
    Therefore, Constraint (e) holds.
        
    It implies $\switchGlobalConf{\gconf{}}{\transposition} \transsched \switchGlobalConf{\gnextconf{}}{\transposition}$.
    
    \item \emph{Sub-round Send.}
    We prove that all Constraints (a)--(c) for the sub-round Send  hold as follows. 
    By Definition~\ref{def:TransposAndGlobalConf}, we have 
    \begin{alignat*}{2}
     \getActive{\switchGlobalConf{\gconf{}}{\transposition}}{\switchIndex{\transposition}{i}} & = \getActive{\gconf{}}{i} \\
     \getPC{\getProc{\switchGlobalConf{\gconf{}}{\transposition}}{\switchIndex{\transposition}{i}}}  & = \getPC{\switchTempConf{\getProc{\gconf{}}{i}}{\transposition}} 
    \end{alignat*}
    By Definition~\ref{def:TransposAndTempConf}, we have
    $\getPC{\switchTempConf{\getProc{\gconf{}}{i}}{\transposition}} = \getPC{\getProc{\gconf{}}{i}}$
    for every $i \in \nid$.
    Hence, we have $\getPC{\getProc{\switchGlobalConf{\gconf{}}{\transposition}}{\switchIndex{\transposition}{i}}}  = \getPC{\getProc{\gconf{}}{i}}$.
    For every $i \in \nid$, we have $ \CanRun(\gconf{}, i, \locsnd) \Leftrightarrow \CanRun(\switchGlobalConf{\gconf{}}{\transposition}, \switchIndex{\transposition}{i}, \locsnd)$ 
	by the definition of $\CanRun$ in Section~\ref{sec:global-model}. 
    It implies that a process $\getProc{\gconf{}}{s}$ is enabled in this sub-round if and only if a process $\getProc{\switchGlobalConf{\gconf{}}{\transposition}}{\switchIndex{\transposition}{s}}$ is enabled in this sub-round for every $s \in \nid$.

    Now we focus on Constraint (a) by examining processes which are not enabled  in this sub-round Send.
    Let $i$ be an arbitrary index in $\nid$ such that $\lnot \CanRun(\gconf{}, i, \locsnd)$.
    By the semantics of the sub-round Send in Section~\ref{sec:global-model}, it follows that $\getProc{\gconf{}}{i} \transstutter \getProc{\gnextconf{}}{i}$.
    By Definition~\ref{def:TransposAndGlobalConf}, we have 
    \[
    \CanRun(\switchGlobalConf{\gconf{}}{\transposition}, \switchIndex{\transposition}{i}, \locsnd) 
    =  \CanRun(\gconf{}, i, \locsnd)
    \]
    Since $\lnot \CanRun(\gconf{}, i, \locsnd)$  (we are examining inactive processes in this sub-round Send), we have $\lnot  \CanRun(\switchGlobalConf{\gconf{}}{\transposition}, \switchIndex{\transposition}{i}, \locsnd)$.
    We prove that 
    \[
    \SFrozen(\switchGlobalConf{\gconf{}}{\transposition}, \switchGlobalConf{\gnextconf{}}{\transposition}, \switchIndex{\transposition}{i})
    \]
    as follows.
    By Definition~\ref{def:TransposAndGlobalConf}, we have 
    \begin{align*}
    \switchTempConf{\getProc{\gconf{}}{i}}{\transposition} & = \getProc{\switchGlobalConf{\gconf{}}{\transposition}}{\switchIndex{\transposition}{i}} \\
    \switchTempConf{\getProc{\gnextconf{}}{i}}{\transposition} & = \getProc{\switchGlobalConf{\gnextconf{}}{\transposition}}{\switchIndex{\transposition}{i}}
    \end{align*}
    
    By Proposition~\ref{prop:TranspositionOtherTempActions}, it follows that
    $
    \switchTempConf{\getProc{\gconf{}}{i}}{\transposition}
    \transstutter    
    \switchTempConf{\getProc{\gnextconf{}}{i}}{\transposition}
    $.    
    It follows
    $
    \getProc{\switchGlobalConf{\gconf{}}{\transposition}}{\switchIndex{\transposition}{i}} \transstutter \getProc{\switchGlobalConf{\gnextconf{}}{\transposition}}{\switchIndex{\transposition}{i}}
    $.
    We now examine the control component for a process $p_i$.
    By Definition~\ref{def:TransposAndGlobalConf}, we have
    \begin{align*}
    \getActive{\gconf{}}{i} & = \getActive{\switchGlobalConf{\gconf{}}{\transposition}}{\switchIndex{\transposition}{i}} \\
    \getActive{\gnextconf{}}{i} & = \getActive{\switchGlobalConf{\gnextconf{}}{\transposition}}{\switchIndex{\transposition}{i}}
    \end{align*}
    By definition of $\SFrozen$ in Section~\ref{sec:global-model}, we have $\getActive{\gconf{}}{i} = \getActive{\gnextconf{}}{i}$.
    It follows $\getActive{\switchGlobalConf{\gconf{}}{\transposition}}{\switchIndex{\transposition}{i}} 
    = \getActive{\switchGlobalConf{\gnextconf{}}{\transposition}}{\switchIndex{\transposition}{i}}$.
    We now show that  each outgoing message buffer from a process $p_i$ 
    %$\getProc{\switchGlobalConf{\gconf{}}{\transposition}}{\switchIndex{\transposition}{i}}$
    is unchanged from $\switchGlobalConf{\gconf{}}{\transposition}$ to $\switchGlobalConf{\gnextconf{}}{\transposition}$.       
    By Definition~\ref{def:TransposAndGlobalConf}, we have $\getBuf{\switchGlobalConf{\gnextconf{}}{\transposition}}{\switchIndex{\transposition}{i}}{\switchIndex{\transposition}{\ell}} = \getBuf{\gnextconf{}}{i}{\ell}$.
    Since $\getBuf{\gnextconf{}}{i}{\ell}  = \getBuf{\gconf{}}{i}{\ell}$,  (we are examining inactive processes in this sub-round Send), it follows
    $\getBuf{\switchGlobalConf{\gnextconf{}}{\transposition}}{\switchIndex{\transposition}{i}}{\switchIndex{\transposition}{\ell}} = \getBuf{\gconf{}}{i}{\ell}$.
    By Definition~\ref{def:TransposAndGlobalConf}, we have 
    $\getBuf{\gconf{}}{i}{\ell} = \getBuf{\switchGlobalConf{\gconf{}}{\transposition}}{\switchIndex{\transposition}{i}}{\switchIndex{\transposition}{\ell}}$. 
    It follows that
    \begin{align*}
    \getBuf{\switchGlobalConf{\gnextconf{}}{\transposition}}{\switchIndex{\transposition}{i}}{\switchIndex{\transposition}{\ell}} = \getBuf{\switchGlobalConf{\gconf{}}{\transposition}}{\switchIndex{\transposition}{i}}{\switchIndex{\transposition}{\ell}}
    \end{align*}
    Therefore, it follows $\SFrozen(\switchGlobalConf{\gconf{}}{\transposition}, \switchGlobalConf{\gnextconf{}}{\transposition}, \switchIndex{\transposition}{i})$. Constraint (a) holds.

    Now we focus on Constraint (b) by examining processes which are enabled in this sub-round Send. 
    Let $s \in \nid$ be an arbitrary index such that $\getActive{\gconf{}}{i}$.
    By the semantics of the sub-round Send in Section~\ref{sec:global-model}, it follows $\getProc{\gconf{}}{s} \tempsnd{}{m} \getProc{\gnextconf{}}{s}$.
    We have $\getProc{\gconf{}}{s} = \getProc{\switchGlobalConf{\gconf{}}{\transposition}}{\switchIndex{\transposition}{s}}$ and $\getProc{\gnextconf{}}{s} = \getProc{\switchGlobalConf{\gnextconf{}}{\transposition}}{\switchIndex{\transposition}{s}}$ by Definition~\ref{def:TransposAndGlobalConf}.
    By Proposition~\ref{prop:TranspositionOtherTempActions}, it follows 
    \[
    \getProc{\switchGlobalConf{\gconf{}}{\transposition}}{\switchIndex{\transposition}{s}} \tempsnd{}{m} \getProc{\switchGlobalConf{\gnextconf{}}{\transposition}}{\switchIndex{\transposition}{s}}
    \]
    We show that $m$ is new in buffers $\getBuf{\switchGlobalConf{\gconf{}}{\transposition}}{\switchIndex{\transposition}{s}}{1}, \ldots,$ $\getBuf{\switchGlobalConf{\gconf{}}{\transposition}}{\switchIndex{\transposition}{s}}{1}$.
    By Definition~\ref{def:TransposAndGlobalConf}, for every $s, r \in \nid$, we have
    \begin{align*}
    \getBuf{\switchGlobalConf{\gconf{}}{\transposition}}{\switchIndex{\transposition}{s}}{\switchIndex{\transposition}{r}} & = \getBuf{\gconf{}}{s}{r} \\ \getBuf{\switchGlobalConf{\gnextconf{}}{\transposition}}{\switchIndex{\transposition}{s}}{\switchIndex{\transposition}{r}}  & = \getBuf{\gnextconf{}}{s}{r}     
    \end{align*}
    
    We have $m \notin \getBuf{\gconf{}}{s}{r}$ and $m \in \getBuf{\gnextconf{}}{s}{r}$ by the semantics of the sub-round Send in Section~\ref{sec:global-model}.
    It follows 
    \[
    m \notin \getBuf{\switchGlobalConf{\gconf{}}{\transposition}}{\switchIndex{\transposition}{s}}{\switchIndex{\transposition}{r}} \quad \text{and} \quad m \in \getBuf{\switchGlobalConf{\gconf{}}{\transposition}}{\switchIndex{\transposition}{s}}{\switchIndex{\transposition}{r}}
    \]
    In other words, the message $m$ is new in the buffer $\getBuf{\switchGlobalConf{\gconf{}}{\transposition}}{\switchIndex{\transposition}{s}}{\switchIndex{\transposition}{r}}$.     
    As a result, Constraint (b) holds.
    
    Now we focus on Constraint (c). 
    Let $s \in \nid$ be an arbitrary index such that $\CanRun(\gconf{}, i, \locsnd) = \True$.
    By arguments at the beginning of the proof of Proposition~\ref{prop:ProjGlobalInternalTransitions}, we have $\CanRun(\switchGlobalConf{\gconf{}}{\transposition}, \switchIndex{\transposition}{i}, \locsnd)$.
    By the semantics of the sub-round Send in Section~\ref{sec:global-model}, we have $\lnot \getActive{\gnextconf{}}{i}$.
    By Definition~\ref{def:TransposAndGlobalConf}, we have $\getActive{\gnextconf{}}{i} = \getActive{\switchGlobalConf{\gnextconf{}}{\transposition}}{\switchIndex{\transposition}{i}}$.
    It follows $\lnot \getActive{\switchGlobalConf{\gnextconf{}}{\transposition}}{\switchIndex{\transposition}{i}}$.
    Constraint (c) holds.
    
    It implies that $\switchGlobalConf{\gconf{}}{\transposition} \transsnd \switchGlobalConf{\gnextconf{}}{\transposition}$.
    
    \item \emph{Sub-round Receive.} By similar arguments in the case of the sub-round Send, we have that Constraints (a) and (c) in the sub-round Receive holds.
    In the following, we focus on Constraint (b).
    By similar arguments in the case of the sub-round Send, we have that a process $\getProc{\gconf{}}{s}$ is enabled in this sub-round Receive if and only if a process $\getProc{\switchGlobalConf{\gconf{}}{\transposition}}{\switchIndex{\transposition}{s}}$ is enabled in this sub-round Receive for every $s \in \nid$.
    Hence, we focus on  processes which are enabled in this sub-round Receive. Let $r \in \nid$ be an index such that $\CanRun(\gconf{}, i, \locrcv)$.
    By the semantics of the sub-round Receive in Section~\ref{sec:global-model}, we have
    $\getProc{\gconf{}}{r} \temprcv{}{S_1}{S_N} \getProc{\gnextconf{}}{r} $ for some sets  $S_1, \ldots, S_N \subseteq \set(\msg)$ of messages.
    By Definition~\ref{def:TransposAndTempConf}, we have 
    \begin{align*}
      \switchTempConf{\getProc{\gconf{}}{r}}{\transposition} & = \getProc{\switchGlobalConf{\gconf{}}{\transposition}}{\switchIndex{\transposition}{r}} \\
      \switchTempConf{\getProc{\gnextconf{}}{r}}{\transposition} & = \getProc{\switchGlobalConf{\gnextconf{}}{\transposition}}{\switchIndex{\transposition}{r}}
    \end{align*}
    By Proposition~\ref{prop:TranspositionTempReceive}, it follows
    \[
    \getProc{\switchGlobalConf{\gconf{}}{\transposition}}{\switchIndex{\transposition}{r}} \switchRcv{S_1}{S_N}{\transposition} \getProc{\switchGlobalConf{\gnextconf{}}{\transposition}}{\switchIndex{\transposition}{r}}
    \]
    Now we focus on the update of message buffers. 
    By the semantics of the sub-round Receive in Section~\ref{sec:global-model}, we have $S_s \subseteq \getBuf{\gconf{}}{s}{r}$ for every $s \in \nid$.
    By Definition~\ref{def:TransposAndGlobalConf}, we have $\getBuf{\gconf{}}{s}{r} = \getBuf{\switchGlobalConf{\gconf{}}{\transposition}}{\switchIndex{\transposition}{s}}{\switchIndex{\transposition}{r}}$.
    It follows that $S_s \subseteq \getBuf{\switchGlobalConf{\gconf{}}{\transposition}}{\switchIndex{\transposition}{s}}{\switchIndex{\transposition}{r}}$ for every $s \in \nid$.
    %     and similar arguments in the proof of the sub-round Send, it is easy to see that $S_{\switchIndex{\transposition}{s}}$ of messages is a subset of $
    We now prove that for every $s \in \nid$, $S_s$ is removed from the message buffer $\getBuf{\switchGlobalConf{\gnextconf{}}{\transposition}}{\switchIndex{\transposition}{s}}{\switchIndex{\transposition}{r}}$.
    By Definition~\ref{def:TransposAndGlobalConf}, we have
    \begin{align*}
    \getBuf{\switchGlobalConf{\gnextconf{}}{\transposition}}{\switchIndex{\transposition}{s}}{\switchIndex{\transposition}{r}} & = \getBuf{\gnextconf{}}{s}{r}  \\
    \getBuf{\switchGlobalConf{\gconf{}}{\transposition}}{\switchIndex{\transposition}{s}}{\switchIndex{\transposition}{r}}  & = \getBuf{\gconf{}}{s}{r}  
    \end{align*}
    By the semantics of the sub-round Receive in Section~\ref{sec:global-model}, we have
    \begin{align*}
    S_s & \cap \getBuf{\gnextconf{}}{s}{r} = \emptyset \\
    \getBuf{\gconf{}}{s}{r} & = \getBuf{\gnextconf{}}{s}{r} \cup S_s
    \end{align*}
    It follows that 
    \begin{align*}
    S_s & \cap \getBuf{\switchGlobalConf{\gnextconf{}}{\transposition}}{\switchIndex{\transposition}{s}}{\switchIndex{\transposition}{r}} = \emptyset \\
    \getBuf{\switchGlobalConf{\gconf{}}{\transposition}}{\switchIndex{\transposition}{s}}{\switchIndex{\transposition}{r}} & = \getBuf{\switchGlobalConf{\gnextconf{}}{\transposition}}{\switchIndex{\transposition}{s}}{\switchIndex{\transposition}{r}} \cup S_s
    \end{align*}
    Constraint (b) holds. 
    It implies that 
    $\switchGlobalConf{\gnextconf{}}{\transposition}  \transrcv \switchGlobalConf{\gconf{}}{\transposition}$.

    \item \emph{Sub-round Computation.} By applying similar arguments in above sub-rounds.

  \end{enumerate}
  Therefore, Proposition~\ref{prop:SymGlobalSmallActions} holds. %$\qedhere$ 
\end{proof}

\newcommand{\propSymGlobalTransitions}{
  Let $\gsys{N} = \left( \gconfs{N}, \gtr{N}, \grel{N}, \ginitconf{N} \right)$ be a global transition system with indexes $\nid$ and a process template $\mytemp{N} = (\tempconfs{N}, \temptr{N}, \temprel{N}, \tempinitconf{N})$.
  Let $\gconf{}$ and $\gnextconf{}$  be configurations in $\gsys{N}$ such that $\gconf{} \leadsto \gnextconf{}$.
  Let $\transposition$ be a transposition on $\nid$, and
  $\transposition_C$ be a global transposition based on $\transposition$ (from Definition~\ref{def:TransposAndGlobalConf}).  
  It follows  $\switchGlobalConf{\gconf{}}{\transposition} \leadsto \switchGlobalConf{\gnextconf{}}{\transposition} $. 
}

\begin{prop} \label{prop:SymGlobalTransitions}
  \propSymGlobalTransitions
\end{prop}

\begin{proof}
It immediately follows by Proposition~\ref{prop:SymGlobalSmallActions} and the fact that for all $i \in \nid$, we have $ \getActive{\switchGlobalConf{\gconf{}}{\transposition}}{\switchIndex{\transposition}{i}} = \getActive{\gconf{}}{i}$. %$\qedhere$
\end{proof}

\subsection{Trace Equivalence of $\gsys{2}$ and $\gsys{N}$ under $AP_{\{1, 2\}}$} \label{app:TraceEquivalenceTwo}

Recall that  $\gsys{2}$ and $\gsys{N}$ are two global transition systems of 2 and $N$ processes, respectively, such that every correct process runs the same arbitrary~\classname{} algorithm, and a set $AP_{\{1, 2\}}$ contains predicates 
%that takes one of the forms: $P_1(1)$, $P_2(2)$, $P_3(1, 2)$, or $P_4(2, 1)$ where $1$ and $2$ are process indexes.
that takes one of the forms: $P_{1}(1)$, $P_{2}(2)$, $P_{3}(1, 2)$, or $P_{4}(2, 1)$ where $1$ and $2$ are process indexes.

\begin{prop} \label{prop:IndexTempConfProjTypeOK}
  Let $\mathcal{A}$ be an arbitrary~\classname{} algorithm.
  Let $\mytemp{2}$ and $\mytemp{N}$ be two process templates of $\mathcal{A}$ for some $N \geq 2$, and 
  $\tempconf{}^N \in \tempconfs{N}$ be a state of $\mytemp{N}$.
  Let $\tempconf{}^2$ be a tuple that is the application of Construction \ref{const:TempConfProj} to $\tempconf{}^N$ and indexes $\{1, 2\}$.
  Then, $\tempconf{}^2$ is a template state of $\mytemp{2}$.
\end{prop}

\begin{proof}
It immediately follows by Construction~\ref{const:TempConfProj}.
\end{proof}

\begin{prop} \label{prop:IndexGlobalConfProjTypeOK}
  Let $\mathcal{A}$ be an arbitrary~\classname{} algorithm.
  Let $\gsys{2}$ and $\gsys{N}$ be two global transition systems of two instances of $\mathcal{A}$ for some $N \geq 2$, and 
  $\gconf{}^N \in \gconfs{N}$ be a global configuration in $\gsys{N}$.
  Let $\gconf{}^2$ be a tuple that is the application of Construction \ref{const:GlobalConfProj} to $\gconf{}^N$ and indexes $\{1, 2\}$.
  Then, $\gconf{}^2$ is a global configuration of $\gsys{2}$.
\end{prop}

\begin{proof}
  It immediately follows by Construction~\ref{const:GlobalConfProj}.
\end{proof}

\recalllemma{lem:IDProjPath}{\lemIDProjPath}

The proof of Lemma~\ref{lem:IDProjPath}  requires the following propositions:
\begin{enumerate}
  \item Proposition~\ref{prop:ProjInitTempConf} says that the application of Construction \ref{const:TempConfProj} to an initial template state of $\gsys{N}$ constructs an initial template state of $\gsys{2}$. 
  
  \item Lemmas \ref{prop:ProjTempReceive} and \ref{prop:ProjTempOtherActions} say that Construction \ref{const:TempConfProj} preserves the process transition relation.
  
  \item Proposition~\ref{prop:ProjInitGlobalConf} says that the application of Construction \ref{const:GlobalConfProj} to an initial global configuration of $\gsys{N}$ constructs an initial global configuration of $\gsys{2}$. 
  
  \item Propositions~\ref{prop:ProjGlobalInternalTransitions} and \ref{prop:ProjGlobalObservableTransitions} 
  say Construction \ref{const:GlobalConfProj} preserves the global transition relation.
  Proposition~\ref{prop:ProjGlobalInternalTransitions} captures internal transitions.
  Proposition~\ref{prop:ProjGlobalObservableTransitions} captures round transitions.
\end{enumerate}

{
  \emph{Proof of Lemma~\ref{lem:IDProjPath}.} It is easy to check that Lemma~\ref{lem:IDProjPath} holds by Propositions \ref{prop:ProjInitTempConf}, \ref{prop:ProjTempOtherActions},  \ref{prop:ProjInitGlobalConf}, \ref{prop:ProjGlobalInternalTransitions}, and \ref{prop:ProjGlobalObservableTransitions}. 
  The detailed proofs of these propositions are given below. $\qedhere$  
}

\newcommand{\propProjInitTempConf}{
  Let $\mathcal{A}$ be an arbitrary~\classname{} algorithm.
  Let $\mytemp{N} = (\tempconfs{N}, \temptr{N}, \temprel{N}, \tempinitconf{N}), \mytemp{2} = (\tempconfs{2}, \temptr{2}, \temprel{2}, \tempinitconf{2})$ be two process templates  of $\mathcal{A}$ for some $N \geq 2$.
  It follows that $\tempinitconf{2}$ is the index projection of $\tempinitconf{N}$ on indexes $\{1, 2\}$.
}

\begin{prop} \label{prop:ProjInitTempConf}
  \propProjInitTempConf
\end{prop}

\begin{proof}
  It follows by Construction \ref{const:TempConfProj} and the definition of $\tempinitconf{2}$ in Section~\ref{sec:temp-model}. % $\qedhere$
\end{proof}

\newcommand{\propProjTempReceive}{
  Let $\mathcal{A}$ be an arbitrary~\classname{} algorithm, and $\mytemp{2}$ and $\mytemp{N}$ be process templates of $\mathcal{A}$.  
  Let  $\tempconf{0}$ and $\tempconf{1}$  be template states in $\mytemp{N}$ such that $\tempconf{0} \temprcv{}{S_1}{S_N} \tempconf{1}$ for some sets $S_1, \ldots, S_N$ of messages . 
  Let $\tempnextconf{0}$ and $\tempnextconf{1}$ be template states of $\mytemp{2}$ such that they are constructed with Construction \ref{const:GlobalConfProj} and based on configurations $\tempconf{0}$ and $\tempconf{1}$ and indexes $\{1, 2\}$.
  It follows  $\tempnextconf{0} \xrightarrow{\mathit{rcv}(S_1, S_2)} \tempnextconf{1}$.  
}

\begin{prop} \label{prop:ProjTempReceive}
  \propProjTempReceive
\end{prop}

\begin{proof}
  We prove that all Constraints (a)--(c) for the transition $\csnd$ defined in Section~\ref{sec:temp-model} hold.
  First, we focus on Constraint (a). We have $\getPC{\switchTempConf{\tempconf{1}}{\transposition}} 
  =   \getPC{\tempconf{1}}$ by Definition~\ref{def:TransposAndTempConf}.
  We have
  $\getPC{\tempconf{1}}  = \nextLoc(\getPC{\tempconf{0}})$ by the semantics of $\crcv(S_1, \ldots, S_N)$ in Section~\ref{sec:temp-model}.
  We have $\nextLoc(\getPC{\tempconf{0}}) = \nextLoc(\getPC{\tempnextconf{0}})$ by Definition~\ref{def:TransposAndTempConf}.
  Hence, it follows $
  \getPC{\switchTempConf{\tempconf{1}}{\transposition}} = \nextLoc(\getPC{\switchTempConf{\tempconf{0}}{\transposition}})$.
  Moreover, we have $\emptyset  = \genMsg(\getPC{\tempconf{0}})$ by the semantics of $\crcv$ in Section~\ref{sec:temp-model}. 
  By Construction~\ref{const:TempConfProj}, we have $\genMsg(\getPC{\tempconf{0}}) = \genMsg(\getPC{\tempnextconf{0}})$. 
  It follows that
  $\genMsg(\getPC{\tempnextconf{0}}) = \emptyset$.
  Constraint (a) holds.

  We now check components related to received messages (Constraint (b)). 
  Let $i$ be an arbitrary index in $\tid$.
  It follows
  \begin{alignat*}{3}
  & \getBox{\tempnextconf{1}}{i} && \\
  = \; & \getBox{\tempconf{1}}{i} && \quad \text{(by Construction \ref{const:TempConfProj})}  \\
  = \; & \getBox{\tempconf{0}}{i} \cup S_{i} && \quad \text{(by the semantics of $\crcv(S_1, \ldots, S_N)$ in Section~\ref{sec:temp-model})}              \\
  = \; & \getBox{\tempnextconf{0}}{i} \cup S_{i}  && \quad \text{(by Construction \ref{const:TempConfProj})}  
  \end{alignat*}
  Hence, we have $\forall i \in \tid \colon \getBox{\tempnextconf{1}}{i} = \getBox{\tempnextconf{0}}{i} \cup S_{i}$.
  Constraint (b) holds.
  Moreover, by similar arguments in the proof of Proposition~\ref{prop:TranspositionTempReceive}, we have  
  \[
  \forall i \in \tid \colon \getDatum{\tempnextconf{1}}{i} =  \nextVar \big(  \getPC{\tempnextconf{0}}, 
%  \getBox{\tempnextconf{1}}{i},               
  S_i,
  \getDatum{\tempnextconf{0}}{i} \big)
  \]
  
  Constraint (c) holds. 
  It follows $\tempnextconf{0} \xrightarrow{\mathit{rcv}(S_1, S_2)} \tempnextconf{1}$. % $\qedhere$
\end{proof}

\newcommand{\propProjTempOtherActions}{
  Let $\mathcal{A}$ be an arbitrary~\classname{} algorithm, and $\mytemp{2}$ and $\mytemp{N}$ be process templates of $\mathcal{A}$.
  Let $\mathit{tr} \in \temptr{N}$ be a transition such that it is one of send, computation, crash or stuttering transitions.
  Let $\tempconf{0}$ and $\tempconf{1}$ be template states in $\mytemp{N}$ such that $\tempconf{0} \xrightarrow{\mathit{tr}} \tempconf{1}$. 
  Let $\tempnextconf{0}$ and $\tempnextconf{1}$ be states of $\mytemp{2}$ such that they are the index projection of $\tempconf{0}$ and $\tempconf{1}$ on indexes $\{1, 2\}$, respectively.
  Then, $\tempnextconf{0}\xrightarrow{\mathit{tr}} \tempnextconf{1}$.  
}

\begin{prop} \label{prop:ProjTempOtherActions}
  \propProjTempOtherActions
\end{prop}

\begin{proof}
  By similar arguments in the proof of Lemma~\ref{prop:ProjTempReceive}.
\end{proof}

Now we turn to properties of global configurations under~Constructions \ref{const:GlobalConfProj}.

\newcommand{\propProjInitGlobalConf}{
  Let $\mathcal{A}$ be an arbitrary~\classname{} algorithm.
  Let $\gsys{N} = (\gconfs{N}, \gtr{N}, \grel{N}, \ginitconf{N})$ and $\gsys{2} = (\gconfs{2}, \gtr{2}, \grel{2}, \ginitconf{2})$ be two transition systems of two instances of $\mathcal{A}$ for some $N \geq 2$.
  It follows that $\ginitconf{2}$ is the index projection of $\ginitconf{N}$ on indexes $\{1, 2\}$.
}

\begin{prop} \label{prop:ProjInitGlobalConf}
  \propProjInitGlobalConf
\end{prop}

\begin{proof}
  It immediately follows by Construction \ref{const:GlobalConfProj} and the definition of $\ginitconf{N}$. % $\qedhere$
\end{proof}

\newcommand{\propProjGlobalInternalTransitions}{
  Let $\mathcal{A}$ be an arbitrary~\classname{} algorithm.
  Let $\gsys{2} = (\gconfs{2}, \gtr{2}, \grel{2}, \ginitconf{2})$ and $\gsys{N} = (\gconfs{N}, \gtr{N}, \grel{N}, \ginitconf{N})$ be global transition systems such that all processes in $\gsys{2}$ and $\gsys{N}$ follow the same algorithm $\mathcal{A}$ for some $N \geq 2$.
  Let $\gconf{0}$ and $\gconf{1}$ be global configurations in $\gsys{N}$ such that $\gconf{0} \xrightarrow{\mathit{tr}} \gconf{1}$ where $\mathit{tr}$ is an internal transition. 
  Let $\gnextconf{0}$ and $\gnextconf{1}$ be the index projection of $\gconf{0}$ and $\gconf{1}$ on a set $\{1, 2\}$ of indexes, respectively.
  It follows that $\gnextconf{0} \xrightarrow{\mathit{tr}} \gnextconf{1}$.
}

\begin{prop} \label{prop:ProjGlobalInternalTransitions}
  \propProjGlobalInternalTransitions
\end{prop}

\begin{proof}
  First, by Proposition~\ref{prop:IndexGlobalConfProjTypeOK}, both $\tempnextconf{0}$ and  $\tempnextconf{1}$ are configurations in $\gsys{2}$.
  We prove Proposition~\ref{prop:ProjGlobalInternalTransitions} by case distinction.
  Here we provide detailed proofs only of two sub-rounds Schedule and Send. The proofs of other sub-rounds are similar.
  \begin{itemize}
    \item \emph{Sub-round Schedule.} 
    We prove that all Constraints (a)--(c) hold in $\gnextconf{0}$ and $\gnextconf{1}$.
    We now focus on Constraint (a).
    We have $  \getActive{\gnextconf{0}}{1} = \getActive{\gconf{0}}{1}$ and $\getActive{\gconf{0}}{2} = \getActive{\gconf{0}}{2}$ by Construction~\ref{const:GlobalConfProj}.
    By the semantics of the sub-round Schedule in Section~\ref{sec:global-model}, we have
    $\lnot \getActive{\gconf{0}}{1}  \land \lnot \getActive{\gconf{0}}{2}$. Hence, it follows
    $\lnot  \getActive{\gnextconf{0}}{1} \land \lnot  \getActive{\gnextconf{0}}{2}$.   
    Hence, the sub-round Schedule can start with a configuration $\gnextconf{0}$.
    Constraint (a) holds.
    By Proposition~\ref{prop:ProjTempOtherActions}, Constraint (b) holds.
    Constraint (c) holds by Construction~\ref{const:GlobalConfProj}.    
    Now we focus on incoming message buffers to correct processes to prove Constraint (d).
    Let $r$ be an index in $\nid$ such that $\getPC{\getProc{\gnextconf{1}}{r}} \neq \loccrash$.
    By Construction \ref{const:GlobalConfProj}, for every $ s \in \tid $, we have
    $\getBuf{\gnextconf{1}}{s}{r}  = \getBuf{\gconf{1}}{s}{r}$ and  $\getBuf{\gconf{0}}{s}{r} = \getBuf{\gnextconf{0}}{s}{r}$.
    By the semantics of the sub-round Schedule in Section~\ref{sec:global-model}, we have $\getBuf{\gconf{1}}{s}{r}  = \getBuf{\gconf{0}}{s}{r}$ for every $s \in \tid$.
    It follows $\getBuf{\gnextconf{1}}{s}{r} = \getBuf{\gnextconf{0}}{s}{r}$ for every $ s \in \tid $.        
    Constraint (d) holds.
    Now we focus on message buffers to crashed processes to prove Constraint (d).
    Let $r$ be an index in $\nid$ such that $\getPC{\getProc{\gnextconf{1}}{r}} = \loccrash$.
    By Construction \ref{const:GlobalConfProj}, for every $ s \in \tid $, we have
    $\getBuf{\gnextconf{1}}{s}{r}  = \getBuf{\gconf{1}}{s}{r}$.
    By the semantics of the sub-round Schedule in Section~\ref{sec:global-model}, we have $\getBuf{\gconf{1}}{s}{r}  = \emptyset$ for every $s \in \tid$,
    It follows $\getBuf{\gnextconf{1}}{s}{r} = \emptyset$ for every $ s \in \tid $.        
    Constraint~(e) holds.   
    It implies that $\gnextconf{0} \transsched \gnextconf{1}$.
    
    \item \emph{Sub-round Send.} 
    We have 
    $  \getActive{\gnextconf{0}}{1} = \getActive{\gconf{0}}{1}$ and $\getActive{\gconf{0}}{2} = \getActive{\gconf{0}}{2}$ 
    for every $k \in \tid$. 
    Hence, if a process $p^N_i$ in $\gsys{N}$ is enabled in this 
    sub-round, a corresponding process $p^2_i$ in $\gsys{N}$ is also for every $i \in \tid$.
    We prove that all Constraints (a)--(c) between $\gnextconf{0}$ and $\gnextconf{1}$ for the sub-round Send defined in Section~\ref{sec:global-model} hold.
    By similar arguments in the proof of Proposition~\ref{prop:ProjGlobalInternalTransitions}, Constraint (a) holds.    
    Now we focus on enable processes to prove Constraints (b) and (c).
    
    Assume that a process  $p^N_i$ in $\gsys{N}$ has sent a message $m$ in this sub-round, we show that  a process  $p^2_i$ in $\gsys{2}$ has also sent the message $m$ in this sub-round where $i \in \tid$.
    By Proposition~\ref{prop:ProjTempOtherActions}, it follows that $\getProc{\gnextconf{0}}{i} \tempsnd{}{m} \getProc{\gnextconf{1}}{i}$.
    Now we show that the message $m$ is new in buffers $\getBuf{\gnextconf{1}}{i}{1}$ and $\getBuf{\gnextconf{1}}{i}{2}$.
    By Construction \ref{const:GlobalConfProj}, we have $\getBuf{\gnextconf{1}}{i}{\ell}  = \getBuf{\gconf{1}}{i}{\ell}$ and $\getBuf{\gnextconf{0}}{i}{\ell}  = \getBuf{\gconf{0}}{i}{\ell}$ for every $\ell \in \tid$.
    By the semantics of the sub-round Send in Section~\ref{sec:global-model}, we have $\getBuf{\gconf{1}}{i}{\ell}  = \{ m \} \cup \getBuf{\gconf{0}}{i}{\ell}$.
    It follows
    $
    \getBuf{\gnextconf{1}}{i}{\ell}  
    = \{ m \} \cup \getBuf{\gnextconf{0}}{i}{\ell}
    $
    for every $\ell \in \tid$. 
    Moreover, since $m \notin \getBuf{\gconf{0}}{i}{\ell}$ for every $\ell \in \tid$, we have $m \notin \getBuf{\gnextconf{0}}{i}{\ell}$. 
    Therefore, the message $m$ is new in a buffer $\getBuf{\gnextconf{1}}{i}{\ell}$ for every $\ell \in \tid$.
    Constraint (b) holds.
    Moreover, by Construction~\ref{const:GlobalConfProj}, we have $\getActive{\gconf{1}}{i} = \getActive{\gnextconf{1}}{i}$.
    We have $\lnot \getActive{\gconf{1}}{i}$ by the semantics of the sub-round Send in Section~\ref{sec:global-model}.
    It follows $\lnot \getActive{\gnextconf{1}}{i}$.
    Constraint (e) holds.
    It follows that $\gnextconf{0} \transsnd \gnextconf{1}$.
    
    \item \emph{Sub-rounds Receive and Computation.} Similar.
  \end{itemize}
  Therefore, Proposition~\ref{prop:ProjGlobalInternalTransitions} holds. % $\qedhere$
\end{proof}

\newcommand{\propProjGlobalObservableTransitions}{
  Let $\mathcal{A}$ be an arbitrary~\classname{} algorithm.
  Let $\gsys{2}$ and $\gsys{N}$ be global transition systems of $\mathcal{A}$ for some $N \geq 2$.
  Let $\gconf{0}$ and $\gconf{1}$ be global configurations of $\gconfs{N}$ such that $\gconf{0} \leadsto \gconf{1}$.
  Let $\gnextconf{0}$ and $\gnextconf{1}$ be the index projection of $\gconf{0}$ and $\gconf{1}$ on indexes $\{1, 2\}$. 
  Then, $\gnextconf{0} \leadsto \gnextconf{1}$.
}

\begin{prop} \label{prop:ProjGlobalObservableTransitions}
  \propProjGlobalObservableTransitions
\end{prop}

\begin{proof}
  It immediately follows by Propositions~\ref{prop:ProjTempOtherActions} and~\ref{prop:ProjGlobalInternalTransitions}.
\end{proof}

Now we present how to construct an admissible path of $\gsys{N}$ from a given admissible path of $\gsys{2}$ with Lemma~\ref{lem:IDRefilledPath} below.
The main argument is that from an admissible sequence of configurations in $\gsys{2}$, we can get an admissible sequence of configurations in $\gsys{N}$ by letting processes 3 to $N$ be initially crashed.
The proof of Lemma~\ref{lem:IDRefilledPath} requires the preliminary Propositions~\ref{prop:IDRefilledConf} and \ref{prop:IDRefilledTrans}.

\newcommand{\propIDRefilledConf}{
  Let $\mathcal{A}$ be an arbitrary~\classname{} algorithm.
  Let $\gsys{2}$ and $\gsys{N}$ be global transition systems of $\mathcal{A}$ for some $N \geq 2$.  
  Let $\gconf{}^2$ be a configuration in $\gsys{2}$. 
  There exists a configuration $\gconf{}^N$ in $\gsys{N}$ such that the following properties hold:
  \begin{itemize}
    \item $\gconf{}^2$ is the index projection of $\gconf{}^N$ on indexes $\{1, 2\}$, and
    
    \item $\forall i \in 3..N \colon \getPC{\getProc{\gconf{}^N}{i}}= \loccrash \land \lnot \getActive{\gconf{}^N}{i}$
    
     \item $\forall s \in 3..N, r \in \nid \colon \getBuf{\gconf{}}{s}{r} = \emptyset $
    
    \item $\forall s \in 3..N, r \in 1..2 \colon     \getBox{\getProc{\gconf{}^N}{s}}{r} = \emptyset$
    
    \item $\forall s \in \nid, r \in 3..N \colon \getBuf{\gconf{}}{s}{r} = \emptyset $
    
  \end{itemize}  
}

\begin{prop} \label{prop:IDRefilledConf}
  \propIDRefilledConf
\end{prop}

\begin{proof}
  Proposition~\ref{prop:IDRefilledConf} is true since our construction simply adds $N-2$ crashed processes that have not sent any messages in the global system.
  The last two constraints requires that processes 1 and 2 have not received any messages from crashes processes and the message buffers to crashed processes are empty. 
  Other components are arbitrary. % $\qedhere$
\end{proof}

\newcommand{\propIDRefilledTrans}{
  Let $\mathcal{A}$ be an arbitrary~\classname{} algorithm.
  Let $\gsys{2}$ and $\gsys{N}$ be global transition systems of $\mathcal{A}$ for some $N \geq 2$.  
  Let $\gconf{0}^2$ and $\gconf{1}^2$ be configurations in $\gsys{2}$ such that $\gconf{0}^2 \xrightarrow{\mathit{tr}} \gconf{1}^2$ where $\mathit{tr}$ is an internal transition.
  There exists two configurations $\gconf{0}^N$ and $\gconf{1}^N$ in $\gsys{N}$ such that
  \begin{enumerate}[nosep]
    \item $\gconf{0}^2$ and $\gconf{1}^2$ are respectively the index projection of $\gconf{0}^N$ and $\gconf{1}^N$ on a set $\{1, 2\}$ of indexes, and
    
    \item $\gconf{0}^N \xrightarrow{\mathit{tr}} \gconf{1}^N$.
  \end{enumerate}
}

\begin{prop} \label{prop:IDRefilledTrans}
  \propIDRefilledTrans  
\end{prop}

\begin{proof}
  By Proposition~\ref{prop:IDRefilledConf}, there exists a configuration $\gconf{0}^N$ in $\gsys{N}$ such that (i) $\gconf{0}^2$ is the index projection of $\gconf{0}^N$ on indexes $\{1, 2\}$, and (ii) all processes with indexes in $3..N$ are crashed and inactive in $\gconf{0}^N$, and (iii) every process $p_i$ has not received any messages from a process $p_j$ where $i \in \{1, 2\}, j \in 3..N$ (as the configuration construction in the proof of Proposition~\ref{prop:IDRefilledConf}).
  
  We construct $\gconf{1}^N$ as the following.
  Intuitively, this construction keeps process $3..N$ crashed, and two processes 1 and 2 in $\gsys{N}$ make similar transitions with processes 1 and 2 in $\gsys{2}$.´
  \begin{enumerate}
    \item $\forall i \in 3..N \colon \getProc{\gconf{1}^N}{i}= \getProc{\gconf{0}^N}{i} \land \lnot \getActive{\gconf{1}^N}{i}$.
    
    \item $\forall s \in \nid, r \in 3..N \colon \getBuf{\gconf{1}^N}{s}{r}  = \emptyset$
    
    \item $\forall i \in \{1, 2\} \colon\getActive{\gconf{1}^N}{i} = \getActive{\gconf{1}^2}{i} $
    
    \item $\forall s \in 3..N, r \in \nid \colon     \getBox{\getProc{\gconf{1}^N}{s}}{r} = \emptyset$
    
    \item $\forall s, r \in \{1, 2\} \colon \getBuf{\gconf{1}^N}{s}{r} = \getBuf{\gconf{1}^2}{s}{r} $
    
    \item For every $s \in \{1, 2\}$, we have: If $\getProc{\gconf{0}^2}{s} \tempsnd{}{m} \getProc{\gconf{1}^2}{s}$ for some $m \in \msg$, then $\getBuf{\gconf{1}^N}{s}{r} = \{m\} \cup \getBuf{\gconf{0}^N}{s}{r}$ for every $3 \leq r \leq N$.
    Otherwise, $\getBuf{\gconf{1}^N}{s}{r} =  \getBuf{\gconf{0}^N}{s}{r}$.
    
    \item For every $i \in \{1, 2\}$, the configurations of processes with index $i$  is updated as following:      
    \begin{itemize}[nosep]
      \item  $\getPC{\getProc{\gconf{1}^N}{i}}= \getPC{\getProc{\gconf{1}^2}{i}}$

      \item $\forall j \in \{1, 2\} \colon \getBox{\getProc{\gconf{1}^N}{i}}{j}
      = \getBox{\getProc{\gconf{1}^2}{i}}{j}$
      
      \item $\forall j \in \{1, 2\} \colon \getDatum{\getProc{\gconf{1}^N}{i}}{j} = \getDatum{\getProc{\gconf{1}^2}{i}}{j} $

      \item $\forall j \in 3..N \colon \getBox{\getProc{\gconf{1}^N}{i}}{j} = \getBox{\getProc{\gconf{0}^N}{i}}{j} = \emptyset$      
      
      \item $\forall j \in 3..N \colon \getDatum{\getProc{\gconf{1}^N}{i}}{j} 
      = \nextVar \big( \getPC{\getProc{\gconf{0}^N}{i}},  
      \emptyset,
      \getDatum{\getProc{\gconf{0}^N}{i}}{j} \big)  $
    \end{itemize}
  \end{enumerate} 
  By the construction of $\gconf{1}^N$, it follows that $\gconf{1}^N$ is a configuration in $\gsys{N}$.
  Moreover, we have $\getProc{\gconf{0}^N}{i} \tempstutter \getProc{\gconf{1}^N}{i}$ and $\getProc{\gconf{0}^N}{i}$ is crashed for every $i \in 3..N$.
  Moreover, no message from a process $p_i$ has been sent or received for every $i \in 3..N$.
  
  By the above construction, it immediately follows that  $\gconf{1}^2$ is the index projection $\gconf{1}^N$ on indexes $\{1, 2\}$. 
  Hence, point 1 in Proposition~\ref{prop:IDRefilledTrans} holds.
  We prove point 2 in Proposition~\ref{prop:IDRefilledTrans} by case distinction.
  
  \begin{itemize}[nosep]
    \item \emph{Sub-round Schedule.} By similar arguments in the proof of Proposition~\ref{prop:ProjGlobalInternalTransitions} and the construction of $\gconf{1}^N$, we have $\gconf{1}^N \transstutter c_2^N$.
    
    \item \emph{Sub-round Send.}  By construction of $\gconf{0}^N$ and $\gconf{1}^N$, we know that every process $p_i$ is crashed, and its state is not updated, and every outgoing message buffer from $p_i$ is always empty for every $i \in 3..N$.
    Therefore, in the following, we focus on only two processes $p_1$ and $p_2$.
    For every $i \in \tid$, if $\lnot \CanRun(c^N_0, c^N_1, \locsnd)$, it follows $\SFrozen(\gconf{0}^N, \gconf{1}^N, i)$ by the construction of configurations $\gconf{0}^N$ and $\gconf{1}^N$.
    Constraint (a) holds.
    We now focus on enabled processes in this sub-round.
    Let $i$ be an index in $\tid$ such that $\CanRun(c^N_0, c^N_1, \locsnd)$. 
    By the semantics of the sub-round Send in Section~\ref{sec:temp-model}, we have
    $\getProc{\gconf{0}}{i} \tempsnd{}{m} \getProc{\gconf{1}}{i}$.
    We prove that $\getProc{\gconf{0}^N}{i} \tempsnd{}{m} \getProc{\gconf{1}^N}{i}$ as follows.
    By the construction of $\gconf{0}^N$, we have that
    $\getPC{\getProc{\gconf{0}^N}{i}} = \getPC{\getProc{\gconf{0}}{i}}$.
    By the construction of $\gconf{1}^N$, we have that
    $\getPC{\getProc{\gconf{1}^N}{i}} = \getPC{\getProc{\gconf{1}}{i}}$.
    By the semantics of the transition $\csnd(m)$ in Section~\ref{sec:temp-model}, we have
    $\getPC{\getProc{\gconf{1}}{i}} = \nextLoc(\getPC{\getProc{\gconf{0}}{i}})$.
    It follows that 
    \[\getPC{\getProc{\gconf{1}^N}{i}} = \nextLoc(\getPC{\getProc{\gconf{0}^N}{i}})\]
    By similar arguments, it follows $\{m\} = \genMsg(\getPC{\getProc{\gconf{0}^N}{i}})$.
    Now we focus on received messages of a process $p_i$.          
    By the semantics of the transtion $\csnd(m)$ in Section~\ref{sec:temp-model}, we have $\getBox{\getProc{\gconf{1}^2}{i}}{j}  = \getBox{\getProc{\gconf{0}^2}{i}}{j}$.
    For every $j \in \{1, 2\}$, we have
    \begin{alignat*}{3}
    & \getBox{\getProc{\gconf{1}^N}{i}}{j} && \\
    = \, & \getBox{\getProc{\gconf{1}}{i}}{j} && \quad \text{(by the construction of $\gconf{1}^N$)} \\
    = \, & \getBox{\getProc{\gconf{0}}{i}}{j} && \quad \text{(by the semantics of $\csnd(m)$ in Section~\ref{sec:temp-model})} \\
    = \, & \getBox{\getProc{\gconf{0}^N}{i}}{j}  && \quad \text{(by the construction of $\gconf{0}^N$)} 
    \end{alignat*}
    For every $j \in 3..N$, we have    
    \begin{alignat*}{3}
     \getBox{\getProc{\gconf{1}^N}{i}}{j} 
    = \, & \emptyset && \quad \text{(by the construction of $\gconf{1}^N$)} \\
= \, & \getBox{\getProc{\gconf{0}^N}{i}}{j}  && \quad \text{(by the construction of $\gconf{0}^N$)} 
    \end{alignat*}    
    Hence, a process $p_i$ does not receive any message when taking a step from $\gconf{0}^N$ to $\gconf{1}^N$.
    By the construction of $\gconf{1}^N$, it follows 
    \begin{alignat*}{2}
      & \getDatum{\getProc{\gconf{1}^N}{i}}{j} \\
    = & \nextVar \big( \getPC{\getProc{\gconf{0}^N}{i}},  
    \emptyset,
    \getDatum{\getProc{\gconf{0}^N}{i}}{j} \big)
    \end{alignat*}
     for every $j \in \nid$.
    Hence, it follows $\getProc{\gconf{0}^N}{i} \tempsnd{}{m} \getProc{\gconf{1}^N}{i}$.
    Now we focus on outgoing message buffers from $p_i$. 
    It is easy to see that the message $m$ is new in every message buffer $\getBuf{\gconf{1}^2}{i}{j}$.
    By construction of $\gconf{0}^N$ and $\gconf{1}^N$, it follows that $m \notin \getBuf{\gconf{0}^N}{i}{j}$ and $m \in \getBuf{\gconf{1}^N}{i}{j}$.
    By similar arguments, we have $ \getBuf{\gconf{0}^N}{i}{j} = \{m\} \cup \getBuf{\gconf{1}^N}{i}{j}$.
    Hence, Constraint~(b) between $\gconf{0}^N$  and $\gconf{1}^N$ in the sub-round Send in Section~\ref{sec:global-model} hold.
    Moreover, by the construction of $\gconf{1}^N$, we have $\getActive{\gconf{1}^N}{i} = \getActive{\gconf{1}}{i}$.
    By the semantics of the sub-round Send in Section~\ref{sec:global-model}, we have $\lnot \getActive{\gconf{1}}{i}$.
    It follows $\lnot \getActive{\gconf{1}^N}{i}$.
    Constraint~(c) between $\gconf{0}^N$  and $\gconf{1}^N$ in the sub-round Send in Section~\ref{sec:global-model} holds.
    Therefore, it follows $\gconf{0}^N \transsnd \gconf{1}^N$.
    
    \item  \emph{Case $\gconf{0} \transrcv \gconf{1}$.} It follows by similar arguments of the sub-round Send, except that if $\getProc{\gconf{0}}{i}
    \xrightarrow{\mathit{rcv}(S_1, S_2)} \getProc{\gconf{1}}{i}$, 
    then \[\getProc{\gconf{0}^N}{i} \xrightarrow{\mathit{rcv}(S_1, S_2, \emptyset, \ldots, \emptyset)} \getProc{\gconf{1}^N}{i}\]
    
    \item \emph{Case $\gconf{0} \transcomp \gconf{1}$.} By similar arguments in the case of the sub-round Send.
  \end{itemize}

    Then, Proposition~\ref{prop:IDRefilledTrans} holds. 
    %$\qedhere$
\end{proof}

Notice that in Proposition~\ref{prop:IDRefilledTrans}, if both $p_1^N$ and $p_2^N$ take a stuttering step from $\gconf{0}^N$ to $\gconf{1}^N$, then $\gconf{0}^N = \gconf{1}^N$.

\recalllemma{lem:IDRefilledPath}{\lemIDRefilledPath}

\begin{proof}
  We prove Lemma~\ref{lem:IDRefilledPath}  by inductively constructing $\gconf{k}^N$.
  
  \emph{Base case.} 
  Since $\pi^2$ and $\pi^N$ are admissible sequences, it follows $\gconf{0}^2 = \ginitconf{2}$ and $\gconf{0}^N = \ginitconf{N}$.¸
  By Proposition~\ref{prop:ProjInitGlobalConf}, we have that $\gconf{0}^2$ is the index projection of $\gconf{0}^N$ on indexes $\{1, 2\}$.
  We construct $\gconf{1}^N$ by scheduling that all processes $3..N$ crash in $\gconf{1}^N$. 
  Formally, we have:
  \begin{enumerate}
    \item $\forall i \in 3..N \colon \getPC{\getProc{\gconf{1}^N}{i}}= \loccrash \land \lnot \getActive{\gconf{1}^N}{i}$.
    
    \item $\forall i \in \{1, 2\} \colon\getActive{\gconf{1}^N}{i} = \getActive{\gconf{1}^2}{i} \land \getPC{\getProc{\gconf{1}^N}{i}} = \getPC{\getProc{\gconf{1}^2}{i}}$
    
    \item $\forall s, r \in \nid \colon \getBuf{\gconf{1}^N}{s}{r} = \getBuf{\gconf{0}^N}{s}{r}$
    
    \item $\forall s, r \in \nid \colon \getBox{\getProc{\gconf{1}^N}{s}}{r} = \getBox{\getProc{\gconf{0}^N}{s}}{r}$
    
    \item $\forall s, r \in \nid \colon \getDatum{\getProc{\gconf{1}^N}{s}}{r} = \getDatum{\getProc{\gconf{0}^N}{s}}{r}$
  \end{enumerate} 
%  Informally, $\gconf{1}^N$ agrees with $\gconf{1}^2$ on the schedule for processes $p_1, p_2$ and make others crash.
  The above constraints ensure that $\gconf{0}^N \transsched \gconf{1}^N$.
  
  \emph{Induction step.}  It immediately follows by Proposition~\ref{prop:IDRefilledTrans}. 
  
  Hence, Lemma~\ref{lem:IDRefilledPath} holds. % $\qedhere$
\end{proof}

\recalllemma{lem:TraceEquivTwo}{\lemTraceEquivTwo{}}

\begin{proof}
  It immediately follows by Definition~\ref{def:traceEquivalence}, Lemma~\ref{lem:IDProjPath} and Lemma~\ref{lem:IDRefilledPath}.
\end{proof}

\subsection{Trace Equivalence of $\gsys{1}$ and $\gsys{N}$ under $AP_{\{1\}}$} \label{app:TraceEquivalenceOne}

Lemma~\ref{lem:TraceEquivOne} says that two global transition systems $\gsys{1}$ and $\gsys{N}$ whose processes follow an arbitrary~\classname{} algorithm are trace equivalent under a set $AP_{\{1\}}$ of predicates which inspect only variables whose index is 1.
The proof of Lemma~\ref{lem:TraceEquivOne} is similar to one of Lemma~\ref{lem:TraceEquivOne}, but applies Constructions~\ref{const:TempConfProjOne} and~\ref{const:GlobalConfProjOne}.
Constructions~\ref{const:TempConfProjOne} and~\ref{const:GlobalConfProjOne} are respectively similar to Constructions~\ref{const:TempConfProj} and~\ref{const:GlobalConfProj}, but focus on only an index 1. 
Constructions~\ref{const:TempConfProjOne} and~\ref{const:GlobalConfProjOne} are used in the proof of Lemma~\ref{lem:TraceEquivOne}.

\begin{construction} \label{const:TempConfProjOne}
  Let $\mathcal{A}$ be an arbitrary~\classname{} algorithm.
  Let $\mytemp{N}$ be a process template of $\mathcal{A}$ for some $N \geq 2$, and
  $\tempconf{}^N$ be a template state of $\mytemp{N}$.
  We construct a tuple $\tempconf{}^1 = (pc_1, rcvd_1, v_1)$ based on $\tempconf{}^N$ and a set $\{1\}$ of process indexes in the following way: $pc_1  = \getPC{\tempconf{}^N}$, $rcvd_1 = \getBox{\tempconf{}^N}{1}$, and $v_1 = \getDatum{\tempconf{}^N}{1}$.   
\end{construction}

\begin{construction} \label{const:GlobalConfProjOne}
  Let $\mathcal{A}$ be a~\classname{} algorithm.
  Let $\gsys{1}$ and $\gsys{N}$ be two global transitions of two instances of $\mathcal{A}$ for some $N \geq 1$, and $\gconf{}^N \in \gconfs{N}$ be a global configuration in $\gsys{N}$.
  We construct a tuple $\gconf{}^2 = (s_1, buf_1^1, act_1)$ based on $\gconf{}^N$ and a set $\{1\}$ in the following way: $s_1$ is constructed from  $\getProc{\gconf{}^N}{1}$ with Construction \ref{const:TempConfProjOne} and an index 1, and $buf^1_1 = \getBuf{\gconf{}^N}{1}{1}$, and $act_1 = \getActive{\gconf{}^N}{1} $.
\end{construction}

\newcommand{\lemTraceEquivOne}{
  Let $\mathcal{A}$ be a~\classname{} algorithm. Let $\gsys{1}$ and $\gsys{N}$ be its instances for some $N \geq 2$.
  Let $AP_{\{1\}}$ be a set of predicates which inspect only variables whose index is 1.
  It follows $\gsys{1}$ and $\gsys{N}$ are trace equivalent under $AP_{\{1\}}$.
}

\begin{lem} \label{lem:TraceEquivOne}
  \lemTraceEquivOne
\end{lem}

\begin{proof}
By applying similar arguments in the proof of Lemma~\ref{lem:TraceEquivTwo} with Constructions~\ref{const:TempConfProjOne} and~\ref{const:GlobalConfProjOne}.  %$\qedhere$
\end{proof}

\subsection{Cutoff results in the unrestricted model} \label{app:CutoffOne}

In the following, we prove Propositions~\ref{prop:MoveQuantifiersOne} and~\ref{prop:MoveQuantifiersTwo} which allows us to change positions of big conjunctions in specific formulas.
Propositions~\ref{prop:MoveQuantifiersOne} and~\ref{prop:MoveQuantifiersTwo} are used in the proof of our cutoff results, Theorems~\ref{thm:NCutoffOneIndexNetysA} and~\ref{thm:NCutoffTwoIndexesNetysA}, respectively.

\begin{prop} \label{prop:MoveQuantifiersOne}
  Let $\mathcal{A}$ be a~\classname{} algorithm. 
  Let  $\gsys{N}$ be instances of $N$ processes for some $N \geq 1$.
  Let $\mathit{Path}_N$ be sets of all admissible sequences of configurations  in $\gsys{N}$.
  Let $\omega_{\{i\}}$ be a \ltlx{} formula  in which every predicate takes one of the forms: $P_1(i)$ or $P_2(i, i)$ where $i$ is an index in $\nid$. 
  Then,
  \begin{align}
  & \Big(\forall \pi_N \in \mathit{Path}_N \colon \gsys{N}, \pi_N \models \bigwedge_{i \in \nid}  \omega_{\{i\}} \Big) \label{fm:5_1} \\
  \Leftrightarrow \; & \Big( \bigwedge_{i \in \nid} \big( \forall \pi_N \in \mathit{Path}_N \colon \gsys{N}, \pi_N \models   \omega_{\{i\}} \big) \Big) \label{fm:5_2}
  \end{align}
\end{prop} 

\begin{proof}
  $(\Rightarrow)$ Let $\pi_N$ be an arbitrary admissible sequence of configurations in $\mathit{Path}_N$
  such that $\gsys{N}, \pi_N \models \bigwedge_{i \in \nid}  \omega_{\{i\}}$.
  Let $i_0$ be an arbitrary index in $\nid$, we have $\gsys{N}, \pi_N \models   \omega_{\{i_0\}}$.
  Hence, for every $\pi_N \in \mathit{Path}_N$, for every $i_0 \in \nid$, it follows $\gsys{N}, \pi_N \models   \omega_{\{i_0\}}$.
  Therefore, Formula \ref{fm:5_1} implies: for every $i_0 \in \nid$, for every $\pi_N \in \mathit{Path}_N$,  it follows $\gsys{N}, \pi_N \models   \omega_{\{i_0\}}$.
  It follows that: for every $i_0 \in \nid$, $\forall \pi_N \in \mathit{Path}_N \colon\gsys{N}, \pi_N \models   \omega_{\{i_0\}}$.
  Now, we have that Formula \ref{fm:5_1} implies Formula \ref{fm:5_2}.

  $(\Leftarrow)$ By applying similar arguments. %$\qedhere$
\end{proof}

\begin{prop} \label{prop:MoveQuantifiersTwo}
  Let $\mathcal{A}$ be a~\classname{} algorithm. 
  Let $\gsys{N}$ be instances of $N$ processes respectively for some $N \geq 1$.
  Let $\mathit{Path}_N$ be sets of all admissible sequences of configurations in $\gsys{N}$.
  Let $\omega_{\{i\}}$ be a \ltlx{} formula  in which every predicate takes one of the forms: $P_1(i)$ or $P_2(i, i)$ where $i$ is an index in $\nid$. 
  Then, 
  \begin{align}
  & \big(  \forall \pi_N \in \mathit{Path}_N \colon \gsys{N}, \pi_N \models  \bigwedge^{i \neq j}_{i, j \in \nid} \psi_{\{i, j\}} \big) \label{fm:5_3}  \\
  \Leftrightarrow \; & \Big( \bigwedge^{i \neq j}_{i, j \in \nid} \big(  \forall \pi_N \in \mathit{Path}_N \colon \gsys{N}, \pi_N \models   \psi_{\{i, j\}} \big) \Big) \label{fm:5_4}
  \end{align}
\end{prop} 

\begin{proof}
  $(\Rightarrow)$ Let $\pi_N$ be an arbitrary admissible sequence of configurations in $\mathit{Path}_N$
  such that $\gsys{N}, \pi_N \models \bigwedge^{i \neq j}_{i, j \in \nid} \psi_{\{i, j\}} $.
  Let $i_0$ and $j_0$ be arbitrary indexes in $\nid$ such that $i_0 \neq j_0$, we have that $\gsys{N}, \pi_N \models   \psi_{\{i_0, j_0\}}$.
  Hence, for every $\pi_N \in \mathit{Path}_N$, for every $i_0 \in \nid$, for every $j_0 \in \nid$ such that $i_0 \neq j_0$, it follows that $\gsys{N}, \pi_N \models   \psi_{\{i_0, j_0\}}$.
  Therefore, Formula \ref{fm:5_3} implies: for every $i_0 \in \nid$, for every $j_0 \in \nid$ such that $i_0 \neq j_0$, for every $\pi_N \in \mathit{Path}_N$,  it follows $\gsys{N}, \pi_N \models  \psi_{\{i_0, j_0\}}$.
  It follows that: for every $i_0 \in \nid$, for every $j_0 \in \nid$ such that $i_0 \neq j_0$, it holds $\forall \pi_N \in \mathit{Path}_N \colon\gsys{N}, \pi_N \models   \psi_{\{i_0, j_0\}}$.
  Therefore, Formula \ref{fm:5_3} implies Formula \ref{fm:5_4}.
  
  $(\Leftarrow)$ By applying similar arguments. % $\qedhere$
\end{proof}

\recalltheorem{thm:NCutoffOneIndexNetysA}{\thmNCutoffOneIndexNetysA}

\begin{proof}
%\paragraph{Proof of Theorem~\ref{thm:NCutoffOneIndex}.}
By Proposition~\ref{prop:MoveQuantifiersOne}, we have  
\[
\Big(  \forall \pi_N \in \mathit{Path}_N \colon \gsys{N}, \pi_N \models  \bigwedge_{i \in \nid} \omega_{\{i\}} \Big)  \Leftrightarrow  
\Big(  \bigwedge_{i \in \nid}  \big( \forall \pi_N \in \mathit{Path}_N \colon \gsys{N}, \pi_N \models  \omega_{\{i\}} \big) \Big)
\]
Let $i$ be an  index in a set $\nid$.
Hence, $\transposition = \swapij{i}{1}$  is a transposition on $\nid$ (*). 
%  Let $\pi_N = \gconf{N}^0 \gconf{N}^1 \ldots$ be an admissible sequence of configurations in $\gsys{N}$.
By Lemma~\ref{lem:SymGlobalSys}, we have: (i) $\psi_{\{\switchIndex{\transposition}{i}}  = \psi_{\{1\}}$, and (ii)   $\transposition ( \gsys{N}) ) = \gsys{N}$, and (iii)
$\transposition ( \ginitconf{N})   =  \ginitconf{N}$.

Since $\omega_{\{i\}}$ is an \ltlx{} formula, $\All \omega_{\{i\}}$ is a LTL\textbackslash{}X formula where $\All$ is a path operator in LTL\textbackslash{}X (see~\cite{clarke2018model}). 
By the semantics of the operator $\All$, it follows that $\forall \pi_N \in \mathit{Path}_N \colon \gsys{N}, \pi_N \models  \omega_{\{i\}}$ if and only if $\gsys{N}, \ginitconf{N} \models \All \omega_{\{1\}}$.
By point~(*), it follows $\gsys{N}, \ginitconf{N} \models \All \omega_{\{i\}}$ if and only if $\gsys{N}, \ginitconf{N} \models \All \omega_{\{1\}}$.
Since an index $i$ is arbitrary,
we have
$ 
\gsys{N}, \ginitconf{N} \models \bigwedge_{i \in \nid} \All \omega_{\{i\}}$ if and only if $\gsys{N}, \ginitconf{N} \models \All \omega_{\{i\}}
$.

We have that $\gsys{N}, \ginitconf{N} \models \All \omega_{\{i\}}$ if and only if $\forall \pi_N \in \mathit{Path}_N \colon \gsys{N}, \pi_N \models \omega_{\{i\}}$ by the semantics of the operator $\All$.  
It follows $\forall \pi_N \in \mathit{Path}_N \colon \gsys{N}, \pi_N \models \omega_{\{i\}}$ if and only if $\forall \pi_2 \in \mathit{Path}_2 \colon \gsys{2}, \pi_2 \models \omega_{\{1\}}$ 
by Lemma \ref{lem:TraceEquivOne}.
Then, Theorem \ref{thm:NCutoffOneIndexNetysA} holds.
% $\qedhere$
\end{proof}

\recalltheorem{thm:NCutoffTwoIndexesNetysA}{\thmNCutoffTwoIndexesNetysA}

\begin{proof}
By similar arguments in the proof of Theorem~\ref{thm:NCutoffOneIndexNetysA}.
\end{proof}

\subsection{Verification of the Failure Detector of~\cite{CT96} with the Cutoffs} \label{app:Application}

In the following, we present Lemmas~\ref{lem:StrongCompletness} which explains why the cutoff result~\ref{thm:NCutoffTwoIndexesNetysA} allows us to verify the strong completeness property of the failure detector of~\cite{CT96} under synchrony by model checking instances of size 2.

\newcommand{\propMoveQuantifiers}{
  Let $\gsys{N} = (\gconfs{N}, \gtr{N}, \grel{N}, \ginitconf{N})$ be a global transition system of a symmetric point-to-point algorithm under the unrestricted model. 
  Its indexes are $\nid$ for some $N \geq 1$.
  Let $i$ and $j$ be two indexes in the set $\nid$.
  Let $\mu_{\{i, j\}}$ be a first-order formula in which every predicate takes one of the forms: $Q_1(i)$, or $Q_2(j)$, or $Q_3(i, j)$, or $Q_4(j, i)$.
  The following conditions hold:
  \begin{enumerate}[nosep]
  \item $\F \G   \bigwedge_{i, j \in \nid} \mu_{\{i, j\}}$ be an \ltlx{} formula.
  
  \item $\bigwedge_{i, j \in \nid} \F \G  \mu_{\{i, j\}}$ be an \ltlx{} formula.
  
  \item Let $\pi = \gconf{0} \gconf{1} \ldots$ be an admissibe sequence of configurations in $\gsys{N}$. 
  It follows  $\gsys{N}, \pi \models \F \G  \bigwedge_{i, j \in \nid} \mu_{\{i, j\}}$
  if and only if $\gsys{N}, \pi \models \bigwedge_{i, j \in \nid} \F \G  \mu_{\{i, j\}}$.
  \end{enumerate}  
}

\begin{prop} \label{prop:MoveQuantifiers}
\propMoveQuantifiers
\end{prop}

\begin{proof} 
Points (1) and (2) hold by the definition of \ltlx{} (see~\cite{clarke2018model}).
We prove Point (3) as follows.

($\Rightarrow$) Since $\gsys{N}, \pi \models \F \G  \bigwedge_{i, j \in \nid} \mu_{\{i, j\}}$, there exists $\ell_0 \geq 0$ such that for every $\ell \geq \ell_0$, we have $\gconf{\ell}  \models   \bigwedge_{i, j \in \nid} \mu_{\{i, j\}}$.
Let $i_0$ and $j_0$ be two indexes in $\nid$.
We have $\gconf{\ell}  \models  \mu_{\{i_0, j_0\}}$ for every $\ell \geq \ell_0$.
Hence, it follows $\gsys{N}, \pi \models \F \G  \mu_{\{i_0, j_0\}}$.
Because $i_0$ and $j_0$ are arbitrary indexes in $\nid$, it follows that $\gsys{N}, \pi \models \bigwedge_{i, j \in \nid} \F \G  \mu_{\{i, j\}}$.

($\Leftarrow$)  Let $i_0$ and $j_0$ be two indexes in $\nid$. 
Since $\gsys{N}, \pi \models \bigwedge_{i, j \in \nid} \F \G  \mu_{\{i, j\}}$, it follows that 
$\gsys{N}, \pi \models \F \G  \mu_{\{i_0, j_0\}}$.
Therefore, there exists $\ell^i_j \geq 0$ such that $\gconf{\ell}  \models  \mu_{\{i_0, j_0\}}$ for every $\ell \geq \ell^{i_0}_{j_0}$.
Let $\ell = \mathit{max}(\{\ell^{i}_{j} \colon i \in \nid \land j \in \nid\})$ where $\mathit{max}$ is a function to pick a maximum number in a finite set of natural numbers.
It follows that $\gconf{\ell}  \models  \mu_{\{i, j\}}$ for every $\ell \geq \ell_0$, for every $i, j \in \nid$.
Therefore, $\gsys{N}, \pi \models \F \G  \bigwedge_{i, j \in \nid} \mu_{\{i, j\}}$.
\end{proof}

\newcommand{\lemStrongCompletness}{
  Let $\gsys{2}$ and $\gsys{N}$ be two global transition systems of a symmetric point-to-point algorithm such that:
  \begin{enumerate}[nosep]
    \item These systems $\gsys{2}$ and $\gsys{N}$ have 2 and $N$ processes respectively where $N \geq 2$.
    \item All processes $\gsys{2}$ and $\gsys{N}$ follow Algorithm~\ref{algo:detector1}.
    \item The model of computation of these systems is 
%    under synchrony.
under the unrestricted model. 
  \end{enumerate}
  Let $\Proc{}$ be a set of indexes, and $\mu(\Pi)$ denote the strong completeness property in which the set of process indexes is $\Pi$, i.e.
  \[
  \mu(\Pi) \triangleq
  \F \G (\A p, q \in \Pi : (\mathit{Correct}(p) \land \lnot\mathit{Correct}(q)) \Rightarrow \mathit{Suspected}(p, q))  
  \]
  Let $\mathit{Path}_2$ and $\mathit{Path}_N$ be sets of all admissible sequences of configurations in $\gsys{2}$ and in $\gsys{N}$, respectively.
  Then, it holds: 
  \begin{align*}
  & \forall \pi_N \in \mathit{Path}_N \colon \gsys{N}, \pi_N \models \mu(\nid)  \\
  \Leftrightarrow & \forall \pi_2 \in \mathit{Path}_2 \colon \gsys{2}, \pi_2  \models \mu(\tid)    
  \end{align*}
  
}

\begin{lem} \label{lem:StrongCompletness}
  \lemStrongCompletness
\end{lem}

\begin{proof}
  To keep the presentation simple, let $\nu(p, q)$ be a predicate such that $\nu(p, q) \triangleq (\mathit{Correct}(p) \land \lnot\mathit{Correct}(q))$.
  Let $\pi_N$ be an admissible sequence of configurations in $\gsys{N}$.
  We have 
  $\gsys{N}, \pi_N \models \mu(\Pi)$ if and only if 
  $\gsys{N}, \pi_N \models \F \G (\A p, q \in \nid \colon \nu(p, q) \Rightarrow \mathit{Suspected}(p, q))$.
  It follows
  \begin{alignat*}{3}
  &  \gsys{N}, \pi_N \models \mu(\nid) && \\
  \Leftrightarrow \; &  \gsys{N}, \pi_N \models \F \G (\bigwedge_{p, q \in \nid}  \nu(p, q) \Rightarrow \mathit{Suspected}(p, q))  && \\
  \Leftrightarrow \; &  \gsys{N}, \pi_N \models \bigwedge_{p, q \in \nid}  \F \G ( \nu(p, q) \Rightarrow \mathit{Suspected}(p, q)) && \quad \text{(by Proposition~\ref{prop:MoveQuantifiers})} 
  \end{alignat*}
  The last formula is equivalent to
  \begin{alignat*}{3}
  &  \gsys{N}, \pi_N \models \bigwedge^{p \neq q}_{p, q \in \nid}  \F \G ( \nu(p, q) \Rightarrow \mathit{Suspected}(p, q)) && \\
  \land \; &  \gsys{N}, \pi_N \models \bigwedge^{p = q}_{p, q \in \nid}  \F \G ( \nu(p, q) \Rightarrow \mathit{Suspected}(p, q)) && 
  \end{alignat*}
  For every $p, q \in \nid$, if $p = q$, then $(\mathit{Correct}(p) \land \lnot\mathit{Correct}(q)) = \False$ (*).
  Hence, it follows that
  \begin{alignat*}{3}
  &  \gsys{N}, \pi_N \models \bigwedge_{p, q \in \nid}  \F \G ( \nu(p, q) \Rightarrow \mathit{Suspected}(p, q)) && \\
  \Leftrightarrow \; & \gsys{N}, \pi_N \models \bigwedge^{p \neq q}_{p, q \in \nid}  \F \G ( \nu(p, q) \Rightarrow \mathit{Suspected}(p, q)) && 
  \end{alignat*}
  It follows that $\forall \pi_N \in \mathit{Path}_N \colon \gsys{N}, \pi_N \models \mu$ if and only if 
  \[
  \forall \pi_N \in \mathit{Path}_N \colon \gsys{N}, \pi_N \models \bigwedge^{p \neq q}_{p, q \in \nid}  \F \G ( \nu(p, q) \Rightarrow \mathit{Suspected}(p, q))
  \]
  By Theorem~\ref{thm:NCutoffTwoIndexesNetysA}, it follows 
  \begin{align*}
   \forall \pi_N \in \mathit{Path}_N \colon \gsys{N}, \pi_N & \models \mu(\nid) \\
  \Leftrightarrow \;  \forall \pi_2 \in \mathit{Path}_2 \colon \gsys{2}, \pi_2 & \models \bigwedge^{p \neq q}_{p, q \in \tid}  \F \G ( \nu(p, q) \Rightarrow \mathit{Suspected}(p, q))
  \end{align*}
  By point (*), we have
  \begin{align*}
  \forall \pi_N \in \mathit{Path}_N \colon \gsys{N}, \pi_N & \models \mu(\nid) \\
  \Leftrightarrow \; \forall \pi_2 \in \mathit{Path}_2 \colon \gsys{2}, \pi_2 &  \models \bigwedge_{p, q \in \tid}  \F \G ( \nu(p, q) \Rightarrow \mathit{Suspected}(p, q))
  \end{align*}
  By Proposition~\ref{prop:MoveQuantifiers}, it follows
  \begin{align*}
  \forall \pi_N \in \mathit{Path}_N \colon \gsys{N}, \pi_N & \models \mu(\nid) \\
  \Leftrightarrow \; \forall \pi_2 \in \mathit{Path}_2 \colon \gsys{2}, \pi_2 &  \models   \F \G ( \bigwedge_{p, q \in \tid} (\nu(p, q) \Rightarrow \mathit{Suspected}(p, q)))  \\
  \Leftrightarrow \; \forall \pi_2 \in \mathit{Path}_2 \colon \gsys{2}, \pi_2 &  \models \mu(\tid)  
  \end{align*}
  Hence, Lemma~\ref{lem:StrongCompletness} holds. % $\qedhere$
\end{proof}

\newcommand{\lemEventuallyStrongAccuracy}{
  Let $\gsys{1}$ and $\gsys{2}$ and $\gsys{N}$ be three global transition systems of a symmetric point-to-point algorithm such that:
  \begin{enumerate}[nosep]
    \item These systems $\gsys{1}$ and $\gsys{2}$ and $\gsys{N}$ have 1, 2 and $N$ processes respectively.
    \item All processes in $\gsys{1}$ and $\gsys{2}$ and $\gsys{N}$ follow Algorithm~\ref{algo:detector1}.
    
    \item Three sets $\mathit{Path}_1$ and $\mathit{Path}_2$ and $\mathit{Path}_N$ be sets of admissible sequences of configurations in $\gsys{1}$ and $\gsys{2}$ and $\gsys{N}$, respectively.
    \item The model of computation of these systems is %under synchrony.
    under the unrestricted model. 
  \end{enumerate}
  Let $\Proc{}$ be a set of process indexes, and $\mu(\Pi)$ denote the eventually strong accuracy property in which the set of process indexes is $\Pi$, i.e.,
  \[
  \mu(\Pi) \triangleq
  \F \G (\A p, q \in \Pi \colon (\mathit{Correct}(p) \land \mathit{Correct}(q)) \Rightarrow \lnot \mathit{Suspected}(p, q))
  \]
  It follows   $\forall \pi_N \in \mathit{Path}_N \colon \gsys{N}, \pi_N \models \mu(\nid)$  
  if and only if both $\forall \pi_2 \in \mathit{Path}_2 \colon \gsys{2}, \pi_2  \models \mu(\tid)$ and $\forall \pi_1 \in \mathit{Path}_1 \colon \gsys{1}, \pi_1  \models \mu(1..1)$.
}

\begin{lem} \label{lem:EventuallyStrongAccuracy}
\lemEventuallyStrongAccuracy
\end{lem}

\begin{proof}
To keep the presentation simple, we define 
\[
\nu(p, q) = (\mathit{Correct}(p) \land \mathit{Correct}(q)) \Rightarrow \lnot \mathit{Suspected}(p, q)
\]
By similar arguments in Proposition~\ref{lem:StrongCompletness}, we have  $\forall \pi_N \in \mathit{Path}_N \colon \gsys{N}, \pi_N \models \mu(\nid)$ is equivalent to the following conjunction
\begin{align*}
& \Big( \forall \pi_N \in \mathit{Path}_N \colon  \gsys{N}, \pi_N \models \bigwedge^{p = q}_{p, q \in \nid} \F \G \nu(p, q) \Big)\\
\; \land \; &
\Big( \forall \pi_N \in \mathit{Path}_N \colon \gsys{N}, \pi_N \models \bigwedge^{p \neq q}_{p, q \in \nid}  \F \G \nu(p, q) \Big)
\end{align*}

By Theorems~\ref{thm:NCutoffOneIndexNetysA} and~\ref{thm:NCutoffTwoIndexesNetysA}, the above conjunction is equivalent to
\begin{align*}
& \Big( \forall \pi_1 \in \mathit{Path}_1 \colon  \gsys{1}, \pi_1 \models \bigwedge^{p = q}_{p, q \in \nid}  \F \G \nu(1, 1) \Big) \\
\; \land \; &
\Big( \forall \pi_2 \in \mathit{Path}_2 \colon \gsys{2}, \pi_2 \models \bigwedge^{p \neq q}_{p, q \in \nid}  \F \G \nu(1, 2) \Big)
\end{align*}

By Proposition~\ref{prop:MoveQuantifiers}, the above conjunction is equivalent to 
\begin{align*}
& \Big( \forall \pi_1 \in \mathit{Path}_1 \colon  \gsys{1}, \pi_1 \models  \F \G  \bigwedge^{p = q}_{p, q \in \nid} \nu(1, 1) \Big) \\
\; \land \; &
\Big( \forall \pi_2 \in \mathit{Path}_2 \colon \gsys{2}, \pi_2 \models   \F \G \bigwedge^{p \neq q}_{p, q \in \nid} \nu(1, 2) \Big)
\end{align*}
Therefore, Lemma~\ref{lem:EventuallyStrongAccuracy} holds.
\end{proof}

\section{Cutoff Results in the Case of Unknown Time Bounds}  \label{sec:forte21-cutoff-res}
In this section, we extend the above cutoff results on a number of processes (see Theorems~\ref{thm:NCutoffOneIndexNetysA} and \ref{thm:NCutoffTwoIndexesNetysA})  for partial synchrony in case of unknown bounds $\Delta$ and $\Phi$.
The extended results are formalized in Theorems~\ref{thm:NCutoffOneIndex} and~\ref{thm:NCutoffTwoIndexes}.
It is straightforward to adapt our approach to other models of partial synchrony in~\cite{DLS88,CT96}.

\newcommand{\thmNCutoffOneIndex}{
	Let $\mathcal{A}$ be a~\classname{} algorithm under partial synchrony with unknown bounds $\Delta$ and $\Phi$. 
	Let $\gsys{1}$ and $\gsys{N}$ be instances of $\mathcal{A}$ with 1 and $N$ processes respectively for some $N \geq 1$.
	Let $\mathit{Path}_1$ and $\mathit{Path}_N$ be sets of all admissible sequences of configurations in $\gsys{1}$ and in $\gsys{N}$ under partial synchrony, respectively.
	Let $\omega_{\{i\}}$ be a \ltlx{} formula  in which every predicate takes one of the forms: $P_1(i)$ or $P_2(i, i)$ where $i$ is an index in $\nid$. 
	It follows that: 
  \[
	\Big(\forall \pi_N \in \mathit{Path}_N \colon \gsys{N}, \pi_N \models \bigwedge_{i \in \nid}  \omega_{\{i\}} \Big) \; \Leftrightarrow \; \Big(\forall \pi_1 \in \mathit{Path}_1 \colon \gsys{1}, \pi_1 \models  \omega_{\{1\}}\Big)
  \]
}

\begin{thm} \label{thm:NCutoffOneIndex}
	\thmNCutoffOneIndex
\end{thm}

\newcommand{\thmNCutoffTwoIndexes}{
	Let $\mathcal{A}$ be a~\classname{} algorithm under partial synchrony with unknown bounds $\Delta$ and $\Phi$. 
	Let $\gsys{2}$ and $\gsys{N}$ be instances of $\mathcal{A}$ with 2 and $N$ processes respectively for some $N \geq 2$.
	Let $\mathit{Path}_2$ and $\mathit{Path}_N$ be sets of all admissible sequences of configurations in $\gsys{2}$ and in $\gsys{N}$ under partial synchrony, respectively.
	Let $\psi_{\{i, j\}}$ be an \ltlx{} formula in which every predicate takes one of the forms: $Q_1(i)$, or $Q_2(j)$, or $Q_3(i, j)$, or $Q_4(j, i)$ where $i$ and $j$ are different indexes in $\nid$.
	It follows that:
	\[
  \big(  \forall \pi_N \in \mathit{Path}_N \colon \gsys{N}, \pi_N \models  \bigwedge^{i \neq j}_{i, j \in \nid} \psi_{\{i, j\}} \big)  \Leftrightarrow  \big(  \forall \pi_2 \in \mathit{Path}_2 \colon \gsys{2}, \pi_2 \models    \psi_{\{1, 2 \}} \big) 
  \]

}

\begin{thm} \label{thm:NCutoffTwoIndexes}
	\thmNCutoffTwoIndexes  
\end{thm}

Since the proofs of these theorems are similar, we here
focus on only Theorem~\ref{thm:NCutoffTwoIndexes}. 
The proof of Theorem~\ref{thm:NCutoffTwoIndexes} follows the approach in~\cite{emerson1995reasoning,tran2020netys}, and is based on the following observations.
Remind that Steps 1 and 2 are already proved in Section~\ref{sec:netys20-cutoff-res}.

\begin{enumerate}
	\item The global transition system and the desired property are symmetric.
	
	\item Let $\gsys{2}$ and $\gsys{N}$ be two instances of a symmetric point-to-point algorithm with 2 and $N$ processes, respectively.
	We have that two instances $\gsys{2}$ and $\gsys{N}$  are trace equivalent under a set of predicates in the desired property.
	
	\item We will now discuss that the constraints maintain partial synchrony. 
	Let $\pi_N$ be an execution in $\mathcal{G}_N$. By applying the index projection to $\pi_N$, we obtain an execution $\pi_2$ in $\mathcal{G}_2$.
	If partial synchrony constraints~\ref{PS1} and \ref{PS2} -- defined in Section~\ref{sec:time-constraints} -- hold on $\pi_N$, these constraints also hold on~$\pi_2$. 
	This result is proved in Lemma~\ref{lem:ps1}.
	
	\item Let $\pi_2$ be an execution in $\mathcal{G}_2$. 
	We construct an execution $\pi_N$ in $\mathcal{G}_N$ based on $\pi_2$ such that all processes $3..N$ crash from the beginning, and $\pi_2$ is an index projection of $\pi_N$ (defined in Section~\ref{subsec:traceequiv}). 
	For instance, see Figure~\ref{fig:2to3}.
	If partial synchrony constraints~\ref{PS1} and \ref{PS2} -- defined in Section~\ref{sec:time-constraints} -- hold on $\pi_2$, these constraints also hold on $\pi_N$. 
	This result is proved in Lemma~\ref{lem:ps2}.
\end{enumerate}

\begin{lem} \label{lem:ps1}
Let $\mathcal{A}$ be a~\classname{} algorithm under partial synchrony with unknown bounds $\Delta$ and $\Phi$. 
Let $\gsys{2}$ and $\gsys{N}$ be instances of $\mathcal{A}$ with 2 and $N$ processes, respectively, for some $N \geq 2$.
Let $\mathit{Path}_2$ and $\mathit{Path}_N$ be sets of all admissible sequences of configurations in $\gsys{2}$ and in $\gsys{N}$ under partial synchrony, respectively.
Let $\pi^N = \gconf{0}^N \gconf{1}^N \ldots$ be an admissible sequences of configurations in $\gsys{N}$.
Let $\pi^2 = \gconf{0}^2 \gconf{1}^2 \ldots$ be a sequence of configurations in $\gsys{2}$ such that $\gconf{k}^2$ be an index projection of $\gconf{k}^N$ on indexes $\{ 1, 2\}$ for every $k \geq 0$.
It follows that
\begin{enumerate}[label=(\alph*)]
\item Constraint~\ref{PS1} on message delay 
holds on $\pi^2$.
\item Constraint~\ref{PS2} on the relative speed of processes 
holds on $\pi^2$.
\end{enumerate}

\end{lem}

\begin{proof}
Recall that the index projection is defined in Section~\ref{subsec:traceequiv}.
In the following, we denote $p^2$ and $p^N$ two processes such that they have the same index, and $p^2$ is a process in $\gsys{2}$, and $p^N$ is a process in $\gsys{N}$.
We prove Lemma~\ref{lem:ps1} by contradiction.

(a)
Assume that Constraint~\ref{PS1} does not hold on $\pi^2$.
Hence, there exist a time $\ell > 0$, and two processes $s^2, r^2 \in \tid$ in $\gsys{2}$ such that after $r^2$ executes Receive at a time $\ell$, there exists an old message in a message buffer from process $s^2$ to process $r^2$.
By the definition of the index projection, for every $k \geq 0$, we have that:
\begin{itemize}
	\item Let $p^N, q^N \in \{1, 2\}$ be two processes in $\gsys{N}$, and $p^2, q^2$ be corresponding processes in $\gsys{2}$.
	For every $k \geq 0$, two message buffers from process $p^2$ to process $q^2$ in $\gconf{k}^2$, and from process $p^N$ to process $q^N$ in $\gconf{k}^N$ are the same.
	
	\item Let $p^N \in \{1, 2\}$ be a process in $\gsys{N}$, and $p^2$ be a corresponding process in $\gsys{2}$. 
	For every $k \geq 0$, process $p^2$ takes an action \textit{act} in configuration $\gconf{k}^2$ if and only if process $p^N$ takes the same action in configuration $\gconf{k}^N$.
\end{itemize}
It implies that process $r^N$ in $\gsys{N}$ also executes Receive at a time $\ell$, and there exists an old message in a buffer from process $s^N$ to process $r^N$.
Contradiction. 

(b) By applying similar arguments in case (a). % $\qedhere$
\end{proof}

\begin{lem} \label{lem:ps2}
	Let $\mathcal{A}$ be a~\classname{} algorithm under partial synchrony with unknown bounds $\Delta$ and $\Phi$. 
	Let $\gsys{2}$ and $\gsys{N}$ be instances of $\mathcal{A}$ with 2 and $N$ processes, respectively, for some $N \geq 2$.
	Let $\mathit{Path}_2$ and $\mathit{Path}_N$ be sets of all admissible sequences of configurations in $\gsys{2}$ and in $\gsys{N}$ under partial synchrony, respectively.
	Let $\pi^2 = \gconf{0}^2 \gconf{1}^2 \ldots$ be an admissible sequences of configurations in $\gsys{2}$.
	Let $\pi^N = \gconf{0}^N \gconf{1}^N \ldots$ be a sequence of configurations in $\gsys{N}$ such that (i) every process $p \in 3..N$ crashes from the beginning, and (ii) $\gconf{k}^2$ be an index projection of $\gconf{k}^N$ on indexes $\{ 1, 2\}$ for every $k \geq 0$.
	It follows that
	\begin{enumerate}[label=(\alph*)]
		\item Constraint~\ref{PS1} on message delay 
    holds on $\pi^N$.
		\item Constraint~\ref{PS2} on the relative speed of processes 
    holds on $\pi^N$.
	\end{enumerate}
\end{lem}

\begin{proof}
By applying similar arguments in the proof of Lemma~\ref{lem:ps1}, and the facts that every process $p \in 3..N$ crashes from the beginning, and that $\gconf{k}^2$ be an index projection of $\gconf{k}^N$ on indexes $\{ 1, 2\}$ for every $k \geq 0$. %$\qedhere$
\end{proof}

\section{Encoding the Chandra and Toueg Failure Detector} \label{sec:forte21-encoding}
In this section, we first discuss why it is sufficient to verify the failure detector by checking a system with only one sender and one receiver by applying the cutoffs presented in Section~\ref{sec:forte21-cutoff-res}.
Next, we introduce two approaches to encoding the message buffer, and an abstraction of in-transit messages that are older than $\Delta$ time-units. 
Finally, we present how to encode the relative speed of processes with counters over natural numbers.
These techniques allow us to tune our models to the strength of the verification tools: \fast{}, IVy, and model checkers for~\tlap{}.

\subsection{The System with One Sender and One Receiver}

The cutoff results in Section~\ref{sec:forte21-cutoff-res} allow us to verify the Chandra and Toueg failure detector under partial synchrony by checking only instances with two processes.
In the following, we discuss the model with two processes, and formalize the properties with two-process indexes.
By process symmetry, it is sufficient to verify Strong Accuracy, Eventually Strong Accuracy, and Strong Completeness by checking the following properties.
\begin{align}
\G ((\mathit{Correct}(1) \land \mathit{Correct}(2)) & \Rightarrow \lnot \mathit{Suspected}(2, 1)) \label{eq:simpleSA} \\
\F \G ((\mathit{Correct}(1) \land \mathit{Correct}(2)) & \Rightarrow \lnot \mathit{Suspected}(2, 1)) \label{eq:simpleESA} \\
\F \G ((\lnot \mathit{Correct}(1) \land \mathit{Correct}(2)) & \Rightarrow \mathit{Suspected}(2, 1)) \label{eq:simpleSC}
\end{align}

We can take a further step towards facilitating verification of the failure detector.
First, every process typically has a local variable to store messages that it needs to send to itself, instead of using a real communication channel.
Hence, we can assume that there is no delay for those messages, and that each correct process never suspects itself.
Second, local variables in Algorithm~\ref{algo:detector1} are arrays whose elements correspond one-to-one with a remote process, e.g., $\mathit{timeout}[2, 1]$
and $\mathit{suspected}[2, 1]$.
Third, communication between processes is point-to-point.
When this is not the case, one can use cryptography to establish one-to-one communication.
Hence, reasoning about Properties~\ref{eq:simpleSA}--\ref{eq:simpleSC} requires no information about messages from process 1 to itself, local variables of process 1, and messages from process 2.

Due to the above characteristics, it is sufficient to consider process 1 as a sender, and process 2 as a receiver. 
In detail, the sender follows Task 1 in Algorithm~\ref{algo:detector1}, but does nothing in Task 2 and Task 3.
The sender does not need the initialization step, and local variables \textit{suspected} and \textit{timeout}.
In contrast, the receiver has local variables corresponding to the sender, and follows only the initialization step, and Task 2, and Task 3 in Algorithm~\ref{algo:detector1}.
The receiver can increase its waiting time in Task 1, but does not send any message.

\subsection{Encoding the Message Buffer}

Algorithm~\ref{algo:detector1} assumes unbounded message buffers between processes that produce an infinite state space.
Moreover, a sent message might be in-transit for a long time before it is delivered.
We first introduce two approaches to encode the message buffer based on a logical predicate, and a counter over natural numbers.
The first approach works for~\tlap{} and IVy, but not for counter automata (\fast{}).
The latter is supported  by all mentioned tools, but it is less efficient as it requires more transitions.
Then, we present an abstraction of in-transit messages that are older than $\Delta$ time-units. 
This technique reduces the state space, and allows us to tune our models to the strength of the verification tools.

\subsubsection{Encoding the message buffer with a predicate.}

In Algorithm~\ref{algo:detector1}, only ``alive'' messages are sent, and the message delivery depends only on the age of in-transit messages.
Moreover, the computation of the receiver does not depend on the contents of its received messages.
Hence, we can encode a message buffer by using a logical predicate $\existsMsgOfAge(x)$.
For every $k \geq 0$, predicate $\existsMsgOfAge(k)$
refers to whether there exists an in-transit message
that is $k$ time-units old.
The number 0 refers to the age of a fresh message in the buffer.

It is convenient to encode the message buffer\textquotesingle{}s behaviors in this approach. For instance, Formulas~\ref{eq:add1} and~\ref{eq:add2} show constraints on the message buffer when a new message is sent:
\begin{align}
& \existsMsgOfAge^{\prime}(0)  \label{eq:add1} \\
& \forall x \in \mathbb{N} \qdot x > 0 \Rightarrow \existsMsgOfAge^{\prime}(x) = \existsMsgOfAge(x) \label{eq:add2}
\end{align}
where $\existsMsgOfAge^{\prime}$ refers to the value of \existsMsgOfAge{} in the next state.
Formula~\ref{eq:add1} implies that a fresh message has been added to the message buffer.
Formula~\ref{eq:add2} ensures that other in-transit messages are unchanged.

Another example is the relation between $\existsMsgOfAge$ and $\existsMsgOfAge^{\prime}$ after the message delivery.
This relation is formalized with Formulas~\ref{eq:msgDel1-noabstraction}--\ref{eq:msgDel4-noabstraction}.
Formula~\ref{eq:msgDel1-noabstraction} requires that there exists an in-transit message in $\existsMsgOfAge$ that can be delivered.
Formula~\ref{eq:msgDel2-noabstraction} ensures that no old messages are in transit after the delivery.
Formula~\ref{eq:msgDel3-noabstraction} guarantees that no message is created out of thin air.
Formula~\ref{eq:msgDel4-noabstraction} implies that at least one message is delivered.\begin{align}
& \exists x \in \mathbb{N} \qdot \existsMsgOfAge(x) \label{eq:msgDel1-noabstraction} \\
& \forall x \in \mathbb{N} \qdot x \geq \Delta \Rightarrow \lnot \existsMsgOfAge^{\prime}(x) \label{eq:msgDel2-noabstraction} \\
& \forall x \in \mathbb{N} \qdot  \existsMsgOfAge^{\prime}(x) \Rightarrow \existsMsgOfAge(x) \label{eq:msgDel3-noabstraction} \\
& \exists x \in \mathbb{N} \qdot  \existsMsgOfAge^{\prime}(x) \neq \existsMsgOfAge(x) \label{eq:msgDel4-noabstraction}
\end{align}

This encoding works for~\tlap{} and IVy, but not for~\fast{}, because the input language of~\fast{} does not support functions.

\subsubsection{Encoding the message buffer with a counter.}
In the following, we present an encoding technique for the buffer that can be applied in all tools~\tlap{}, IVy, and~\fast{}.
This approach encodes the message buffer with a counter \texttt{buf} over natural numbers.
The $k^{\mathit{th}}$ bit refers to whether there exists an in-transit message with $k$ time-units old.

In this approach, message behaviors are formalized with operations in Presburger arithmetic. For example, assume $\Delta > 0$, we write $\texttt{buf}^{\prime} = \texttt{buf} + 1$ to add a fresh message in the buffer.
Notice that the increase of \texttt{buf} by 1 turns on the $0^{th}$ bit, and keeps the other bits unchanged.

To encode the increase of the age of every in-transit message by 1, we simply write $\texttt{buf}^{\prime} = \texttt{buf} \times 2$.
%(or $\texttt{buf}^{\prime} = \texttt{buf} << 1$ where $<<$ is the left-shift operator).
Assume that we use the least significant bit (LSB) first encoding, and the left-most bit is the $0^{th}$ bit. 
By multiplying \texttt{buf} by 2, we have updated $\texttt{buf}^{\prime}$ by shifting to the right every bit in \texttt{buf} by 1.
For example, Figure~\ref{fig:incMsgAge} demonstrates the message buffer after the increase of message ages in case of $\texttt{buf} = 6$.
We have $\texttt{buf}^{\prime} = \texttt{buf} \times 2 = 12$.
It is easy to see that the $1^{st}$ and $2^{nd}$ bits in \texttt{buf} are on, and the $2^{nd}$ and $3^{rd}$ bits in $\texttt{buf}^{\prime}$ are on.

\begin{figure}[tb]

\begin{tikzpicture}
\matrix (array) [matrix of nodes,
nodes={draw, minimum size=4mm,anchor=center},
nodes in empty cells, minimum height = 3mm,
row 1/.style={nodes={draw=none, fill=none}}]
{
	0 & 1 & 2 & 3 & 4 & $\ldots$  \\
	&   &   &  &    \\
};
\node[minimum size=2mm] at (array-2-3) (box) {\includegraphics[width=6mm]{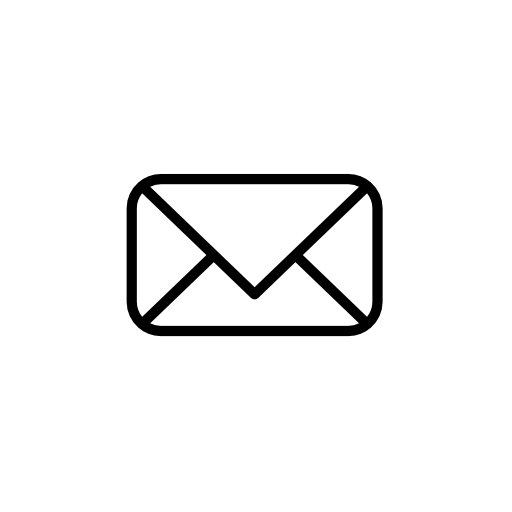}};
\node[minimum size=2mm] at (array-2-2) (box) {\includegraphics[width=6mm]{./fig/email-message-icon.jpg}};
\draw (array-1-1.east)++(0:-5mm) node [left]{Age indexes};
\draw (array-2-1.east)++(0:-5mm) node [left]{Messages in \texttt{buf}};
\end{tikzpicture}
\hspace{5mm}
\begin{tikzpicture}
\matrix (array) [matrix of nodes,
nodes={draw, minimum size=4mm,anchor=center},
nodes in empty cells, minimum height = 3mm,
row 1/.style={nodes={draw=none, fill=none}}]
{
	0 & 1 & 2 & 3 & 4 & $\ldots$  \\
	&   &   &  &    \\
};
\draw (array-1-1.east)++(0:-5mm) node [left]{Age indexes};
\draw (array-2-1.east)++(0:-5mm) node [left]{Messages in $\texttt{buf}^{\prime}$};
\node[minimum size=2mm] at (array-2-4) (box) {\includegraphics[width=6mm]{./fig/email-message-icon.jpg}};
\node[minimum size=2mm] at (array-2-3) (box) {\includegraphics[width=6mm]{./fig/email-message-icon.jpg}};
\end{tikzpicture}
%}
\caption{The message buffer after increasing message ages in case of $\texttt{buf} = 6$} \label{fig:incMsgAge}
\end{figure}

Recall that Presburger arithmetic does not allow one to divide by a variable.
Therefore, to guarantee the constraint in Formula~\ref{eq:msgDel3-noabstraction},
we need to enumerate all constraints on possible  values of \texttt{buf} and $\texttt{buf}^{\prime}$ after the message delivery.
For example, assume $\texttt{buf} = 3$, and $\Delta = 1$. After the message delivery, $\texttt{buf}^{\prime}$ is either 0 or 1.
If $\texttt{buf} = 2$ and $\Delta = 1$, $\texttt{buf}^{\prime}$ must be 0 after the message delivery.
Importantly, the number of transitions for the message delivery depends on the value of $\Delta$.

To avoid the enumeration of all possible cases, Formula~\ref{eq:msgDel3-noabstraction} can be rewritten with bit-vector arithmetic.
However, bit-vector arithmetic are currently not supported in all verification tools~\tlap{}, \fast{}, and IVy.

The advantage of this encoding is that when bound $\Delta$ is fixed,
every constraint in the system behaviors can be rewritten in Presburger arithmetic.
Thus, we can use~\fast{}, which accepts constraints in Presburger arithmetic.
To specify cases with arbitrary $\Delta$, the user can use~\tlap{} or IVy.

\subsubsection{Abstraction of old messages.}
Algorithm~\ref{algo:detector1} assumes underlying unbounded message buffers between processes. 
Moreover, a sent message might be in transit for a long time before it is delivered. 
To reduce the state space, we develop an abstraction of in-transit messages that are older than $\Delta$ time-units; we call such messages ``old''.
This abstraction makes the message buffer between the sender and the receiver bounded.
In detail, the message buffer has a size of $\Delta$.
Importantly, we can apply this abstraction to two above encoding techniques for the message buffer.

In partial synchrony, if process $p$ executes Receive at some time point from the Global Stabilization Time, \emph{every} old message sent to $p$ will be delivered immediately.
Moreover, the computation of a process in Algorithm~\ref{algo:detector1} does not depend on the content of received messages.
Hence, instead of tracking all old messages, our abstraction keeps only one old message that is $\Delta$ time-units old, does not increase its age, and throws away other old messages.

In the following, we discuss how to integrate this abstraction into the  encoding techniques of the message buffer.
We demonstrate our ideas by showing the pseudo code of the increase of message ages.
It is straightforward to adopt this abstraction to the message delivery, and to the sending of a new message.

Figure~\ref{fig:incMsgAge-abstraction}(a) presents the increase of message ages with this abstraction in a case of $\Delta = 2$, and $\texttt{buf} = 6$.
Unlike Figure~\ref{fig:incMsgAge}, there exists no in-transit message that is 3 time-units old in Figure~\ref{fig:incMsgAge-abstraction}(a).
Moreover, the message buffer in Figure~\ref{fig:incMsgAge-abstraction}(a) has a size of 3.
In addition, $\texttt{buf}^{\prime}$ has only one in-transit message that is 2 time-units old.
We have $\texttt{buf}^{\prime} = 4$ in this case.
Figure~\ref{fig:incMsgAge-abstraction}(b) demonstrates another case of $\Delta = 2$, $\texttt{buf} = 5$, and $\texttt{buf}^{\prime} = 6$.

\begin{figure}[tb]
	\centering
	\begin{tikzpicture}
	\matrix (array) [matrix of nodes,
	nodes={draw, minimum size=4mm,anchor=center},
	nodes in empty cells, minimum height = 3mm,
	row 1/.style={nodes={draw=none, fill=none}}]
	{
		0 & 1 & 2   \\
		&   &       \\
	};
	\node[minimum size=2mm] at (array-2-3) (box) {\includegraphics[width=6mm]{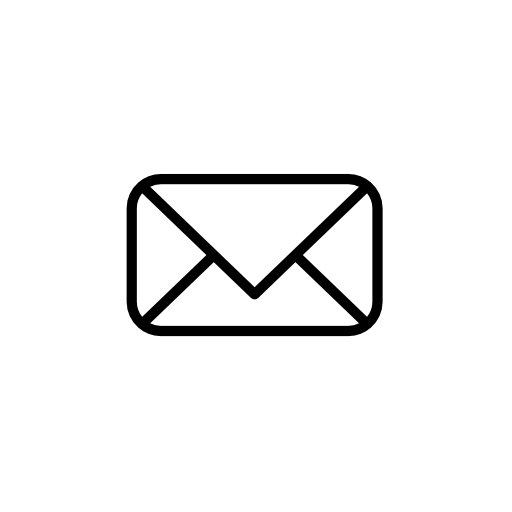}};
	\node[minimum size=2mm] at (array-2-2) (box) {\includegraphics[width=6mm]{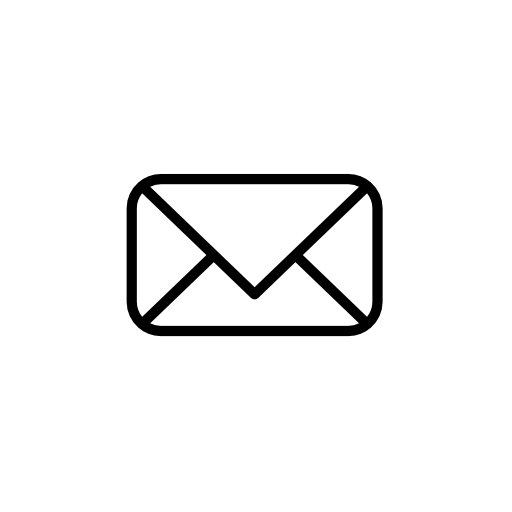}};
	\draw (array-1-1.east)++(0:-5mm) node [left]{Age indexes};
	\draw (array-2-1.east)++(0:-5mm) node [left]{Messages in \texttt{buf}};
	
	\draw (array-2-1.east)++(0:-35mm) node [left]{(a)};
	\end{tikzpicture}
	\hspace{3mm}
	\begin{tikzpicture}
	\matrix (array) [matrix of nodes,
	nodes={draw, minimum size=4mm,anchor=center},
	nodes in empty cells, minimum height = 3mm,
	row 1/.style={nodes={draw=none, fill=none}}]
	{
		0 & 1 & 2    \\
		&   &       \\
	};
	\draw (array-1-1.east)++(0:-5mm) node [left]{Age indexes};
	\draw (array-2-1.east)++(0:-5mm) node [left]{Messages in $\texttt{buf}^{\prime}$};
	\node[minimum size=2mm] at (array-2-3) (box) {\includegraphics[width=6mm]{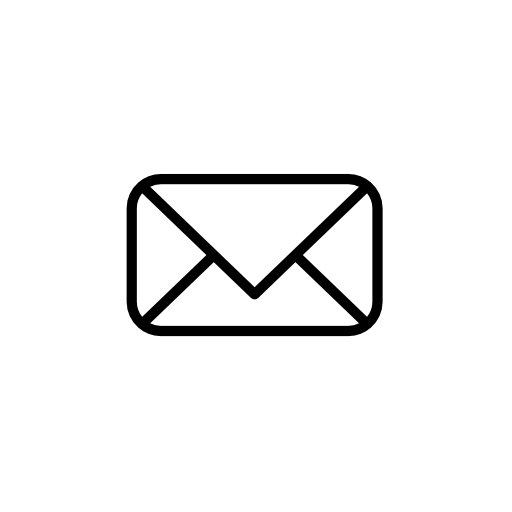}};
	\end{tikzpicture}		
	\begin{tikzpicture}
	\draw[line cap=rect] (0,0) -- (0.85\linewidth-\pgflinewidth,0);
	\end{tikzpicture}
	\begin{tikzpicture}
	\matrix (array) [matrix of nodes,
	nodes={draw, minimum size=4mm,anchor=center},
	nodes in empty cells, minimum height = 3mm,
	row 1/.style={nodes={draw=none, fill=none}}]
	{
		0 & 1 & 2  \\
		&   &       \\
	};
	\node[minimum size=2mm] at (array-2-3) (box) {\includegraphics[width=6mm]{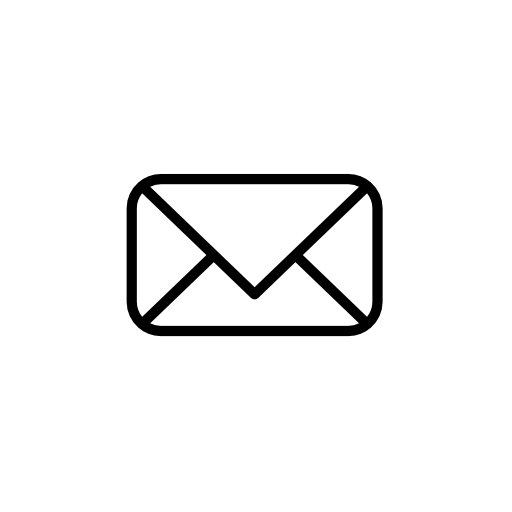}};
	\node[minimum size=2mm] at (array-2-1) (box) {\includegraphics[width=6mm]{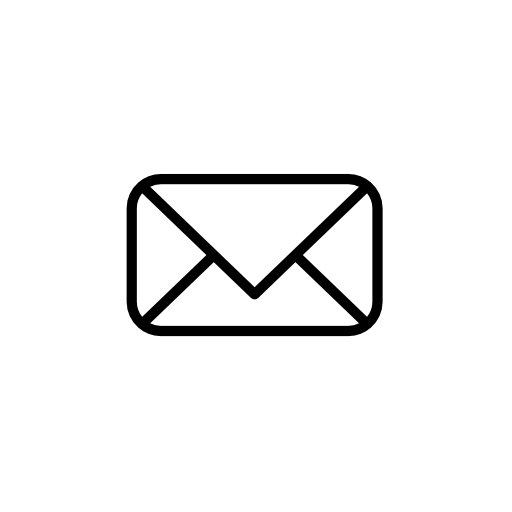}};
	\draw (array-1-1.east)++(0:-5mm) node [left]{Age indexes};
	\draw (array-2-1.east)++(0:-5mm) node [left]{Messages in \texttt{buf}};
	\draw (array-2-1.east)++(0:-35mm) node [left]{(b)};
	\end{tikzpicture}
	\hspace{3mm}
	\begin{tikzpicture}
	\matrix (array) [matrix of nodes,
	nodes={draw, minimum size=4mm,anchor=center},
	nodes in empty cells, minimum height = 3mm,
	row 1/.style={nodes={draw=none, fill=none}}]
	{
		0 & 1 & 2   \\
		&   &      \\
	};
	\draw (array-1-1.east)++(0:-5mm) node [left]{Age indexes};
	\draw (array-2-1.east)++(0:-5mm) node [left]{Messages in $\texttt{buf}^{\prime}$};
	\node[minimum size=2mm] at (array-2-3) (box) {\includegraphics[width=6mm]{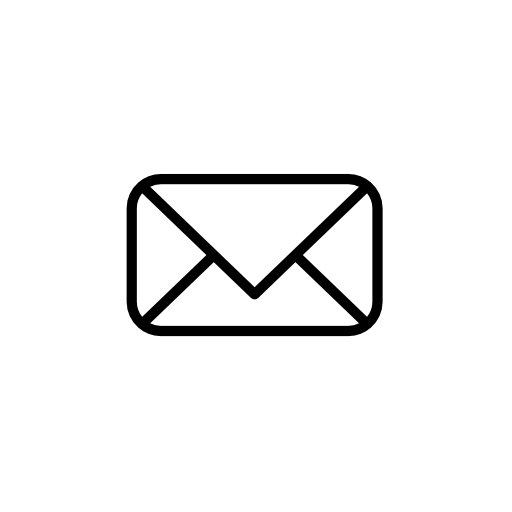}};
	\node[minimum size=2mm] at (array-2-2) (box) {\includegraphics[width=6mm]{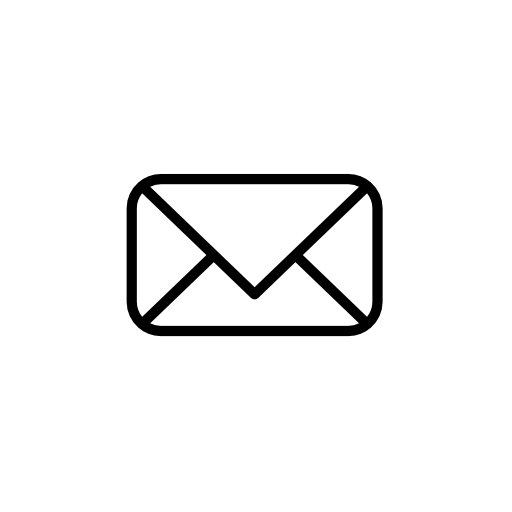}};	
	\end{tikzpicture}
	\caption{The increase of message ages with the abstraction of old messages. In the case (a), we have $\Delta = 2$, $\texttt{buf} = 6$, and $\texttt{buf}^{\prime} = 4$. In the case (b), we have $\Delta = 2$, $\texttt{buf} = 5$, and $\texttt{buf}^{\prime} = 6.$} \label{fig:incMsgAge-abstraction}
\end{figure}
%\end{center}

Formally, Figure~\ref{fig:incMsgAge2} presents the pseudo code of the increase of message ages that is encoded with a counter \texttt{buf}, and the abstraction of old messages.
There are three cases.
In the first case (Line 1), there exist no old messages in $\texttt{buf}$, and we simply set $\texttt{buf}^{\prime} = \texttt{buf} \times 2$.
In other cases (Lines 3 and 4), $\texttt{buf}$ contains an old message.
Figure~\ref{fig:incMsgAge-abstraction}(a) demonstrates the second case (Line 3).
We subtract $2^{\Delta + 1}$ to remove an old message with $\Delta + 1$ time-units old from the buffer.
Figure~\ref{fig:incMsgAge-abstraction}(b) demonstrates the third case (Line 4).
In the third case, we also need to remove an old message with $\Delta + 1$ time-units old from the buffer.
Moreover, we need to put an old message with $\Delta$ time-units old to the buffer by adding $2^{\Delta}$.

\begin{figure}[tb]
	\centering
	\begin{minipage}{.75\linewidth}
		\begin{algorithmic}[1] 
			%%%%%%%%%%%%%%%%%%%%%%%%%%%%%%%%%%%%%%%%%%%%%%%%%%%%%%%%%%%%%%%%%%%%%%%%%%%%%%%%	
			\If{$\texttt{buf} < 2^{\Delta}$}  $\texttt{buf}^{\prime} \gets \texttt{buf} \times 2$
			\Else 
			\If{$\texttt{buf} \geq 2^{\Delta} + 2^{\Delta - 1}$} $\texttt{buf}^{\prime} \gets \texttt{buf} \times 2 - 2^{\Delta + 1}$
			\Else $\texttt{ buf}^{\prime} \gets \texttt{buf} \times 2 - 2^{\Delta + 1} + 2^{\Delta}$
			\EndIf
			\EndIf
		\end{algorithmic}
	\caption{Encoding the increase of message ages with a counter \texttt{buf}, and the abstraction of old messages.} 
		\label{fig:incMsgAge2}
	\end{minipage}	
\end{figure}

Now we discuss how to integrate the abstraction of old messages in the encoding of the message buffer with a predicate.
Formulas~\ref{eq:inc1}--\ref{eq:inc4} present the relation between $\existsMsgOfAge$ and $\existsMsgOfAge^{\prime}$ when message ages are increased by 1, and this abstraction is applied.
Formula~\ref{eq:inc1} ensures that no fresh message will be added to $\existsMsgOfAge^{\prime}$.
Formula~\ref{eq:inc2} ensures that the age of every message that is until $(\Delta - 2)$ time-units old will be increased by 1.
Formulas~\ref{eq:inc3}--\ref{eq:inc4} are introduced by this abstraction.
Formula~\ref{eq:inc3} implies that if there exists an old message or a message with $(\Delta - 1)$ time-units old in $\existsMsgOfAge{}$, there will be an old message that is $\Delta$ time-units old in $\existsMsgOfAge^{\prime}$.
Formula~\ref{eq:inc4} ensures that there exists no message that is older than $\Delta$ time-units old.
\begin{align}
& \lnot \existsMsgOfAge^{\prime}(0)  \label{eq:inc1} \\
& \forall x \in \mathbb{N} \qdot ( 0 \leq x \leq \Delta - 2) \nonumber \\
& \qquad \qquad \quad \Rightarrow \existsMsgOfAge^{\prime}(x + 1) = \existsMsgOfAge(x) \label{eq:inc2} \\
& \existsMsgOfAge^{\prime}(\Delta) = \existsMsgOfAge(\Delta) \lor \existsMsgOfAge(\Delta - 1) \label{eq:inc3} \\
& \forall x \in \mathbb{N} \qdot x > \Delta  \Rightarrow \existsMsgOfAge^{\prime}(x) = \False \label{eq:inc4} 
\end{align}

\subsection{Encoding the Relative Speed of Processes}

Recall that we focus on the case of unknown bounds $\Delta$ and $\Phi$. 
In this case, every correct process must take at least one step in every contiguous time interval containing $\Phi$ time-units~\cite{DLS88}.

To maintain this constraint on executions generated by the verification tools, we introduced two additional control variables \sTimer{} and \rTimer{} for the sender and the receiver, respectively.
These variables work as timers to keep track of how long a process has not taken a step, and when a process can take a step.
Since these timers play similar roles, we here focus on $\rTimer{}$.
In our encoding, only the global system can update $\rTimer{}$.
To schedule the receiver, the global systems non-deterministically executes one of two actions in the sub-round Schedule: (i) resets $\rTimer{}$ to 0, and (ii)~if $\rTimer{} < \Phi$, increases $\rTimer{}$ by~1.
In other sub-rounds, the value of $\rTimer{}$ is unchanged.
Moreover, the receiver must take a step whenever $\rTimer{} = 0$.

\section{Reduce Liveness Properties to Safety Properties}
\label{sec:forte21-reduction}
To verify the liveness properties Eventually Strong Accuracy and Strong Completeness with IVy, we first need to reduce them to safety properties.
Intuitively, these liveness properties are bounded; therefore, they become safety ones.
In the following, we explain how to do that.

\subsection{Eventually Strong Accuracy} 

By cutoffs discussed in Section~\ref{sec:forte21-cutoff-res}, it is sufficient to verify Eventually Strong Accuracy
on the Chandra and Toueg failure detector by checking the following property on instances with 2 processes.
\begin{align}
\F \G ((\mathit{Correct}(1) \land \mathit{Correct}(2)) & \Rightarrow \lnot \mathit{Suspected}(2, 1)) \label{eq:esa1} 
\end{align}
where process 1 is the sender and process 2 is the receiver.  

In the following, we present how to reduce Formula~\ref{eq:esa1} to a safety property. 
Our reduction is based on the following observations:
\begin{enumerate}
\item Fairness (Line 5 in Algorithm~\ref{algo:detector1}): correct processes send ``alive'' infinitely often.

\item The reliable communication (Constraint~\ref{TC1}): 
Let $\mathit{rcv\_msg\_from}(2, 1)$ be a predicate that refers to whether process $2$ receives a message from process $1$. 
If processes $p$ and $q$ are always correct, then it holds $\G \F \mathit{rcv\_msg\_from}(2, 1)$.

\item Transition invariant: Let 
$\psi_1(2, 1)$ be a predicate such that
\[ 
\psi_1(2, 1) \triangleq  \mathit{rcv\_msg\_from}(2, 1)  
\land \mathit{Correct}(1) 
\land \mathit{Correct}(2) 
\land \mathit{Suspected}(2,1)
\]
Then, the following property is a transition invariant.
\begin{align}
\G ( \psi_1(p, q) \Rightarrow 
%\timeout^{\prime}[2, 1] 
\mathit{timeout}^{\prime}[2, 1]
= \mathit{timeout}[2, 1] + 1)
\end{align}

\end{enumerate}

Points 1--3 implies that if $\mathit{timeout}[2, 1]$ is always less than some constant in an arbitrary execution path, then Formula~\ref{eq:esa1} holds in this path.

Now we discuss why $\mathit{timeout}[2, 1]$ does not keep increasing forever if processes 1 and 2 are correct.
To that end, we find a specific guard $g$ for $\mathit{timeout}[2, 1]$ such that if $\mathit{timeout}[2, 1] \geq g$~\footnote{As the default-value of $\mathit{timeout}[2, 1]$ is a parameter, $\mathit{timeout}[2, 1]$ might be greater than $g$ from the initialization.} and the sender is correct, then the receiver waits for the sender in less than $g$ time-units.
Moreover, the value of $g$ depends only on the values of $\Delta$ and $\Phi$.
Hence, it is sufficient to verify Formula~\ref{eq:esa1} by checking Formula~\ref{eq:esa2}.
\begin{align}
\G \big(\timeout[2, 1] \geq g \Rightarrow ((\mathit{Correct}(1) \land \mathit{Correct}(2)) & \Rightarrow \lnot \mathit{Suspected}(2, 1)) \big) \label{eq:esa2} 
\end{align}

\subsection{Strong Completeness} 

By cutoffs discussed in Section~\ref{sec:forte21-cutoff-res}, it is sufficient to verify Strong Completeness
on the Chandra and Toueg failure detector by checking the following property on instances with 2 processes.
\begin{align}
\F \G ((\lnot \mathit{Correct}(1) \land \mathit{Correct}(2)) & \Rightarrow \mathit{Suspected}(2, 1)) \label{eq:sc1}
\end{align}

Notice that in partial synchrony, every sent message is eventually delivered. 
Hence, after the sender crashes, the receiver eventually receives nothing from the sender.
To reduce Formula~\ref{eq:sc1} to a safety property, we first introduced a ghost variable \howLongFromSCrash{} to measure for how long the sender has crashed. 
\howLongFromSCrash{} is set to 0 when the sender crashes.
After that, \howLongFromSCrash{} is increased by 1 in every global step if the receiver has not suspected the crashed sender.
Let $\psi_2(2, 1)$ denote the constraint:
$ 
\psi_2(2, 1) \triangleq 
\lnot \mathit{Correct}(1) 
\land \mathit{Correct}(2) 
\land \lnot \mathit{Suspected}(2,1)
$.
Then, the following property is a transition invariant.
\begin{align}
\G ( \psi_2(p, q) \Rightarrow \howLongFromSCrash^{\prime} = \howLongFromSCrash + 1) \label{eq:psi2}
\end{align}

By Formula~\ref{eq:psi2}, if  \howLongFromSCrash{}  is always less than some constant in an arbitrary execution path, then Formula~\ref{eq:sc1} holds in this path. 

Now we show that \howLongFromSCrash{} cannot keep increasing forever.
To that end, we find a specific guard $g^{\prime} > 0$ for \howLongFromSCrash{}  such that if $\howLongFromSCrash{} = g^{\prime}$, then the receiver   suspects the sender.
It implies that \howLongFromSCrash{} is unchanged.
Moreover, the value of $g^{\prime}$ depends only on the values of $\Delta$ and $\Phi$.
Hence, it is sufficient to verify Formula~\ref{eq:sc1} by checking Formula~\ref{eq:sc2}.
\begin{align}
\G \big(\howLongFromSCrash = g^{\prime} \Rightarrow ((\lnot \mathit{Correct}(1) \land \mathit{Correct}(2)) & \Rightarrow  \mathit{Suspected}(2, 1)) \big) \label{eq:sc2}
\end{align}

\section{Experiments for Small $\Delta$ and $\Phi$} \label{sec:forte21-experiments-fixed-paras}
In this section, we describe our experiments with~\tlap{} and \fast{}.
We ran the following experiments on a virtual machine with Core i7-6600U CPU and 8GB DDR4. 
Our specifications can be found at~\cite{fdSpec_Zenedo}.

\subsection{Model Checkers for~\tlap{}: TLC and~\apalache{}}
We first use~\tlap{}~\cite{lamport2002specifying} to specify the failure detector with both encoding techniques for the message buffer, and the abstraction in Section~\ref{sec:forte21-encoding}.
Then, we use the model checker TLC in the~\tlap{} Toolbox version 1.7.1~\cite{yu1999model,toolbox} and the model checker \apalache{} version 0.15.0~\cite{konnov2019tla,apalache19} to verify instances with
fixed bounds $\Delta$ and $\Phi$, and the GST $T_0 = 1$.
This approach helps us to search constraints in inductive invariants in case of fixed parameters.
The main reason is that counterexamples and inductive invariants in case of fixed parameters, e.g., $\Delta \leq 1$ and $\Phi \leq 1$, are simpler than in case of arbitrary parameters.
Hence, if a counterexample is found, we can quickly analyze it, and change constraints in an inductive invariant candidate.	
We apply the counterexample-guided approach to find inductive strengthenings.
After obtaining inductive invariants in small cases, we can generalize them for cases of arbitrary bounds, and check with theorem provers, e.g., IVy (Section~\ref{sec:forte21-experiments-unknown-paras}).

\tlap{} offers a rich syntax for sets, functions, tuples, records, sequences, and control structures~\cite{lamport2002specifying}.
Hence, it is straightforward to apply the encoding techniques and the abstraction presented in Section~\ref{sec:forte21-encoding} in~\tlap{}.
For example, Figure~\ref{fig:send-tla} represents a \tlap{} action \emph{SSnd} for sending a new message in case of $\Delta > 0$.
Variables $\mathit{ePC}$ and $\mathit{sPC}$ are program counters for the environment and the sender, respectively.
Line 1 is a precondition, and refers to that the environment is in subround Send.
Lines 2--3 say that if the sender is active in subround Send, the counter $\mathit{buf}^{\prime}$ is increased by 1.
Otherwise, two counters $\mathit{buf}$ and $\mathit{buf}^{\prime}$ are the same (Line 4).
Line 5 implies that the environment is still in the subround Send, but it is now the receiver\textquotesingle{}s turn.
Line 6 guarantees that other variables are unchanged in this action.

Figure~\ref{fig:next-tla} represents the next--state relation in \tlap{}.
Line 1 describes actions in sub-round Schedule. 
The environment schedules the Sender, schedules the Receiver, and then increases message ages.
Lines 2, 3, and 4 describes actions in sub-rounds Send, Receive, and  Computation, respectively.  
The program counter $ePC$ of the environment is used to ensure that every action is repeated periodically and in order.

Figure~\ref{fig:rsched-tla} represents how the environment schedules the Receiver in~\tlap{}.
Line 1 says that the current step is to schedule the Receiver, and Line 2 refers to the next action that is to increase message ages.
Line 3 non-deterministically sets the Receiver active in the current global step.
Lines 4--6 are to update the program counter \textit{rPC} of the Receiver.
The environment schedules the Sender, schedules the Receiver, and then increases message ages.
Lines 7--8 non-deterministically sets the Receiver inactive in the current global step if the Receiver is not frozen in the last $\textit{Phi} - 1$ global steps.
Line 9 is to keep other variables unchanged.

	\begin{figure}[tb]
	\begin{minipage}{\linewidth}
		\input{./forte21/send-tla}
		\caption{Sending a new message in~\tlap{} in case of $\Delta > 0$} \label{fig:send-tla}
	\end{minipage}
	\end{figure}

	\begin{figure}[tb]
  \begin{minipage}{\linewidth}
    \input{./forte21/next-tla}
    \caption{The Next predicate for the next-state relation in~\tlap{}} \label{fig:next-tla}
  \end{minipage}
\end{figure}

	\begin{figure}[tb]
  \begin{minipage}{\linewidth}
    \input{./forte21/rsched-tla}
    \caption{The RSched predicate for scheduling the Receiver in~\tlap{}} \label{fig:rsched-tla}
  \end{minipage}
\end{figure}

Now we present the experiments with TLC and~\apalache{}. 
We used these tools to verify (i) the safety property Strong Accuracy, and (ii) an inductive invariant for Strong Accuracy, and (iii) an inductive invariant for a safety property reduced from the liveness property Strong Completeness in case of fixed bounds, and GST = 1 (initial stabilization).
The structure of the inductive invariants verified here are very close to one in case of arbitrary bounds $\Delta$ and $\Phi$.
While all parameters are assigned specific values in the inductive invariants of small instances, they have arbitrary values in the case of arbitrary bounds.

Table~\ref{tab:exp-sa-inv} shows the results in verification of Strong Accuracy in case of the initial stabilization, and fixed bounds $\Delta$ and $\Phi$.
Table~\ref{tab:exp-sa-inv} shows the experiments with the three tools TLC, \apalache{}, and \fast{}.
The column ``\#states'' shows the number of distinct states explored by TLC.
The column ``\#depth'' shows the maximum execution length reached by TLC and~\apalache{}.
The column ``buf'' shows how to encode the message buffer.
The column ``LOC'' shows the number of lines in the specification of the system behaviors (without comments).
The symbol ``-'' (minus) refers to that the experiments are intentionally missing since~\fast{} does not support the encoding of the message buffer with a predicate.
The abbreviation ``pred'' refers to the encoding of the message buffer with a predicate.
The abbreviation ``cntr'' refers to the encoding of the message buffer with a counter.
The abbreviation ``TO'' means a timeout of 6 hours.
In these experiments, we initially set $\timeout{} = 6 \times \Phi + \Delta$, and Strong Accuracy is satisfied.
The experiments show that TLC finishes its tasks faster than the others, and  \apalache{} prefers the encoding of the message buffer with a predicate.

Table~\ref{tab:exp-sa-inv-bug} summarizes the results in verification of Strong Accuracy with the tools TLC, \apalache{}, and \fast{} in case of the initial stabilization, and small bounds $\Delta$ and $\Phi$, and initially $\timeout{} = \Delta + 1$.
Since \timeout{} is initialized with a too small value, there exists a case in which sent messages are delivered after the timeout expires.
The tools reported an error execution where Strong Accuracy is violated.
In these experiments, \apalache{} is the winner.
The abbreviation ``TO'' means a timeout of 6 hours.
The meaning of other columns and abbreviations is the same as in Table~\ref{tab:exp-sa-inv}.

\begin{table}[tb]
	\centering
	\caption{Showing Strong Accuracy for fixed parameters.} \label{tab:exp-sa-inv} 
	\begin{tabular}{>{\centering\arraybackslash}p{0.4cm} |  >{\centering\arraybackslash}p{0.4cm} | 
			>{\centering\arraybackslash}p{0.4cm} |
			>{\centering\arraybackslash}p{0.8cm} |
			>{\raggedleft\arraybackslash}p{0.7cm}
			>{\raggedleft\arraybackslash}p{1.4cm}
			>{\raggedleft\arraybackslash}p{1.1cm}
			%			>{\raggedleft\arraybackslash}p{0.8cm}	
			>{\raggedleft\arraybackslash}p{1.1cm} |
			>{\raggedleft\arraybackslash}p{0.7cm}
			>{\raggedleft\arraybackslash}p{1.1cm}
			|
			>{\raggedleft\arraybackslash}p{0.7cm} 
			>{\raggedleft\arraybackslash}p{1cm}} 
		\multirow{2}{*}{\#} & \multirow{2}{*}{$\Delta$} &  \multirow{2}{*}{$\Phi$} & \multirow{2}{*}{buf} & \multicolumn{4}{c|}{TLC} & \multicolumn{2}{c|}{\apalache{}} & \multicolumn{2}{c}{\fast} \\   
		\cline{5-12}
		& &  & & time & \#states & depth & LOC & time & depth & time & LOC \\
		\hline
		1  & \multirow{2}{*}{2} &  \multirow{2}{*}{4} & pred & 3s & 10.2K & 176 & 190 & 8m & 176 & - & -  \\
		2  &  &  & cntr & 3s & 10.2K & 176 & 266 & 9m & 176 & 16m & 387 \\
		\hline
		3  & \multirow{2}{*}{4} &  \multirow{2}{*}{4} & pred & 3s & 16.6K & 183  & 190 & 12m & 183 &  - & -  \\
		4  &  &  & cntr & 3s & 16.6K & 183 & 487 & 35m & 183 &  TO &  2103 \\
		\hline
		5  & \multirow{2}{*}{4} &  \multirow{2}{*}{5} & pred & 3s & 44.7K & 267 & 190 & TO & 222 & - & - \\
		6  &  &  & cntr & 3s & 44.7K &  267 & 487 & TO & 223 & TO & 2103 \\
	\end{tabular}
\end{table}

\begin{table}[tb]
	\centering
	\caption{Violating Strong Accuracy for fixed parameters.} \label{tab:exp-sa-inv-bug}
	\begin{tabular}{>{\centering\arraybackslash}p{0.4cm} |  >{\centering\arraybackslash}p{0.4cm} | 
			>{\centering\arraybackslash}p{0.4cm} |
			>{\centering\arraybackslash}p{0.8cm} |
			>{\raggedleft\arraybackslash}p{0.7cm}
			>{\raggedleft\arraybackslash}p{1.4cm}
			>{\raggedleft\arraybackslash}p{1.1cm} |
			>{\raggedleft\arraybackslash}p{1.3cm}
			>{\raggedleft\arraybackslash}p{1.1cm} |
			>{\raggedleft\arraybackslash}p{1cm}} 
		\multirow{2}{*}{\#} & \multirow{2}{*}{$\Delta$} & \multirow{2}{*}{$\Phi$} & \multirow{2}{*}{buf} & \multicolumn{3}{c|}{TLC} & \multicolumn{2}{c|}{\apalache{}} & \multicolumn{1}{c}{ \fast} \\   
		\cline{5-10}
		& & & & time & \#states & depth & time & depth  & time  \\
		\hline
		1  & \multirow{2}{*}{2} &  \multirow{2}{*}{4} & pred & 1s & 840 & 43  & 11s & 42 & -  \\
		2  &  & & cntr & 1s & 945 & 43 & 12s & 42 & 10m  \\
		\hline
		3  & \multirow{2}{*}{4} &  \multirow{2}{*}{4} & pred & 2s & 1.3K & 48 & 15s & 42  & -  \\
		4  &  &  & cntr & 2s & 2.4K & 56 & 16s & 42 & TO                      \\
		\hline
		5  & 20 & 20 & pred & TO & 22.1K & 77 & 1h15m & 168  & -  \\
	\end{tabular}
\end{table}

\begin{table}[tb]
	\centering
	\caption{Proving inductive invariants  with TLC and~\apalache{}.} \label{tab:exp-sa-iinv}
	\begin{tabular}{>{\centering\arraybackslash}p{0.4cm} |  >{\centering\arraybackslash}p{0.4cm} | 
			>{\centering\arraybackslash}p{0.4cm} |
			>{\centering\arraybackslash}p{3.3cm} |
			>{\raggedleft\arraybackslash}p{0.7cm}
			>{\raggedleft\arraybackslash}p{1.4cm} |
			>{\centering\arraybackslash}p{2.4cm}} 
		\multirow{2}{*}{\#} & \multirow{2}{*}{$\Delta$} & \multirow{2}{*}{$\Phi$} & \multirow{2}{*}{Property} & \multicolumn{2}{c|}{TLC} & \apalache{}   \\  
		\cline{5-7}
		&  &  &  & time & \#states & \multicolumn{1}{c}{time}   \\  
		\hline
		1  & 4   &  40 & Strong Accuracy & 33m & 347.3M & 12s \\    
		\hline
		2  &  4   &  10 & Strong Completeness & 44m & 13.4M & 17s
	\end{tabular}
\end{table}

Table~\ref{tab:exp-sa-iinv} shows the results in verification of inductive invariants for Strong Accuracy and Strong Completeness with TLC and \apalache{} in case of the initial stabilization, and slightly larger but fixed bounds $\Delta$ and $\Phi$, e.g., $\Delta = 20$ and $\Phi = 20$.
The message buffer was encoded with a predicate in these experiments.
In these experiments, inductive invariants hold, and \apalache{} is faster than TLC in verifying them.
In our experiment, we applied the counterexample-guided approach to manually find inductive strengtheningss. 

As one sees from the tables,~\apalache{} is fast at proving inductive invariants, and at finding a counterexample when a desired safety property is violated.
TLC is a better option in cases where a safety property is satisfied.

In order to prove correctness of the failure detector in cases where parameters $\Delta$ and $\Phi$ are arbitrary, the user can use the interactive theorem prover \tlap{} Proof System (\tlaps{})~\cite{chaudhuri2010tla+}.
A shortcoming of \tlaps{} is that it does not provide a counterexample when an inductive invariant candidate is violated.
Moreover, proving the failure detector with \tlaps{} requires more human effort than with IVy.
Therefore, we provide IVy proofs in Section~\ref{sec:forte21-experiments-unknown-paras}.

\subsection{\fast{}}

A shortcoming of the model checkers TLC and~\apalache{} is that parameters $\Delta$ and $\Phi$ must be fixed before running these tools.
\fast{} is a tool designed to reason about safety properties of counter systems, i.e. automata extended with unbounded integer variables~\cite{bardin2006fast}.
If $\Delta$ is fixed, and the message buffer is encoded with a counter, the failure detector becomes a counter system.
%In our experiments, 
We specified the failure detector in~\fast{}, and made experiments with different parameter values to understand the limit of~\fast{}: (i) the initial stabilization, and small bounds $\Delta$ and $\Phi$, and (ii) the initial stabilization, fixed $\Delta$, but unknown $\Phi$.

Figure~\ref{fig:send-fast} represents a \fast{} transition for sending a new message in case of $\Delta > 0$.
Line~2 describes the (symbolic) source state of the transition, and region \texttt{incMsgAge} is a set of configurations in the failure detector that is reachable from a transition for increasing message ages.
Line 3 mentions the (symbolic) destination state of the transition, and region \texttt{sSnd} is a set of configurations in the failure detector that is reachable from a transition named ``SSnd\_Active'' for sending a new message.
Line 4 represents the guard of this transition.
Line 5 is an action.
Every unprimed variable that is not written in Line 5 is unchanged.

\begin{figure}[tb]
\centering
	\begin{minipage}{.7\linewidth}

	\begin{algorithmic}[1] 
		%%%%%%%%%%%%%%%%%%%%%%%%%%%%%%%%%%%%%%%%%%%%%%%%%%%%%%%%%%%%%%%%%%%%%%%%%%%%%%%%		 
		\State \textbf{transition} SSnd\_Active  $\coloneq \{$
		\State $\qquad $\textbf{from} $\coloneq$ incMsgAge;
		\State $\qquad $\textbf{to} $\coloneq$ ssnd;
		\State $\qquad $\textbf{guard} $\coloneq$ sTimer = 0;
		\State $\qquad $\textbf{action} := buf$^{\prime}$ = buf + 1; $\}$;
	\end{algorithmic}
	\caption{Sending a new message in~\fast{} in case of $\Delta > 0$} 
	\label{fig:send-fast}
	\end{minipage}
\end{figure}
The input language of~\fast{} is based on Presburger arithmetics
for both system and properties specification.
Hence, we cannot apply the encoding of the message buffer with a predicate in~\fast{}.

Tables~\ref{tab:exp-sa-inv} and~\ref{tab:exp-sa-inv-bug} described in the previous subsection summarize the experiments with~\fast{}, and other tools where all parameters are fixed.
Moreover, we ran~\fast{} to verify Strong Accuracy in case of the initial stabilization, $\Delta \leq 4$, and arbitrary $\Phi$.
\fast{} is a semi-decision procedure; therefore, it does
not terminate on some inputs. 
Unfortunately, \fast{} could not prove Strong Accuracy in case of arbitrary $\Phi$, and crashed after 30 minutes.

\section{IVy Proofs for Parametric $\Delta$ and $\Phi$} \label{sec:forte21-experiments-unknown-paras}
While TLC, \apalache{}, and~\fast{} can automatically verify some instances of the failure detector with fixed parameters, these tools cannot handle cases with unknown bounds $\Delta$ and $\Phi$.
To overcome this problem, we specify and prove correctness of the failure detector with the interactive theorem prover IVy~\cite{mcmillan2020ivy}.
In the following, we first discuss the encoding of the failure detector, and then present the experiments with IVy.

The encoding of the message buffer with a counter requires that bound $\Delta$ is fixed.
We here focus on cases where bound $\Delta$ is unknown.
Hence, we encode the message buffer with a predicate in our IVy specifications, 

In IVy, we declare
$\texttt{relation } \existsMsgOfAge(X \colon \tnumber)$. 
Type \texttt{num} is interpreted as integers.
Since IVy does not support primed variables, we need an additional relation $\tmpExistsMsgOfAge(X : \tnumber)$.  
Intuitively, we first compute and store the value of~\existsMsgOfAge{} in the next state in~\tmpExistsMsgOfAge{}, then copy the value of~\tmpExistsMsgOfAge{} back to~\existsMsgOfAge{}.
We do not consider the requirement of \tmpExistsMsgOfAge{} as a shortcoming of IVy since it is still straightforward to transform the ideas in Section~\ref{sec:forte21-encoding} to IVy.

Figure~\ref{fig:send-ivy} represents how to add a fresh message in the message buffer in IVy.
Line~1 means that $\tmpExistsMsgOfAge$ is assigned an arbitrary value. 
Line~2 guarantees the appearance of a fresh message.
Line~3 ensures that every in-transit message in $\existsMsgOfAge$ is preserved in $\tmpExistsMsgOfAge$.
Line~4 copies the value of $\tmpExistsMsgOfAge$ back to $\existsMsgOfAge$.

\begin{algorithm}[!htb]
\small
\begin{algorithmic}[1] 
		\State $\tmpExistsMsgOfAge(X) \coloneq *;$
		\State \textbf{assume} $\tmpExistsMsgOfAge(0);$
		\State \textbf{assume forall} $X \colon \tnumber \qdot 0 < X \rightarrow \existsMsgOfAge(X) = \tmpExistsMsgOfAge(X);$
		\State $\existsMsgOfAge(X) \coloneq \tmpExistsMsgOfAge(X);$
\end{algorithmic}
\caption{Adding a fresh message in IVy} 
\label{fig:send-ivy}
\end{algorithm}

Importantly, our specifications are not in decidable theories supported by IVy.
In Formula~\ref{eq:inc2}, the interpreted function $``+"$ (addition) is applied to a universally quantified variable $x$.

The standard way to check whether a safety property \textit{Prop} holds in an IVy specification is to find an inductive invariant \textit{IndInv} with \textit{Prop}, and to (interactively) prove that \textit{Indinv} holds in the specification. 
To verify the liveness properties Eventually Strong Accuracy, and Strong Completeness, we reduced them into safety properties by applying a reduction technique in Section~\ref{sec:forte21-reduction}, and found inductive invariants containing the resulting safety properties reduced from the liveness properties.
These inductive invariants are the generalization of the inductive invariants in case of fixed parameters that were found in the previous experiments.

	\begin{table}[!htb]
		\centering
	\caption{Proving inductive invariants with IVy for arbitrary $\Delta$ and $\Phi$.} \label{tab:exp-iinv-ivy}
	\begin{tabular}{>{\centering\arraybackslash}p{0.3cm} |  			>{\centering\arraybackslash}p{3.7cm} |
			>{\raggedright\arraybackslash}p{2.2cm} 
			>{\raggedleft\arraybackslash}p{0.9cm}
			>{\raggedleft\arraybackslash}p{1cm}
			>{\raggedleft\arraybackslash}p{1.4cm}
			>{\raggedleft\arraybackslash}p{2.6cm}} 
		\multirow{2}{*}{\#} & \multirow{2}{*}{Property} & \multirow{2}{*}{$\;\;\timeout{}_{\text{init}}$} & \multirow{2}{*}{time} & \multirow{2}{*}{LOC} & \multirow{2}{*}{$\#\text{line}_I$} & \#strengthening  \\ 
		& &  & & & & steps \\                 
		\hline
		1  & Strong Accuracy & $= 6 \times \Phi + \Delta$ & 4s & 183 & 30 & 0 \\    
		\hline
		\multirow{2}{*}{2}  & \multirow{2}{*}{\shortstack[c]{Eventually \\ Strong Accuracy}} & \multirow{2}{*}{$= \star$} & \multirow{2}{*}{4s} & \multirow{2}{*}{186} & \multirow{2}{*}{35} & \multirow{2}{*}{0} \\                
		& &  &  & & &\\                      		
		\hline
		3  & \multirow{3}{*}{Strong Completeness} & $= 6 \times \Phi + \Delta $ & 8s & 203 & 111 & 0 \\                
		4  & & $\geq 6 \times \Phi + \Delta$ &  22s & 207 & 124 & 15 \\                
		5  & & $= \star$ & 44s & 207 & 129 & 0\\                
	\end{tabular}
	\end{table}

Table~\ref{tab:exp-iinv-ivy} shows the experiments on verification of the failure detector with IVy in case of unknown $\Delta$ and $\Phi$.
The symbol $\star$ refers to that the initial value of \timeout{} is arbitrary.
The column ``$\#\text{line}_I$'' shows the number of lines of an inductive invariant, and the column ``$\#\text{strengthening steps}$'' shows the number of lines of strengthening steps that we provided for IVy.
The meaning of other columns is the same as in Table~\ref{tab:exp-sa-inv}.
While our specifications are not in the decidable theories supported in IVy, our experiments show that IVy needs no user-given strengthening steps to prove most of our inductive invariants.
Hence, it took us about 4 weeks to learn IVy from scratch, and to prove these inductive invariants.

The most important thing to prove a property satisfied in an IVy specification is to find an inductive invariant.
Our inductive invariants use non-linear integers, quantifiers, and uninterpreted functions.
(The inductive invariants in Table~\ref{tab:exp-iinv-ivy} are given in the repository~\cite{fdSpec_Zenedo}.)

While IVy supports a liveness-to-safety reduction~\cite{padon2017reducing}, 
this technique is not fully automated, and 
IVy still needs user-guided inductive invariant for reduced safety properties that may be different from those in Table~\ref{tab:exp-iinv-ivy}.
Moreover, IVy has not supported reasoning techniques for clocks. 
Therefore, we did not try the liveness-to-safety reduction of IVy.

It is straightforward to generalize the inductive invariants in Table~\ref{tab:exp-iinv-ivy} for partially synchronous models with known time bounds in~\cite{DLS88,CT96}.
To reason about models with $\text{GST} > 0$, we need to find additional inductive strengthenings because the global system is under asynchrony before GST.
Other partially synchronous models in~\cite{attiya1987achievable} consider additional paramaters, e.g., message order or point-to-point transmission that are out of scope of this paper.

\section{Related work} \label{sec:relate-work}
\subsection{Cutoffs} 
Distributed algorithms are typically parameterized in the number of participants, e.g., two-phase commit protocol~\cite{lampson1979crash} and the Chandra and Toueg failure detector in Section~\ref{sec:preliminaries}.
While the general parameterized verification problem is undecidable~\cite{AK86,suzuki1988proving,2015Bloem}, many distributed algorithms such as mutual exclusion and cache coherence enjoy the cutoff property, 
which reduces the parameterized verification problem to verification of a small
number of instances. 
In a nutshell, a cutoff for a parameterized algorithm $\mathcal{A}$ and a property $\phi$ is a number $k$ such that $\phi$ holds for every instance of $\mathcal{A}$ if and only if $\phi$ holds for instances with $k$ processes~\cite{emerson1995reasoning,2015Bloem}.
In the last decades, researchers have proved the cutoff results for various models of computation: ring-based message-passing systems~\cite{emerson1995reasoning,EmersonK04}, purely disjunctive guards and conjunctive guards~\cite{emerson2000reducing,emerson2003exact},  token-based communication~\cite{clarke2004verification}, and quorum-based algorithms~\cite{maric2017cutoff}.
However, we cannot apply these results to the Chandra and Toueg failure detector because it relies on point-to-point communication and timeouts.
Moreover, distributed algorithms discussed in ~\cite{emerson1995reasoning, emerson2000reducing,emerson2003exact,EmersonK04,clarke2004verification,maric2017cutoff} are not in the symmetric point-to-point class.

\subsection{Formal verification for partial synchrony}

Partial synchrony is a well-known model of computation in distributed computing.
To guarantee liveness properties, many practical protocols, e.g., the failure detector in Section~\ref{sec:preliminaries} and proof-of-stake blockchains~\cite{buchman2018latest,yin2019hotstuff}, assume time constraints under partial synchrony.
That is the existence of bounds $\Delta$ on message delay and $\Phi$ on the relative speed of processes after some time point.

While partial synchrony is important for system designers, it
is challenging for verification.
The mentioned constraint makes partially synchronous algorithms parametric in time bounds. 
Moreover, partially synchronous algorithms are typically parameterized in the number of processes. 

Research papers about
partially synchronous algorithms, including papers about failure
detectors~\cite{larrea1999efficient,aguilera2006consensus,aguilera2008implementing}
contain manual proofs and no formal specifications.
Without these details, proving those distributed algorithms with interactive theorem provers~\cite{cousineau2012tla,mcmillan2020ivy} is impossible.

System designers can use timed automata~\cite{alur1994theory} and parametric verification
frameworks~\cite{larsen1997uppaal,andre2012imitator,lime2009romeo}
to specify and verify timed systems.
In the context of timed systems, we are aware of only one paper about verification of failure detectors~\cite{atif2012formal}.
In this paper, the authors used three tools, namely UPPAAL~\cite{larsen1997uppaal}, mCRL2~\cite{bunte2019mcrl2}, and FDR2~\cite{roscoe2010understanding} to verify small instances of a failure detector based on a logical ring arrangement of 
processes.
Their verification approach required that message buffers were bounded, and had restricted behaviors in the specifications.
Moreover, they did not consider the bound $\Phi$ on the relative speed of processes.
In contrast, there are no restrictions on message buffers, and no ring topology in the Chandra and
Toueg failure detector.

In recent years, automatic parameterized verification techniques~\cite{KLVW17:FMSD,stoilkovska2019verifying,druagoi2020programming} have been introduced for distributed systems, but they are designed for synchronous and/or asynchronous models.
Interactive theorem provers have been used to prove correctness of distributed algorithms recently. 
For example, researchers proved safety of Tendermint consensus with IVy~\cite{ivyTendermint}.

\section{Conclusion} \label{sec:conclusion}
We have presented parameterized and parametric verification of both safety and liveness of the Chandra and Toeug failure detector.
To this end, we first introduce and formalize
the class of symmetric point-to-point algorithms that contains the failure detector.
Second, we show that the symmetric point-to-point algorithms have a cutoff, and the cutoff properties hold in three models of computation: synchrony, asyncrony, and partial synchrony.

Next, we develop the encoding techniques to efficiently specify the failure detector, and to tune our models to the strength of the verification tools: model checkers for~\tlap{} (TLC and~\apalache), counter automata (\fast{}), and the  theorem prover Ivy.
We verify safety in case of fixed parameters by running the tools TLC,~\apalache{}, and~\fast{}.
To cope with cases of arbitrary bounds $\Delta$ and $\Phi$, we reduce liveness properties to \textbf{}safety properties, and proved inductive invariants with desired properties in Ivy.
While our specifications are not in the decidable theories supported in Ivy, our experiments show that Ivy needs no additional user assistance to prove most of our inductive invariants.

Modeling the failure detector in~\tlap{} helped us understand and find inductive invariants in case of fixed parameters.
Their structure is simpler but similar to the structure of parameterized inductive invariants.
We found that the~\tlap{} Toolbox~\cite{kuppe2019tla+} has convenient features, e.g., Profiler and Trace Exploration.
A strong point of Ivy is in producing a counterexample quickly when a property is violated, even if all parameters are arbitrary.
In contrast, \fast{} reports no counterexample in any case.
Hence, debugging in~\fast{} is very challenging.

While our specification describes executions
of the Chandra and Toueg failure detector, we conjecture that many time constraints on network behaviors, correct processes, and failures in our inductive invariants can be reused to
prove other algorithms under partial synchrony.
We also conjecture that correctness of other partially synchronous algorithms may be proven by following the presented methodology.
For future work, we would like to extend the above results for cases where GST is arbitrary. 
It is also interesting to investigate how to 
express discrete partial synchrony
in timed automata~\cite{alur1994theory},
e.g., UPPAAL~\cite{larsen1997uppaal}.

\newpage

\bibliographystyle{alphaurl}

\bibliography{lit}

\appendix

\section{Linear temporal logic (LTL)} \label{sec:ltl}
The syntax of LTL formulae is given by the following syntax~\cite{pnueli1977temporal}:
\[
\phi \Coloneqq \top \mid p \mid \lnot \phi \mid \phi \land \phi \mid \X \phi \mid \phi \U \phi
\]
where

\begin{itemize}
  \item $\top$ stands for true,
  
  \item $p$ ranges over a countable set \emph{AP} of atomic predicates,
  
  \item $\lnot$ and $\land$ are Boolean operators negation and conjunction respectively,
  
  \item $\X$ and $\U$ are the temporal operators next and until respectively.

\end{itemize}

Other Boolean operators $\lor, \Rightarrow$, and $\Leftrightarrow$ are defined in the standard way. We also
use $\bot$ (false), $\F$ (eventually), $\G$ (always) to abbreviate
$\lnot \top$, $\top \U \phi$ , and $\lnot (\top \U \lnot \phi)$ respectively.

Let us consider an LTL formula $\phi$ over propositions in \textit{AP} and a word $w \colon \mathbb{N} \rightarrow 2^{\textit{AP}}$.
We define the relation $w \vDash \phi$ (the word $w$ satisfies the formula $\phi$) inductively as follows.

\begin{itemize}

\item $w \vDash \top$ for all $w$

\item $w \vDash p$ if $p \in w[0]$

\item $w \vDash \lnot \phi$ if not $w \vDash \phi$

\item $w \vDash \phi_1 \land \phi_2$ if $w \vDash \phi_1$ and $w \vDash \phi_2$

\item $w \vDash \X \phi$ if $w[1...] \vDash \phi$ where $w[k...]$ denotes the suffix $w[k]w[k+1]\ldots$ 

\item $w \vDash \phi_1 \U \phi_2$ if there exists $j \geq 0$ such that $w[j...] \vDash \phi_2$ and for all $0 \leq i < j, w[i...] \vDash \phi_1$

\end{itemize}

Two words $w$ and $w^{\prime}$ are stuttering equivalent if there are two infinite sequences of numbers
$0 = i_{0} <i_{1} < i_{2} < \ldots $
and 
$0 = j_{0} < j_{1} < j_{2} <\ldots $
such that for every $\ell \geq 0$, it holds 
$ 
  w_{i_{\ell}}
= w_{i_{\ell}+1}
= \ldots 
= w_{i_{\ell + 1} - 1}
= w_{j_{\ell}}^{\prime}
= w_{j_{\ell}+1}^{\prime}
= \ldots 
= w_{j_{\ell+1}-1}^{\prime}
$
~\cite{clarke2018model}.

In this paper, we deal with properties in LTL\textbackslash{}X (LTL minus the next operator).
Importantly, all LTL  properties that are invariants under stuttering equivalence can be expressed without the
next operator $\X$~\cite{peled1997stutter}.

\end{document}